\pdfoutput=1
\documentclass[pra,aps,superscriptaddress,twocolumn,nofootinbib,longbibliography,a4paper]{revtex4-2}

\RequirePackage[l2tabu, orthodox]{nag}
\usepackage[utf8]{inputenc}
\usepackage[T1]{fontenc}

\usepackage{graphicx,color,amsmath,amsfonts,enumerate,amsthm,amssymb,mathtools,enumitem,thmtools,mathdots,enumitem,bm,soul,bbm,tikz,pgfplots,float,thm-restate,graphicx,hyperref,xcolor,colortbl,overpic,tikz,pgfplots,multirow,booktabs}
\usepackage[capitalise, noabbrev, nameinlink]{cleveref}
\usepackage[normalem]{ulem}
\usepackage[caption=false]{subfig}

\definecolor{darkblue}{RGB}{0,0,128}
\definecolor{darkgreen}{RGB}{0,150,0}
\hypersetup{breaklinks, colorlinks, linkcolor=blue, citecolor=darkgreen, filecolor=red, urlcolor=darkblue, pdftitle={Quantum dichotomies}}

\pgfplotsset{compat=newest}

\newtheorem{lemma}{Lemma}
\newtheorem{theorem}[lemma]{Theorem}
\newtheorem{cor}[lemma]{Corollary}
\newtheorem{conj}[lemma]{Conjecture}

\newcommand{\bra}[1]      {\left\langle #1\right|}
\newcommand{\ket}[1]      {\left|#1\right\rangle}

\newcommand{\ketbra}[2]   {\left|#1\middle\rangle\!\middle\langle#2\right|}
\newcommand{\braopket}[3] {\left\langle #1\middle|#2\middle|#3\right\rangle}
\newcommand{\proj}[1]     {\left| #1\middle\rangle\!\middle\langle#1\right|}
\newcommand{\abs}[1]      {\left| #1 \right|}
\newcommand{\norm}[1]     {\left\| #1 \right\|}
\newcommand{\Tr}          {\mathrm{Tr}}
\newcommand{\ot}{\otimes}
\newcommand{\rel}[2]      {\!\left( #1 \middle\| #2 \right)}
\newcommand{\reli}[1]     {\!\left( \rho_{#1} \middle\| \sigma_{#1} \right)}
\newcommand{\pinch}[2]    {\mathcal P_{#2}\!\left(#1\right)}

\newenvironment{aligns}{\subequations \align} {\endalign \endsubequations}

\newcommand\E{ {\mathcal E} }
\newcommand{\tr}[1]{\mathrm{Tr}\left( #1 \right)}
\newcommand{\iden}{\mathbbm{1}}
\let\vv\v
\renewcommand{\v}[1]{\ensuremath{\boldsymbol #1}}

\usetikzlibrary{3d,shadings,fadings}
\tikzfading[name=fade out, 
    inner color=transparent!0!,
    outer color=transparent!100!]
\newcommand{\tikzsetnextfilename}[1]{}

\newcommand{\lefthat}[1]{
  {\vbox{\offinterlineskip\ialign{
    \hfil##\hfil\cr
    $\scriptscriptstyle\leftarrow$\cr
    \noalign{\kern.25ex}
    $#1$\cr
}}}}
\newcommand{\righthat}[1]{
  {\vbox{\offinterlineskip\ialign{
    \hfil##\hfil\cr
    $\scriptscriptstyle\rightarrow$\cr
    \noalign{\kern.25ex}
    $#1$\cr
}}}}

\DeclareFontFamily{U}{mathx}{}
\DeclareFontShape{U}{mathx}{m}{n}{ <-> mathx10 }{}
\DeclareSymbolFont{mathx}{U}{mathx}{m}{n}
\DeclareFontSubstitution{U}{mathx}{m}{n}
\DeclareMathAccent{\widecheck}{0}{mathx}{"71}

\newcommand\Dp{\overline D}
\newcommand\Dm{\widecheck D}
\newcommand{\DL}{\lefthat{D}}
\newcommand{\DR}{\righthat{D}}
\newcommand{\Dstar}{D^\star}
\newcommand{\betaL}{\lefthat{\beta}}
\newcommand{\betaR}{\righthat{\beta}}
\newcommand{\gammaL}{\lefthat{\gamma}}
\newcommand{\gammaR}{\righthat{\gamma}}
\newcommand{\GammaL}{\lefthat{\Gamma}}
\newcommand{\GammaR}{\righthat{\Gamma}}

\newcommand{\ev}{\stackrel{\text{ev.}}}
\newcommand{\ltev}{{\ev <}}
\newcommand{\gtev}{{\ev >}}

\newcommand{\kk}[1]{{\color{black}#1}}

\newcommand{\ctc}{\color{black}}

\newcommand{\ratetheorems}{thm:rate_zero,thm:rate_largedev_lo,thm:rate_moddev,thm:rate_smalldev,thm:rate_largedev_hi,thm:rate_extreme}

\newcommand{\ratetheoremsthermo}{thm:rate_smalldev,thm:rate_moddev,thm:rate_largedev_hi,thm:rate_extreme}

\newcommand\blfootnote[1]{
  \renewcommand\thefootnote{}\footnote{#1}
  \addtocounter{footnote}{-1}
  \renewcommand{\thefootnote}{\arabic{footnote}}
}

\makeatletter
\def\thmt@innercounters{}
\def\thmt@trivialref#1#2{
  \ifcsname r@#1\endcsname
    \@xa\@xa\@xa\thmt@trivi@lr@f\csname r@#1\endcsname\relax\@nil
  \else #2\fi 
  \ifx#2 (Restated)\fi
}
\makeatother

\begin{document}
\title{Quantum dichotomies and coherent thermodynamics beyond first-order asymptotics}

\blfootnote{*~These authors contributed equally.}

\author{Patryk Lipka-Bartosik$^*$}
\affiliation{Department of Applied Physics, University of Geneva, 1211 Geneva, Switzerland}

\author{Christopher T.\ Chubb$^{*,\dag}$}
\affiliation{Institute for Theoretical Physics, ETH Zurich, 8093 Z\"urich, Switzerland}
\blfootnote{\dag~\href{mailto:paper@christopherchubb.com}{\texttt{paper@christopherchubb.com}}}

\author{Joseph M.\ Renes}
\affiliation{Institute for Theoretical Physics, ETH Zurich, 8093 Z\"urich, Switzerland}

\author{Marco Tomamichel}
\affiliation{Department of Electrical and Computer Engineering, National University of Singapore, Singapore 117583, Singapore}
\affiliation{Centre for Quantum Technologies, National University of Singapore, Singapore 117543, Singapore}

\author{Kamil Korzekwa}
\affiliation{Faculty of Physics, Astronomy and Applied Computer Science, Jagiellonian University, 30-348 Kraków, Poland.}

\begin{abstract}
    We address the problem of exact and approximate transformation of quantum dichotomies in the asymptotic regime, i.e., the existence of a quantum channel $\E$ mapping $\rho_1^{\otimes n}$ into $\rho_2^{\otimes R_nn}$ with an error~$\epsilon_n$ (measured by trace distance) and $\sigma_1^{\otimes n}$ into $\sigma_2^{\otimes R_n n}$ exactly, for a large number~$n$. We derive second-order asymptotic expressions for the optimal transformation rate $R_n$ in the small, moderate, and large deviation error regimes, as well as the zero-error regime, for an arbitrary pair $(\rho_1,\sigma_1)$ of initial states and a commuting pair $(\rho_2,\sigma_2)$ of final states. We also prove that for $\sigma_1$ and $\sigma_2$ given by thermal Gibbs states, the derived optimal transformation rates in the first three regimes can be attained by thermal operations. This allows us, for the first time, to study the second-order asymptotics of thermodynamic state interconversion with fully general initial states that may have coherence between different energy eigenspaces. Thus, we discuss the optimal performance of thermodynamic protocols with coherent inputs and describe three novel resonance phenomena allowing one to significantly reduce transformation errors induced by finite-size effects. What is more, our result on quantum dichotomies can also be used to obtain, up to second-order asymptotic terms, optimal conversion rates between pure bipartite entangled states under local operations and classical communication.        

\end{abstract}

\date{\today}

\maketitle


\section{Introduction}


\subsection{Statistical inference}

Statistical inference is a powerful tool that allows us to explain the inner workings of the physical world by using statistical models based on data that holds crucial information about reality. From scientific discoveries to technological advancements, statistical inference is the backbone of many fields that have shaped our world. This process begins by forming a hypothesis, constructing an appropriate model (often represented by a family of probability distributions), and testing it against observed data. The theoretical foundations of statistical inference provide a solid framework for many essential fields, such as statistical estimation~\cite{altman1990practical,kay1993fundamentals,walter1997identification,cumming2013understanding}, metrology~\cite{weise1993bayesian,taylor1997introduction,rabinovich2006measurement,bernardo2009bayesian}, hypothesis testing~\cite{fisher1955statistical,lehmann1986testing,neyman1933ix,wald1939contributions}, decision theory~\cite{von1947theory,schoemaker1982expected,johnson1985effort,myerson1997game}, and machine learning~\cite{mitchell1997machine,bishop2006pattern,james2013introduction,pedregosa2011scikit}.

One of the central problems of the theory of statistical inference is to determine which statistical models are more informative, i.e., which probability distributions more accurately reflect reality~\cite{blackwell1953equivalent,hardy1952inequalities,cohen1998comparisons,SHANNON1958390,jorswieck2007majorization}. Given two probability distributions, $\mathbf{p}_1$ and $\mathbf{p}_2$, that describe some property of the physical system (e.g., the probability of observing given energy in the spectrum of a hydrogen atom), we say that $\mathbf{p}_1$ is more informative than~$\mathbf{p}_2$ when the latter can be obtained from the former by bistochastic processing. \kk{One can also imagine a more general situation where the physical system depends on some hidden parameter, and hence it can be described by multiple models, depending on the value of the hidden parameter (such a parameter can, for example, specify if the system is in or out of thermal equilibrium). One is then interested in quantifying how well a given collection of models describes the system in question. In the case when the hidden parameter is binary, the system can be described with a pair of probability distributions $(\mathbf{p}, \mathbf{q})$. Now, imagine that we want to decide whether one pair $(\mathbf{p}_1, \mathbf{q}_1)$ provides a better statistical model, i.e., is more informative,  than another pair $(\mathbf{p}_2, \mathbf{q}_2)$.} We say that a pair of probability distributions, or a \emph{dichotomy}, $(\mathbf{p}_1, \mathbf{q}_1)$, is more informative than $(\mathbf{p}_2, \mathbf{q}_2)$ when there exists stochastic processing which maps $\mathbf{p}_1$ into $\mathbf{p}_2$, while also mapping $\mathbf{q}_1$ into $\mathbf{q}_2$. When such processing exists, then the first dichotomy \emph{relatively majorises} the second~\cite{hardy1952inequalities}, a property that can be characterised using the techniques of hypothesis testing~\cite{blackwell1953equivalent}. 

Since the processes that \kk{underlie} our physical observations are fundamentally quantum and given the recent rapid development of quantum technologies, it is natural to ask how the techniques of statistical inference translate into the quantum realm. This is the main focus of quantum statistical inference~\cite{helstrom1969quantum,alberti1980problem,holevo1982testing,matsumoto2010,Buscemi_2012}, a theoretical framework that forms the bedrock of quantum estimation theory~\cite{helstrom1969quantum,aharonov1987phase,paris2009quantum,hyllus2012fisher}, quantum sensing and metrology~\cite{giovannetti2011advances,boss2017quantum,toth2014quantum,pezze2018quantum,pirandola2018advances,demkowicz2012elusive}, quantum statistical mechanics~\cite{reif2009fundamentals,hill1986introduction,hertz2018quantum}, and quantum computing~\cite{feynman2018simulating,cleve1998quantum,aaronson2013quantum}. The main conceptual difference between the classical and quantum statistical inference is the fact that statistical models in quantum theory must be described by density operators rather than probability distributions. Therefore, the objects to be compared are \emph{quantum dichotomies} denoted by $(\rho, \sigma)$ for density operators $\rho$ and $\sigma$. We say that the dichotomy $(\rho_1, \sigma_1)$ is more informative than $(\rho_2, \sigma_2)$ if there exists a quantum channel that jointly transforms~$\rho_1$ into~$\rho_2$ and~$\sigma_1$ into~$\sigma_2$. If such a channel exists, then the first dichotomy precedes the second one in the so-called Blackwell order~\cite{shmaya2005comparison,chefles2009quantum}. Importantly, when the two density operators forming a quantum dichotomy commute, they can be simultaneously diagonalised and can thus be treated classically. This is not the case for non-commuting quantum dichotomies, in which case the inference task becomes genuinely quantum. This regime naturally leads to a much richer behaviour, but is notoriously harder to characterise. 


\subsection{Quantum thermodynamics}

Perhaps one of the most impressive applications of statistical inference is in the field of thermodynamics. Indeed, modern thermodynamics started from the realisation that statistical models can effectively describe macroscopic processes like flows of heat and its fluctuations~\cite{Kubo_1966,callen1951}, phase transitions~\cite{stanley1987introduction,binder1987theory} or the dynamics of chemical reactions~\cite{prigogine1962chemical,reif2009fundamentals}. These processes generally involve unfathomable numbers of degrees of freedom, and therefore finding their complete description by solving the corresponding equations of motion is usually beyond reach. It is nowadays widely accepted that when the numbers of particles are large enough, one can use the techniques of statistical inference to build statistical models describing the physical system with an accuracy (or error) that increases (decreases) with the number of particles~\cite{gibbs1902elementary,ledoux2001concentration,Touchette_2015}. In the limit  when the system of interest is composed of infinitely many particles (the so-called thermodynamic limit), the approximation errors vanish and all relevant macroscopic observables can be fully characterised using only few relevant quantities known as thermodynamic potentials, e.g., the (equilibrium) free energy~\cite{reif2009fundamentals}. 

Thermodynamic limit is a convenient mathematical idealisation, but it cannot be justified in many experimentally and theoretically relevant situations. More specifically, when one is interested in the evolution of finite-size systems, fluctuations of thermodynamic variables cannot be neglected and the system's behaviour depends on more than a single thermodynamic potential. This regime is hardly discussed in thermodynamic textbooks, as it often requires rather advanced mathematical techniques of asymptotic analysis. Interestingly, this regime is surprisingly rich and allows one to investigate, i.a., the fundamental irreversibility of thermodynamic transformations~\cite{chubb2018beyond}, which cannot be observed when working solely in the thermodynamic limit. 

Some of the techniques developed within the framework of quantum statistical inference were recently adapted to study (quantum) thermodynamic processes. This led to the realisation that, in an idealised model of thermodynamics known as the resource theory of thermal operations~\cite{Janzing2000,horodecki2013fundamental,brandao2013resource,Brand_o_2015,Alhambra2016,Lostaglio2019,vinjanampathy2016quantum,halpern2015introducing,Mueller2018,LipkaBartosik2021}, a single quantity -- the (quantum) non-equilibrium free energy -- completely characterises the optimal rates of all thermodynamic transformations \cite{brandao2013resource}. This interpretation, however, is only valid in the thermodynamic limit of infinitely many copies of quantum systems. Despite many significant efforts, characterising thermodynamic transformations for general quantum states beyond the thermodynamic limit has remained a central problem for the resource theory of quantum thermodynamics. This difficulty can be easily understood once we realise that the techniques of statistical inference become accurate only when the numbers of particles are sufficiently large. On the other hand, it is known that quantum effects generally become less relevant with the increase in systems' size, meaning that either the coherence per particle vanishes~\cite{lostaglio2015quantum} or that the local observables begin to commute approximately when the system is comprised of a sufficient number of copies~\cite{Yunger_Halpern_2016}. Therefore, a natural question arises: Can we use the tools of quantum statistical inference to gain new insights into the thermodynamics of genuinely quantum systems beyond the thermodynamic limit? 


\subsection{Summary of results}

In this work we develop a unified mathematical framework that allows one to compare the informativeness of quantum dichotomies up to second-order asymptotics (i.e., when the transformed dichotomies consist of a large number of identical and independent systems) and for various error regimes. Our results are applicable for arbitrary input dichotomies, and commuting target dichotomies. This demonstrates, for the first time, how to compare quantum statistical models outside of the idealised limit of infinite repetitions of the experiments. Second, we apply our results on quantum dichotomies to study the fundamental laws governing thermodynamic transformations for large, but finite numbers of particles. As a consequence, we characterise thermodynamic transformations of general energy-coherent input states outside of the thermodynamic limit. We observe that, in this regime, quantum systems can be fully characterised using only a few relevant quantities, in complete analogy with the classical case. Importantly, this shows that the second-order analysis is an especially interesting regime where statistical inference remains highly accurate, while the quantum nature of the thermodynamic process still plays a prominent role. To demonstrate this, we study, in full generality, the fundamental thermodynamic protocols like work extraction, as well as quantify the minimal free energy dissipation when transforming quantum systems. We furthermore discover three novel resonance phenomena, the most interesting of which indicates that quantum coherence can be exploited to increase the reversibility of state transformations. Finally, we also discuss how our general results on quantum dichotomies can be used to bring novel and unifying insights into other fields, like the theory of entanglement or coherence.  

The paper is organised as follows. In \cref{sec:frame} we summarise the frameworks of quantum dichotomies (\cref{subsec:dich}), as well as the resource theories of thermodynamics (\cref{subsec:thermo}), entanglement (\cref{sec:frame_ent}), and define some relevant information-theoretic notions used throughout the paper (\cref{subsec:inf_notions}). In \cref{sec:results} we discuss our main results. In particular, after presenting an auxiliary lemma on sesquinormal distributions (\cref{subsec:sesq}), we outline our main technical results on quantum dichotomies (\cref{subsec:ncdich}), quantum thermodynamics (\cref{subsec:coh_thermo}) and entanglement (\cref{subsec:res_ent}). In \cref{sec:app} we discuss some of the applications of our results to the thermodynamic and entanglement scenarios. In particular, we show how our results can be used to determine optimal thermodynamic protocols with coherent inputs (\cref{subsec:opt_thermo_prot}), we investigate new types of resonance phenomena (\cref{subsec:res}), and briefly elaborate on the relevance of our results for entanglement theory (\cref{subsec:ent}). In \cref{sec:derive} we give proofs for the asymptotic results we described in previous sections. Specifically, we review and extend the relationship between quantum dichotomies and hypothesis testing (\cref{subsec:ht}), present some of the results on hypothesis testing (\cref{subsec:htass}) and prove the asymptotic transformation rates in different error regimes (\cref{subsec:rates}). Finally, we finish with \cref{sec:outlook}, which gives a short outlook on the potential further applications and extensions on our results. Technical derivations not required to understand the results are given in the \cref{app:pinch,app:sesqui,app:uni,app:consistency,app:res,app:two-side,app:thermal,app:battery,app:entanglement}.


\section{Framework}
\label{sec:frame}

We will denote by $\geq$ the L\"{o}wner partial order, i.e., for two Hermitian matrices $A$ and $B$ the relation $A \geq B$ means that $A - B$ is positive semi-definite. 
To measure distance between two density matrices, $\rho$ and $\sigma$, we will use trace distance  $T(\rho,\sigma):=\frac{1}{2}\|\rho - \sigma \|_{\textrm{tr}}$, where $\norm{X}_{\textrm{tr}} := \Tr |X|$ is the Schatten-1 norm. As a slight abuse of notation, we will also interchangeably refer to the total variation distance on classical distributions, $T(\v p,\v q):=\frac 12\sum_i\abs{p_i-q_i}$, as the trace distance. The fidelity between $\rho$ and $\sigma$ is given by $F(\rho, \sigma) := \norm{\sqrt{\rho}\sqrt{\sigma}}^2_{\textrm{tr}}$. All states we will consider are finite-dimensional, and we will denote the local dimension by $d$ when relevant. We take $\exp(\cdot)$ and $\log(\cdot)$ to be in an arbitrary but compatible base, and use $\ln(\cdot)$ to denote the natural logarithm.


\subsection{Quantum dichotomies}
\label{subsec:dich}

For two quantum dichotomies, $(\rho_1, \sigma_1)$ and $(\rho_2, \sigma_2)$, we will be interested whether there exists a completely positive trace-preserving map $\mathcal{E}$ such that $\rho_2 = \mathcal{E}(\rho_1)$ and $\sigma_2 = \mathcal{E}(\sigma_1)$. If such a channel exists, then we say that the first dichotomy precedes the second one in the Blackwell order~\cite{blackwell1953equivalent}, which we denote by $(\rho_1, \sigma_1) \succeq (\rho_2, \sigma_2)$\footnote{Note that the problem of comparing quantum dichotomies is sometimes formulated as a resource theory of asymmetric distinguishability~\cite{wang2019resource}. Since the two formulations are essentially equivalent, all of our results can be interpreted in this resource-theoretic way.}. We further consider the concept of an approximate Blackwell order by requiring that the two states are only reproduced approximately by the channel. That is, we write $(\rho_1, \sigma_1) \succeq_{(\epsilon_\rho,\epsilon_\sigma)} (\rho_2, \sigma_2)$ if and only if there exists a quantum channel $\mathcal{E}$ such that 
\begin{align}
    T(\mathcal{E}(\rho_1),\rho_2 ) \leq \epsilon_\rho \quad \textnormal{and} \quad
    T(\mathcal{E}(\sigma_1),\sigma_2 ) \leq \epsilon_\sigma.
\end{align}

It is known that for commuting dichotomies, \mbox{$[\rho_1, \sigma_1] = [\rho_2, \sigma_2] = 0$}, the problem of determining a suitable channel reduces to the classical problem of comparing probability distributions. It was observed in Ref.~\cite{renes2016relative} that in this case, by employing Blackwell’s equivalence theorem~\cite{blackwell1953equivalent}, one can show that $(\rho_1, \sigma_1) \succeq_{(\epsilon_\rho,\epsilon_\sigma)} (\rho_2, \sigma_2)$ if and only if
\begin{align}
    \label{eq:transformation_conditions}
    \beta_{x}(\rho_1 \| \sigma_1) \leq \beta_{x-\epsilon_\rho}(\rho_2 \| \sigma_2) + \epsilon_\sigma \quad  \forall x \in (\epsilon_\rho,1).
\end{align}
Here, $\beta_{x}(\rho\|\sigma)$ is the solution of the semi-definite optimisation problem
\begin{aligns}
    \label{eq:beta1}
    \min_Q \quad & \Tr(\sigma Q), \\
    \label{eq:beta2}
    \textnormal{subject to} \quad & 0 \leq Q \leq 1, \\
    \label{eq:beta3}
    & \Tr (\rho Q) \geq 1 - x.
\end{aligns}
The two quantities, $x$ and $\beta_x(\rho\|\gamma)$, can be interpreted as two errors appearing in a binary hypothesis testing problem. More specifically, $\beta_x(\rho\|\sigma)$ is the minimum type-II error given that the type-I error is upper bounded by~$x$ for a binary hypothesis testing with a null hypothesis $\rho$ and an alternative hypothesis $\sigma$~\cite{helstrom1969quantum}. In the fully quantum case, i.e., when $[\rho_1,\sigma_1]\neq 0$ and $[\rho_2, \sigma_2] \neq 0$, the conditions specified by Eq.~\eqref{eq:transformation_conditions} (and referred to as relative majorisation preorder in Ref.~\cite{renes2016relative}) no longer characterise the Blackwell's order~\cite{Matsumoto2014,jenvcova2012comparison,reeb2011hilbert,buscemi2012comparison}, beyond the simplest case of two-dimensional density matrices~\cite{alberti1980problem}. \kk{For attempts to overcome this limitation, see, e.g., Refs.~\cite{buscemi2012comparison,jenvcova2016comparison,gour2018quantum}.}


\subsection{Resource theory of thermodynamics}
\label{subsec:thermo}

In the resource-theoretic approach to thermodynamics, one focuses on a system $S$ with a Hamiltonian \mbox{$H = \sum_{i=1}^{d} E_i \ketbra{i}{i}$} and a heat bath $B$ at some fixed inverse temperature $\beta$ with an arbitrary Hamiltonian~$H_{{B}}$~\cite{Janzing2000,horodecki2013fundamental}. The heat bath is always assumed to be prepared in a thermal Gibbs state,
\begin{equation}
    \gamma_B = \frac{e^{-\beta H_{B}}}{Z_{B}},\quad Z_{{B}} = \tr{e^{-\beta H_{{B}}}}.
\end{equation} 
The interaction of the system with the heat bath is mediated by a unitary $U$ that conserves the total energy, i.e., obeys the additive conservation law \mbox{$[U, H\otimes\iden_B+ \iden\otimes H_{B}] = 0$}. The effective map $\mathcal{E}$ that is obtained by evolving the system and the heat bath using unitary $U$ and discarding part of the joint system is called a \emph{thermal operation} (TO) and can be formally written as
\begin{align}
    \label{eq:TO}
    \mathcal{E}(\rho) = \Tr_{B'} \left[U \left(\rho \ot \gamma_{B}\right) U^{\dagger}\right],
\end{align}
where the partial trace can be performed over any subsystem $B'$ of the joint system. Note that since we allow for $B'\neq B$, the Hamiltonian of the final system may differ from $H$, and so we will use $\gamma_1$ and $\gamma_2$ to denote the Gibbs thermal states of the initial and final systems. We say that $\rho_1 \xrightarrow[\mathrm{TO}]{\epsilon} \rho_2$ when there exists a thermal operation $\E$ such that \mbox{$\E(\rho_1) = \tilde{\rho}_2$}, with $\tilde{\rho}_2$ being a final state that is $\epsilon$-close to the target state $\rho_2$ in trace distance, i.e. $T(\tilde{\rho}_2,\rho_2)=\epsilon$.

Characterising the set of transitions achievable via thermal operations in full generality remains an open problem. In the semi-classical case, i.e., when $\rho_1$ and~$\rho_2$ are block-diagonal in the energy eigenbasis (or equivalently when $[\rho_1,\gamma_1]=[\rho_2,\gamma_2]=0$), the existence of a thermal operation transforming $\rho_1$ into $\rho_2$ while changing the Hamiltonian from $H_1$ to $H_2$ is equivalent to the existence of an \emph{arbitrary} quantum channel mapping a quantum dichotomy $(\rho_1, \gamma_1)$ into $(\rho_2, \gamma_2)$~\cite{horodecki2013fundamental,Lostaglio2019}. As a consequence, the Blackwell's theorem in this case fully characterises the set of states achievable under thermal operations~\cite{horodecki2013fundamental}. More specifically, as observed in Ref.~\cite{renes2016relative}, for energy-incoherent (block-diagonal) states $\rho_1$ and $\rho_2$, we have $\rho_1 \xrightarrow[\mathrm{TO}]{\epsilon} \rho_2$ if and only if
\begin{align}
   \beta_{x}(\rho_1 \| \gamma_1) \leq \beta_{x-\epsilon}(\rho_2 \| \gamma_2)   \hspace{10pt} \text{for all} \hspace{10pt} x \in (\epsilon,1).
\end{align}
The above condition is just a special case of Eq.~\eqref{eq:transformation_conditions}, and thus we see that the problems of transforming quantum dichotomies and the thermodynamic state transformation are very closely related. 


\subsection{Resource theory of entanglement}
\label{sec:frame_ent}

The resource theory of entanglement investigates the scenario where a bipartite system is distributed between two spatially separated  agents~\cite{nielsen_chuang_2010}. The agents can act locally on their respective parts and can exchange classical information. The resulting set of free operations is called \emph{local operations and classical communication} (LOCC). Free states of this theory, i.e., states that can be prepared using only LOCC, are given by all separable states. While a complete characterisation of LOCC transformations for general mixed states remains an open problem, for pure states there exists a relatively simple characterisation known as the Nielsen's theorem~\cite{nielsen1999conditions,lo2001concentrating}. The theorem states that a pure bipartite state $\psi_1$ with Schmidt coefficients $\v{p_1}$ can be converted into state $\psi_2$ with Schmidt coefficients $\v{p_2}$ by means of LOCC if and only if there exists a bistochastic matrix mapping $\v{p_2}$ to~$\v{p_1}$. 

It was then observed in Ref.~\cite{renes2016relative} that Nielsen's theorem can be formulated in the language of quantum dichotomies when Schmidt vectors of input and output states, $\v{p}_1$ and $\v{p}_2$, have equal dimension. More specifically, by denoting with $\rho_i$ diagonal matrices with $\v{p}_i$ on the diagonals, the existence of a transformation that $(1)$ maps $\rho_1$ to $\rho_2$ with a transformation error $\epsilon$, and $(2)$ maps a maximally mixed state into itself, is equivalent to the existence of a bistochastic matrix mapping $\v{p}_2$ into $\v{p}_1$ with an error $\epsilon$. Now, in Ref.~\cite{chubb2018beyond} (see the generalisation of Lemma~12 in Appendix~D therein) it was shown that the latter is equivalent to the existence of a bistochastic matrix mapping a distribution $\epsilon$-close to $\v{p}_2$ into $\v{p}_1$. This means that an LOCC map transforming $\psi_1$ into $\psi_2$ with a transformation error $\epsilon$ exists if and only if a quantum dichotomy $(\rho_2, \iden_d/d)$ can be approximately transformed into $(\rho_1, \iden_d/d)$. 

To deal with the case of systems with different lengths of Schmidt vectors, $d_1$ (for input) and $d_2$ (for output), one can extend the input system with a pure bipartite separable state with local dimensions $d_2$ and the output system with an analogous state with local dimensions $d_1$. Then, there exists an LOCC map transforming general pure bipartite state $\psi_1$ into $\psi_2'$ whose Schmidt vector $\v{p_2'}$ is $\epsilon$ away in total variation distance from the Schmidt vector $\v{p_2}$ of $\psi_2$, if and only if 
\begin{equation}
 \!\!\left(\!\rho_{2}\otimes \ketbra{0}{0}_{d_1}\!,\! \frac{\iden_{d_1d_2}}{d_1d_2}\!\right) 
 \!\succeq_{(\epsilon,0)}\!
  \left(\!\rho_{1}\!\otimes\! \ketbra{0}{0}_{d_2}, \frac{\iden_{d_1d_2}}{d_1d_2}\!\right)\!.  \!
\end{equation}
Since states appearing in these dichotomies commute, the Blackwell's theorem fully characterises states achievable under LOCC. More specifically, an LOCC transformation $\psi_1 \xrightarrow[\mathrm{LOCC}]{\epsilon} \psi_2$ exists if and only for all $x \in (\epsilon,1)$ one has
\begin{align}
    \label{eq:entanglement_condition}
   d_2 \beta_{x}\left(\rho_{2} \middle\| \frac{\iden_{d_2}}{d_2}\right) \leq d_1 \beta_{x-\epsilon}\left(\rho_{1} \middle\| \frac{\iden_{d_1}}{d_1}\right).
\end{align}
Of course, the  above conditions are again a special case of Eq.~\eqref{eq:transformation_conditions}.   


\subsection{Information-theoretic and statistical notions}
\label{subsec:inf_notions}
To formulate our results, we will need the following notions. First, the von Neumann entropy and entropy variance are defined as
\begin{aligns}
    S(\rho):=&-\Tr \left(\rho \log \rho \right),\\
    V(\rho):=&\Tr \left(\rho (\log\rho)^2\right) -S(\rho)^2,
\end{aligns}
and their relative cousins, the relative entropy~\cite{umegaki62conditional} and the relative entropy variance~\cite{tomamichel2013hierarchy,li2014second}, as
\begin{aligns}
    D(\rho\|\sigma)&:=\left(\Tr\rho\left(\log\rho-\log\sigma\right)\right),\\
    V(\rho\|\sigma)&:=\Tr\left(\rho\left(\log\rho-\log\sigma\right)^2\right)-D(\rho\|\sigma)^2.
\end{aligns}
Note that, for $\sigma$ given by the thermal Gibbs state $\gamma$, the above quantities can be interpreted as non-equilibrium free energy~\cite{brandao2013resource} and free energy fluctuations~\cite{chubb2018beyond,biswas2022fluctuation} respectively. We also define two variants of the R\'enyi relative entropy \cite{renyi1961measures}, namely the Petz relative entropy $\Dp_\alpha$ \cite{petz1986quasi} and the minimal relative entropy $\Dm_\alpha$ \cite{muller2013quantum,AudenaertDatta2013,Tomamichel2016,WildeWinterYang2013}, that is
\begin{aligns}
    \!\!\!\!\Dp_\alpha(\rho\|\sigma):=&\frac{\log\Tr \left(\rho^{\alpha}\sigma^{1-\alpha}\right)}{\alpha-1},\\
    \!\!\!\!\Dm_\alpha(\rho\|\sigma):=&
    \begin{dcases}
    \frac{\log\Tr\left(\! \left(\sqrt \rho\sigma^{\frac{1-\alpha}{\alpha}}\sqrt \rho\right)^\alpha\right)}{\alpha-1} & \alpha\geq \frac{1}{2},\\
    \frac{\log\Tr\left(\! \left(\sqrt \sigma\rho^{\frac{\alpha}{1-\alpha}}\sqrt \sigma\right)^{1-\alpha}\right)}{\alpha-1} & \alpha\leq \frac{1}{2}.
    \end{dcases}
\end{aligns}
Note that if the states are commuting, $[\rho,\sigma]=0$, then both relative entropies are identical, and in this case we shall denote this without adornment as $D_\alpha$. Finally, for classical probability distributions, we will also use the Shannon entropy and the related entropy variance,
\begin{aligns}
    H(\v{p})&:=-\sum_i p_i \log p_i,\\
    V(\v{p})&:=\sum_i p_i \left(\log p_i - H(\v{p})\right)^2,
\end{aligns}
as well as the R\'enyi entropies,
\begin{equation}
    H_\alpha(\v{p})=\frac{1}{1-\alpha}\log\left(\sum_ip_i^\alpha\right).
\end{equation}

The probability density function and the cumulative distribution function of a normal distribution with mean $\mu$ and variance $\nu$ will be denoted by $\phi_{\mu,\nu}(x)$ and $\Phi_{\mu,\nu}(x)$, whereas their standardised versions (with $\mu=0$ and \mbox{$\nu=1$}) by $\phi(x)$ and $\Phi(x)$. We also introduce the following function,
\begin{align} 
    \label{eq:sesqui_gen}
    S_\nu^{(\delta)}(\mu):=\inf_{A\geq \Phi}{\ctc \delta(A',\phi_{\mu,\nu})}
    ,
 \end{align}
where $\nu\in\mathbb R^+$ is a parameter, $\mu\in \mathbb R$, $\delta$ is a statistical distance, and the infimum is taken over cumulative distribution functions $A$ {\ctc(with probability density function $A'$)} which are pointwise greater than~$\Phi$. As we shall see in \cref{lem:sesqui}, this function is a cumulative distribution function if $\delta$ is chosen to be the trace distance. The introduction of $S_\nu^{(\delta)}$ is inspired by Ref.~\cite{kumagai2016second}, where the authors investigated its special case, called the Rayleigh-Normal distribution, with $\delta$ given by infidelity distance. The name of the function comes from the fact that, as $\nu$ is varied, it interpolates between the normal and the Rayleigh distribution. In this paper, we will mainly focus on another special case, with $\delta$ given by the trace distance, and will denote the corresponding cumulative distribution function simply as
\begin{align}
    \label{eq:sesqui_def}
    S_\nu(\mu):=\frac 12 \inf_{A\geq \Phi}\int_{\mathbb R }\abs{A'(x)-\phi_{\mu,\nu}(x)}\mathrm dx.
\end{align}
We will refer to the above as the \emph{sesquinormal distribution}, since we will prove that it interpolates between the normal and half-normal distributions for varying $\nu$.


\section{Results}
\label{sec:results}

The main technical result of this paper consists of a unified approach for capturing the problem of optimal transformations of quantum dichotomies in the small, moderate, large and extreme deviation regimes. It not only provides a much simpler and clearer derivation than the previously known results employing infidelity to measure transformation error~\cite{kumagai2016second,chubb2018beyond,chubb2019moderate}, but it also extends the formalism to the case of non-commuting input states. This, in turn, leads to the main conceptual result of the paper: The generalisation of the second-order asymptotic analysis of thermodynamic state interconversion to the case of general (energy-coherent) input states. Before formally stating all these results, however, we first present auxiliary results that concern the properties of the sesquinormal distribution, which may be of independent interest.


\subsection{Sesquinormal distribution}
\label{subsec:sesq}
The sesquinormal distribution was defined implicitly via an optimisation in \cref{eq:sesqui_def}. We start by giving an explicitly closed-form solution of this optimisation problem and specify some relevant properties of the sesquinormal distribution.

\begin{restatable}[Sesquinormal distribution]{lemma}{sesqui}
    \label{lem:sesqui}
    The function $S_\nu$ is a cumulative distribution function (cdf) for any \mbox{$\nu\in[0,\infty)$}. Moreover, for \mbox{$\nu\notin\lbrace 0,1,\infty\rbrace$} the cdf has the closed form 
    \begin{align}
          S_{\nu}(\mu)=&\Phi\left(\frac{\mu-\sqrt\nu\sqrt{\mu^2+(\nu-1)\ln\nu}}{1-\nu}\right)\\
          &\qquad-\Phi\left( \frac{\sqrt{\nu}\mu-\sqrt{\mu^2+(\nu-1)\ln\nu}}{1-\nu} \right),\notag
    \end{align}
    and for $0<\nu<\infty$ the inverse cdf can be expressed as
    \begin{align}
        S_{\nu}^{-1}(\epsilon)=\min_{x\in(\epsilon,1)}\sqrt{\nu}\Phi^{-1}(x)-\Phi^{-1}(x-\epsilon).
    \end{align}
    The extreme cases $\nu=0$ and $\nu\to\infty$ reduce to the normal distribution
    \begin{align}
        S_0(\mu)=\lim_{\nu\to\infty}S_\nu(\sqrt \nu \mu)&=\Phi(\mu),
    \end{align}
    and the $\nu=1$ reduces to the half-normal distribution
    \begin{align}
        S_1(\mu)=\max\lbrace 2\Phi(\mu/2)-1,0\rbrace.
    \end{align}
    Finally, the family of sesquinormal distributions has a duality under reciprocating the parameter,
    \begin{align}
        S_\nu(\mu)=S_{1/\nu}(\mu/\sqrt \nu)
        ~~\text{or}~~S_\nu^{-1}(\epsilon)=\sqrt \nu S_{1/\nu}^{-1}(\epsilon).
    \end{align}
\end{restatable}
\begin{proof}
    See \cref{app:sesqui}.    
\end{proof}


\subsection{Non-commuting quantum dichotomies}
\label{subsec:ncdich}

We now turn to our central results on the second-order asymptotic analyses of transformation rates between quantum dichotomies in all error regimes. Specifically, let $R_n^*(\epsilon_n)$ denote the largest rate $R_n$ such that
\begin{align}
    \left(\rho_1^{\otimes n},\sigma_1^{\otimes n}\right)
    \succeq_{(\epsilon_n,0)}
    \left(\rho_2^{\otimes R_nn},\sigma_2^{\otimes R_nn}\right).
\end{align}
\cref{\ratetheorems} will all concern the asymptotic scaling of $R_n^*(\epsilon_n)$, split by the scaling of the error $\epsilon_n$ measured by trace distance. We note that one could also consider a two-sided error variant of this problem with a pair of error sequences $\epsilon_{n}^{(\rho)}$ and $\epsilon_{n}^{(\sigma)}$. We shall neglect this more general problem in the body of this paper, but cover the extension of our results to this regime in \cref{app:two-side}. We do this partially because these two-sided results are not applicable to the resource theoretic problems we are mostly focused on, and partially because this two-sided problem is in fact no more rich, with the optimal transformation simply diverging to infinity in many regimes.

Before we move on to the second-order analysis, we start with the previously studied~\cite{brandao2013resource,Gour_2022,buscemi2019information,wang2019resource}\footnote{It should be noted there is some dispute on the status of proofs of \cref{thm:rate_first}. Ref.~\cite{brandao2013resource} claims to prove it, for non-commuting inputs \emph{and} outputs. However, Ref.~\cite{Gour_2022} raised questions of possible gaps in this proof, especially surrounding the case of non-commuting outputs.} first-order case, which states that the asymptotic transformation rate is controlled by the relative entropy:
\begin{restatable}[First-order rate]{theorem}{ratefirst}
    \label{thm:rate_first}
    For constant \mbox{$\epsilon\in(0,1)$} and $[\rho_2,\sigma_2]=0$, the optimal rate converges:
    \begin{align}
        \lim_{n\to\infty} R_n^*(\epsilon)=\frac{D\reli{1}}{D\reli{2}}.
    \end{align}
    Furthermore, if we consider more general target dichotomies, $[\rho_2,\sigma_2]\neq 0$, then we still have the upper bound:
    \begin{align}
        \limsup_{n\to\infty} R_n^*(\epsilon)\leq \frac{D\reli{1}}{D\reli{2}}.
    \end{align}
\end{restatable}
\noindent \textit{Proof.} See \cref{subsec:rates}.

Second-order asymptotics form refinements of \cref{thm:rate_first} that quantify the rate of convergence to this first-order behaviour. A diagram of the different second-order regimes is presented in \cref{fig:sigmoid_rate}. Below we will state all of our second-order theorems, with their proofs left to \cref{sec:derive}. This analysis is divided up based on the scaling of the error. 

\begin{figure*}[t]
    \centering
    \includegraphics[width=\linewidth]{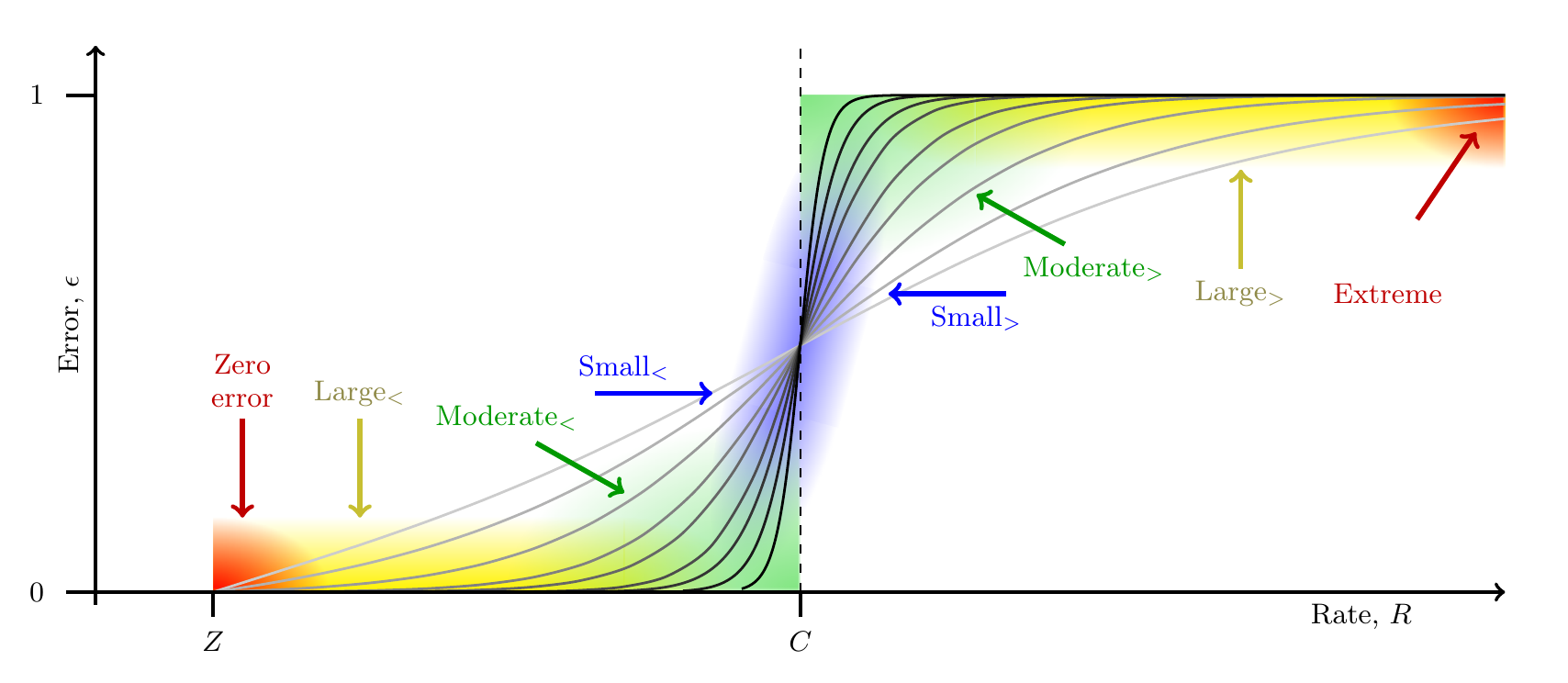}
    \def\yes{{\color{darkgreen}Yes}}
    \def\no{{\color{red}No}}
    
    \begin{tabular}{ccccccccc}
    \toprule
    \multicolumn{2}{c}{\bfseries Regime} & ~~~~ & \textbf{Error ($\epsilon_n$)} & ~~~~ & \textbf{Rate ($R_n$)} &~~& \textbf{Tight?} \\ \midrule                          
    Zero-error & \cref{thm:rate_zero} && Zero && $Z$ && \no\\
    Large$_<$ & \cref{thm:rate_largedev_lo} && Approaching 0 exponentially && $\bigl[Z,C\bigr)$ && \no\\
    Moderate$_<$ & \cref{thm:rate_moddev} && Approaching 0 subexponentially && $C-\omega(1/\sqrt n)\cap o(1)$ && \yes\\
    Small$_<$ & \cref{thm:rate_smalldev} && Constant, $<0.5$ && $C-\Theta(1/\sqrt n)$ && \yes\\ \midrule
    Small$_>$ & \cref{thm:rate_smalldev} && Constant, $>0.5$ && $C+\Theta(1/\sqrt n)$ && \yes\\
    Moderate$_>$ & \cref{thm:rate_moddev} && Approaching 1 subexponentially && $C+\omega(1/\sqrt n)\cap o(1)$ && \yes \\
    Large$_>$ & \cref{thm:rate_largedev_hi} && Approaching 1 exponentialy && $\bigl(C,\infty\bigr)$ && \yes \\
    Extreme & \cref{thm:rate_extreme} && Approaching 1 superexponentially && $\infty$ && \yes \\ \bottomrule
    \end{tabular}

    \caption{\textbf{Summary of our main results.} Asymptotics of transformation rates between quantum dichotomies \mbox{$(\rho_1,\sigma_1)\to(\rho_2,\sigma_2)$} with an error of at most $\epsilon_n$ allowed on the first state. The table summarises the different error regimes, i.e.,\ the different manners in which the error $\epsilon_n$ and rate $R_n$ can scale. In the above, the first-order rate is $C:=D\reli{1}/D\reli{2}$ and zero-error rate is $Z$. For each result we just have upper bounds for general target dichotomies, but for commuting targets, $[\rho_2,\sigma_2]=0$, we have upper \emph{and} lower bounds. The final column denotes whether these bounds coincide, which they do in all-but-one regime.
    }
    \label{fig:sigmoid_rate}
\end{figure*}

The first regime we consider is that of \emph{small deviations}, in which the errors considered are constants other than $0$ or $1$. In this regime, we find that the rate approaches the first-order rate as $O(1/\sqrt n)$, quantified by the relative entropy variance $V\rel{\cdot}{\cdot}$ as well as the sesquinormal distribution $S_{1/\xi}$, where $\xi$ is the \emph{reversibility parameter}~\cite{kumagai2016second,chubb2018beyond}, given by
\begin{align}
    \label{eq:reversibility}
    \xi:=\frac{V(\rho_1\|\sigma_1)}{D(\rho_1\|\sigma_1)} \bigg/ \frac{V(\rho_2\|\sigma_2)}{D(\rho_2\|\sigma_2)}.
\end{align}
Given these, the scaling of the rate in the small deviation regime is:
\begin{restatable}[Small deviation rate]{theorem}{ratesmalldev}
    \label{thm:rate_smalldev}
    Let $\lesssim$/$\simeq$ denote (in)equality up to $o(1/\sqrt n)$. For constant \mbox{$\epsilon\in(0,1)$}, and $[\rho_2,\sigma_2]=0$, the optimal rate is
    \begin{align}
        \label{eq:small_deviation}
        R_n^*(\epsilon)\simeq\frac{D(\rho_1\|\sigma_1)+\sqrt{V\reli{1}/n}\cdot S_{1/\xi}^{-1}(\epsilon)}{D(\rho_2\|\sigma_2)}.
    \end{align}
    Furthermore, if we consider general output dichotomies, $[\rho_2,\sigma_2]\neq0$, then we still have the upper bound
    \begin{align}
        R_n^*(\epsilon)\lesssim\frac{D(\rho_1\|\sigma_1)+\sqrt{V\reli{1}/n}\cdot S_{1/\xi}^{-1}(\epsilon)}{D(\rho_2\|\sigma_2)}.
    \end{align}
\end{restatable}
\noindent \textit{Proof.} See \cref{subsubsec:rate_small}.

The second regime we consider is that of \emph{moderate deviations}, in which errors are tending towards either $0$ or~$1$, but only doing so subexponentially. This causes the rate to approach the first-order rate slower than $O(1/\sqrt n)$, specifically:
\begin{restatable}[Moderate deviation rate]{theorem}{ratemoddev}
    \label{thm:rate_moddev}
    Consider an \mbox{$a\in(0,1)$}, and let $\lesssim/\simeq$ denote (in)equality up to $o\left(\sqrt{n^{a-1}}\right)$. Let $\epsilon_n:=\exp(-\lambda n^a)$ for some $\lambda>0$. For $[\rho_2,\sigma_2]=0$, the optimal rate is 
    \begin{aligns}
        R^*_n(\epsilon_n)&\simeq \frac{
        D(\rho_1\|\sigma_1)
        -\abs{1-\xi^{-1/2}}\sqrt{2\lambda V\reli{1}n^{a-1}}
        }{D(\rho_2\|\sigma_2)},\\
        R^*_n(1-\epsilon_n)&\simeq \frac{
        D(\rho_1\|\sigma_1)
        +\left[1+\xi^{-1/2}\right]\sqrt{2\lambda V\reli{1}n^{a-1}}
        }{D(\rho_2\|\sigma_2)}.
    \end{aligns}
    Furthermore, if we consider general output dichotomies, $[\rho_2,\sigma_2]\neq 0$, then we still have the upper bounds
    \begin{aligns}
        R^*_n(\epsilon_n)&\lesssim \frac{
        D(\rho_1\|\sigma_1)
        -\abs{1-\xi^{-1/2}}\sqrt{2\lambda V\reli{1}n^{a-1}}
        }{D(\rho_2\|\sigma_2)},\\
        R^*_n(1-\epsilon_n)&\lesssim \frac{
        D(\rho_1\|\sigma_1)
        +\left[1+\xi^{-1/2}\right]\sqrt{2\lambda V\reli{1}n^{a-1}}
        }{D(\rho_2\|\sigma_2)}.
    \end{aligns}
\end{restatable}
\begin{proof}
    See \cref{subsubsec:rate_moderate}.
\end{proof}
Third is the \emph{large deviations} regime, in which the error is either exponentially approaching $0$ (large deviation, low-error), or exponentially approaching $1$ (large deviation, high-error). In this case, the error is small/large enough, so that the asymptotic rate shifts away from the first-order rate, and now depends not just on the relative entropy, but also the R\'enyi relative entropies, specifically:
\begin{restatable}[Large deviation rate, low-error]{theorem}{ratelargedevlo}
    \label{thm:rate_largedev_lo}
    For any error of the form $\epsilon_n=\exp(-\lambda n)$ with constant $\lambda>0$, if $[\rho_2,\sigma_2]=0$, then the optimal rate is lower bounded by
    \begin{align}
        \liminf_{n\to\infty} R_n^*(\epsilon_n)\geq  \min_{-\lambda\leq \mu\leq \lambda} \widecheck r(\mu).
    \end{align}
    Furthermore, if we consider general output dichotomies, $[\rho_2,\sigma_2]\neq 0$, then
    the optimal rate is upper bounded by
    \begin{align}
        \limsup_{n\to\infty} R_n^*(\epsilon_n) \leq \min_{-\lambda\leq \mu\leq \lambda} \overline r(\mu).
    \end{align}
    In the above, $\overline r$ and $\widecheck r$ are defined in terms of R\'{e}nyi relative entropies in \cref{subsubsec:rate_large}, and coincide when \mbox{$[\rho_1,\sigma_1]=[\rho_2,\sigma_2]=0$}.
\end{restatable}
\noindent \textit{Proof.} See \cref{subsubsec:rate_large}.

\begin{restatable}[Large deviation rate, high-error]{theorem}{ratelargedevhi}
    \label{thm:rate_largedev_hi}
    For any error of the form $\epsilon_n=1-\exp(-\lambda n)$ with constant $\lambda>0$, if $[\rho_2,\sigma_2]=0$ then the optimal rate is
    \begin{align}
        \lim_{n\to\infty }R_n^*(\epsilon_n)=\inf_{\substack{t_1>1\\0<t_2<1}}\frac{\Dm_{t_1}\reli{1}+\left(\frac{t_1}{t_1-1}+\frac{t_2}{1-t_2}\right)\lambda}{D_{t_2}\reli{2}}.
    \end{align}
    Furthermore, if we consider general output dichotomies, $[\rho_2,\sigma_2]\neq 0$, then we still have the upper bound
    \begin{align}
        \limsup_{n\to\infty }R_n^*(\epsilon_n)\leq \inf_{\substack{t_1>1\\0<t_2<1}}\frac{\Dm_{t_1}\reli{1}+\left(\frac{t_1}{t_1-1}+\frac{t_2}{1-t_2}\right)\lambda}{\Dp_{t_2}\reli{2}}.
    \end{align}
\end{restatable}
\noindent \textit{Proof.} See \cref{subsubsec:rate_large}.

Finally, the fourth regime we consider is that of \emph{extreme deviations}, i.e., with super-exponentially decaying errors. The first is the low-error case of errors superexponentially approaching zero, including exactly zero-error. In this case, we get an expression for the asymptotic rate which is quite similar to the first-order expression, but involves a minimisation over the minimal relative entropies instead of just \emph{the} relative entropy. It gives an additional operational interpretation of the minimal R\'enyi entropy~\cite{LiYao2022}, and is a non-commutative generalisation of Ref.~\cite{MuPomattoStrackTamuz2019,FarooqFritzHaapasaloTomamichel2023}. Specifically, the zero-error rate is:
\begin{restatable}[Zero-error rate]{theorem}{ratezero}
    \label{thm:rate_zero}
    For $[\rho_2,\sigma_2]=0$ the optimal zero-error rate is {\ctc lower bounded}
    \begin{align}
        \ctc \liminf_{n\to\infty}R_n^*(0)
        &\geq \ctc \max\left\lbrace
        \inf_{\alpha\in\mathbb R} \frac{\DL_\alpha\reli1}{D_\alpha\reli2},
        \inf_{\alpha\in\mathbb R} \frac{\DR_\alpha\reli1}{D_\alpha\reli2}
        \right\rbrace.
    \end{align}
    where the divergences $\DL_\alpha$ and $\DR_\alpha$ are defined in \eqref{eq:pinched_renyi_def1} and \eqref{eq:pinched_renyi_def2}. More generally, if $[\rho_2,\sigma_2]\neq 0$, then the optimal transformation rate for all $n$ is upper bounded
    \begin{align}
        R_n^*(0) \leq 
        \min_{\alpha\in\overline{\mathbb R}} \frac{\Dm_\alpha\reli{1}}{\Dm_\alpha\reli{2}}.
    \end{align}
\end{restatable}
\noindent \textit{Proof.} See \cref{subsubsec:rate_extreme}.

Lastly, we are left with the final case of errors exponentially approaching $1$, wherein the rate diverges to infinity:
\begin{restatable}[Extremely high-error rate]{theorem}{rateextreme}
    \label{thm:rate_extreme}
    For \mbox{$[\rho_2,\sigma_2]=0$}, if the error is allowed to be super-exponentially close to $1$, then the optimal rate is unbounded,
    \begin{align}
        \lim_{n\to\infty}R_n^*\bigl(1-\exp(-\omega(n))\bigr)=\infty.
    \end{align}
\end{restatable}
\noindent \textit{Proof.} See \cref{subsubsec:rate_extreme}.

Given these theorems, we also make two conjectures. Firstly, we notice that the form of the small deviation result, \cref{thm:rate_smalldev}, is identical to the infidelity-based results in Refs.~\cite{kumagai2016second,chubb2018beyond}. Our first conjecture is that this extends more generally to other distance measures.

\begin{conj}
    For any fixed and non-maximal transformation error $\epsilon>0$ measured by a (quantum) statistical distance $\delta$ (perhaps subject to some additional `niceness' constraints), the optimal rate for transforming quantum dichotomies with commuting target states in the small deviation regime is given by Eq.~\eqref{eq:small_deviation} with the sesquinormal distribution, $S_{1/\xi}$, replaced by the generalised Rayleigh-Normal distributions $S_{1/\xi}^{(\delta)}$, defined in Eq.~\eqref{eq:sesqui_gen}.
\end{conj}

Secondly, all of the achievability bounds rely on connections to hypothesis testing that only apply for commuting targets $[\rho_2,\sigma_2]=0$, while all of our optimality bounds apply for general states. We conjecture that there might exist alternative protocols capable of saturating these bounds.
\begin{conj}
\label{conj:quantum}
    All of the optimality bounds in \cref{\ratetheorems} are achievable, for general states, $[\rho_1,\sigma_1]\neq 0$ and $[\rho_2,\sigma_2]\neq 0$.
\end{conj}


\subsection{Coherent quantum thermodynamics}
\label{subsec:coh_thermo}

Our technical results find applications in quantum thermodynamics because of the following result, whose proof can be found in \cref{app:thermal}.

\begin{theorem}
    \label{thm:thermo}
    For $\sigma_1=\gamma_1$ and $\sigma_2=\gamma_2$, both being thermal Gibbs states, the optimal transformation rates, captured by
    \cref{\ratetheoremsthermo}, can be attained by thermal operations. Moreover, for energy-incoherent input and output states, $\rho_1$ and $\rho_2$, this extends to all error regimes, i.e., \cref{\ratetheorems} characterise optimal transformation rates under thermal operations.
\end{theorem}

Thus, \cref{\ratetheoremsthermo} describe optimal rates~$R^*_n$ for state transformations under thermal operations between $n$ copies of generic quantum states $\rho_1$ and $R^*_nn$ copies of energy-incoherent states $\rho_2$, in most error regimes. This is not true \cref{thm:rate_largedev_lo,thm:rate_zero} since, as we shall see in \cref{subsec:rates}, these proofs explicitly leverage non-thermal operations when dealing with energy coherent states. Nevertheless, in \cref{app:thermal} we show how we can extract not-necessarily-tight bounds on the achievable rates under thermal operations in these regimes. 

Moreover, one can relatively straightforwardly generalise these results to obtain work-assisted optimal transformation rates. In this case, work is either invested to increase the rate of transformation, or extracted for the price of decreasing the rate. 
More precisely, consider an ancillary battery system $W$ with energy levels $\ket{0}_W$ and $\ket{1}_W$ separated by an energy gap $w$ \cite{horodecki2013fundamental,brandao2013resource,Brand_o_2015,lipka2021second,Lobejko2022workfluctuations}. Then, we say that there exists a $w$-assisted thermal operation transforming $\rho_1$ into a state $\epsilon$-close to $\rho_2$ if
\begin{equation}
    \rho_1\otimes \ketbra{0}{0}_W \xrightarrow[\mathrm{TO}]{\epsilon} \rho_2\otimes \ketbra{1}{1}_W,
\end{equation}
where $w>0$ corresponds to work extraction, whereas $w<0$ means work investment. As we show in \cref{app:battery}, one can modify the proof of \cref{thm:rate_smalldev} and arrive at the following result (note that analogous modifications of \cref{thm:rate_moddev,thm:rate_largedev_lo} are also possible).

\begin{restatable}[Optimal work-assisted rate in the small deviation regime]{theorem}{thmsmallwork}
    \label{thm:small_work}
    Consider a battery system with an energy gap
    \begin{equation}
        w = w_1 n + w_2 \sqrt{n},
    \end{equation}
    with constant $w_1$ and $w_2$. Then, for any fixed transformation error $\epsilon\in(0,1)$, the optimal rate $R_n^*$ for $w$-assisted thermodynamic transformation between $n$ copies of a generic state $\rho_1$ and $R_n^*n$ copies of an energy-incoherent state $\rho_2$ is given by
    \begin{align}
        \label{eq:small_deviation_work}
        R^*_n(\epsilon)\simeq\frac{D(\rho_1\|\gamma_1)-\beta w_1}{D(\rho_2\|\gamma_2)}+\frac{\sqrt{V(\rho_1\|\gamma_1)}S_{1/\xi'}^{-1}(\epsilon)-\beta w_2}{\sqrt{n}D(\rho_2\|\gamma_2)}
        ,
    \end{align}
    where 
    \begin{align}
    \xi':=\frac{V(\rho_1\|\gamma_1)}{D(\rho_1\|\gamma_1)-\beta w_1} \bigg/ \frac{V(\rho_2\|\gamma_2)}{D(\rho_2\|\gamma_2)},
\end{align}
    and $\simeq$ denotes an equality up to terms of order $o(1/\sqrt n)$. Moreover, when $\rho_2=\gamma_2$, any positive transformation rate $R_n^*$ is possible as long as
    \begin{equation}
        \frac{\beta w}{n} \lesssim D(\rho_1\|\gamma_1)+\sqrt{\frac{V(\rho_1\|\gamma_1)}{n}}\Phi^{-1}(\epsilon).
    \end{equation}
\end{restatable}


\subsection{Entanglement transformations}
\label{subsec:res_ent}
Due to the relation between transforming commuting quantum dichotomies and LOCC transformations discussed in \cref{sec:frame_ent}, our technical results also find applications in the resource theory of entanglement.

\begin{restatable}{theorem}{thmentanglement}
    \label{thm:entanglement}
    For pure bipartite states $\psi_1$ and $\psi_2$ characterised by Schmidt vectors $\v{p_1}$ and $\v{p_2}$, the optimal transformation rates between $\psi_1$ and $\psi_2$ under LOCC are captured by \cref{\ratetheorems} with the following substitutions (including the substitutions in the expression for $\xi$):
    \begin{aligns}
        D(\rho_i\|\sigma_i)&\rightarrow H(\v{p_i}),\\
        V(\rho_i\|\sigma_i)&\rightarrow V(\v{p_i}),\\
        \Dp_t(\rho_i\|\sigma_i),~ \Dm_t(\rho_i\|\sigma_i)&\rightarrow H_t(\v{p_i}).
    \end{aligns}
\end{restatable}
The details of necessary manipulations to arrive at the above result can be found in \cref{app:entanglement}.


\section{Discussion and applications}
\label{sec:app}

\kk{
\subsection{Phenomenological model}

We start the discussion by giving an intuitive, but completely non-rigorous, ``derivation'' of the thermodynamic small deviation rate (rates in other regimes can potentially also be ``derived'' in a similar fashion). It is based on three assumptions. First, assume that the thermodynamic resource content of a given state $\rho$ is a random variable $\log\rho-\log\gamma$ (a difference between log-likelihoods for the state and the thermal state), so that its mean and variance are given by the non-equilibrium free energy $D(\rho\|\gamma)$ and its fluctuations $V(\rho\|\gamma)$. Second, the distribution of the thermodynamic resource content of $\rho^{\otimes n}$ for large $n$ is a Gaussian with mean $nD(\rho\|\gamma)$ and variance $nV(\rho\|\gamma)$. And third, assume that every transformation that does not increase the resource content, even probabilistically, is allowed.

Using these three assumptions, let us now find the smallest transformation error $\epsilon$ for which a thermodynamic transformation with a rate $R_n^*$ from the initial state $\rho_1^{\otimes n}$ to the target state $\rho_2^{\otimes R_n^*n}$ is possible. Cumulative distribution functions of the resource content of the initial and target states are given by $\Phi_{\mu_1,\nu_1}$ and $\Phi_{\mu_2,\nu_2}$, where
\begin{subequations}
\begin{align}
    &\mu_1=nD(\rho_1\|\gamma),\quad \mu_2=R_n^*nD(\rho_2\|\gamma),\\
    &\nu_1=nV(\rho_1\|\gamma),\quad \nu_2=R_n^*nV(\rho_2\|\gamma).
\end{align}
\end{subequations}
Let us also denote the cumulative distribution of the resource content of the final state by $A$. Then, the condition for non-increasing the resource content is given by $A\geq \Phi_{\mu_1,\nu_1}$ (i.e., there is always more probability mass with lower resource content for the final state as compared to the initial state). The minimal transformation error is then given by
\begin{equation}
    \label{eq:error_pheno}
    \!\!\epsilon = \!\!\inf_{A\geq \Phi_{\mu_1,\nu_1}}\!\! \delta(A,\Phi_{\mu_2,\nu_2})= \inf_{A\geq \Phi} \delta(A,\Phi_{\mu,\nu})= S_\nu(\mu),
\end{equation}
where
\begin{equation}
    \mu=\frac{\mu_2-\mu_1}{\sqrt{\nu_1}},\quad \nu = \frac{\nu_2}{\nu_1}.
\end{equation}

Finally, by applying $S^{-1}_\nu$ to both sides of Eq.~\eqref{eq:error_pheno}, using the expressions for $\mu$ and $\nu$, and keeping only the leading terms in $n$, we end up recovering the thermodynamic small deviation rate:
\begin{equation}
    R_n^*(\epsilon) =\frac{D(\rho_1\|\gamma)+ \sqrt{V(\rho_1\|\gamma)/n}\cdot S^{-1}_{1/\xi}(\epsilon)}{D(\rho_2\|\gamma)},
\end{equation}
where $\xi$ is given by Eq.~\eqref{eq:reversibility} with $\sigma_1=\sigma_2=\gamma$.

}


\subsection{Optimal thermodynamic protocols with coherent inputs}
\label{subsec:opt_thermo_prot}

The obtained results can be straightforwardly applied to study the optimal performance of thermodynamic protocols, where the processed systems may be initially prepared in coherent superpositions of different energy eigenstates. In what follows, we will briefly discuss how this can be done and what it means for work extraction, information erasure, and thermodynamically free encoding of information. We note that in all these protocols the final states are energy-incoherent, and thus our results allow one to study them in full generality.

In the work extraction protocol, one uses a thermal bath and $n$ copies of a system in a state $\rho$ to excite the battery system $W$ over the energy gap $w$. The aim is to find the largest possible $w$ as a function of the allowed transformation error $\epsilon$. In other words, one wants to find the largest $w$ for which the following thermodynamic transformation exists:
\begin{equation}
    \rho^{\otimes n}\otimes \ketbra{0}{0}_W \xrightarrow[\mathrm{TO}]{\epsilon} \ketbra{1}{1}_W.
\end{equation}
This problem can be directly addressed by employing \cref{thm:small_work} with the target state of the system being thermal, which results in
    \begin{equation}
        \frac{w}{n} \leq \frac{1}{\beta}\left(D(\rho\|\gamma)+\sqrt{\frac{V(\rho\|\gamma)}{n}} \Phi^{-1}(\epsilon)\right).
    \end{equation}
The above yields the optimal amount of $\epsilon$-deterministic work that can be extracted per one copy of the processed system, and generalises the previous small deviation results on work extraction from incoherent states~\cite{chubb2018beyond} and pure states~\cite{biswas2022fluctuation} to general quantum states. 

In the information erasure protocol, one aims at using a thermal bath and the de-excitation of a battery system $W$ with minimal possible energy gap $|w|$ to reset $n$ copies of a system with a trivial Hamiltonian and in a state $\rho$ into a pure state $\ketbra{0}{0}$. This corresponds to finding the smallest $|w|$ (note that, since we de-excite the battery, 
we have $w<0$)
for which the following thermodynamic transformation exists:
\begin{equation}
    \rho^{\otimes n}\otimes \ketbra{0}{0}_W \xrightarrow[\mathrm{TO}]{\epsilon} \ketbra{0}{0}^{\otimes n} \otimes \ketbra{1}{1}_W.
\end{equation}
 Employing \cref{thm:small_work} and solving for $w$ that allows one to achieve rate 1, one arrives at the thermodynamic cost of information erasure per one copy of the system:
\begin{equation}
    \frac{|w|} {n}\simeq \frac{1}{\beta}\left(S(\rho)-\sqrt{\frac{V(\rho)}{n}}\Phi^{-1}(\epsilon)\right),
\end{equation}
which again generalises the previously known results for erasing incoherent states.

Finally, the problem of thermodynamically free encoding of information, introduced in Ref.~\cite{korzekwa2022encoding} and studied for incoherent and pure states in Ref.~\cite{biswas2022fluctuation}, is stated as follows. A sender is given $n$ copies of a quantum system $\rho$ that act as an information carrier, and wants to encode one of $M$ messages into these systems without using any thermodynamic resources, so employing only thermal operations. The aim is to find the maximal number $M$ of messages that can be encoded in a way that allows for decoding them with error probability at most $\epsilon$. In Ref.~\cite{korzekwa2022encoding} it was shown that, in the small deviation regime, $M$ is upper bounded by
\begin{equation}
    \label{eq:encoding}
    \frac{\log M(\rho^{\otimes n},\epsilon)}{n}\lesssim D(\rho\|\gamma)+\sqrt{\frac{V(\rho\|\gamma)}{n}}\Phi^{-1}(\epsilon),
\end{equation}
and in Ref.~\cite{biswas2022fluctuation} it was proved that the above bound can be achieved for states $\rho$ that are either energy-incoherent or pure. Using the results we obtained here, this can be generalised to arbitrary quantum states $\rho$ in the following way. Consider the following thermodynamic transformation:
\begin{equation}
    \rho^{\otimes n} \xrightarrow[\mathrm{TO}]{\epsilon} \ketbra{0}{0}_A^{\otimes R n},
\end{equation}
where the final system $A$ consists of $Rn$ two-level subsystems with trivial Hamiltonians. Note that since all energy levels of the final systems are degenerate, the sender can map the state $\ketbra{0}{0}_A^{\otimes R n}$ to any of $2^{Rn}$ basis states using thermal operations. Thus, the sender can encode $M=2^{Rn}$ messages that, moreover, can be decoded with probability of error $\epsilon$ simply through a computational basis measurement. It is then straightforward to employ \cref{thm:rate_smalldev,thm:thermo} to obtain \cref{eq:encoding} with $\lesssim$ replaced by~$\simeq$.

\begin{figure*}[t]
    \centering  
    \subfloat[]{
        \includegraphics[width=0.8\columnwidth]{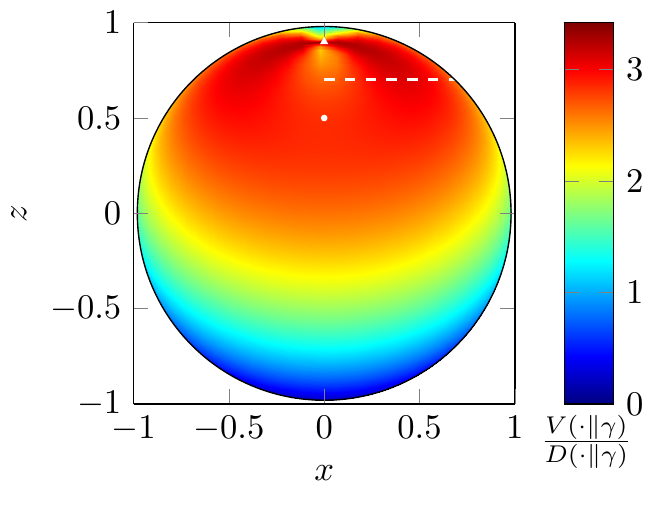}
        \label{subfig:resonance_a}
    }
    \subfloat[]{
        \includegraphics[width=\columnwidth]{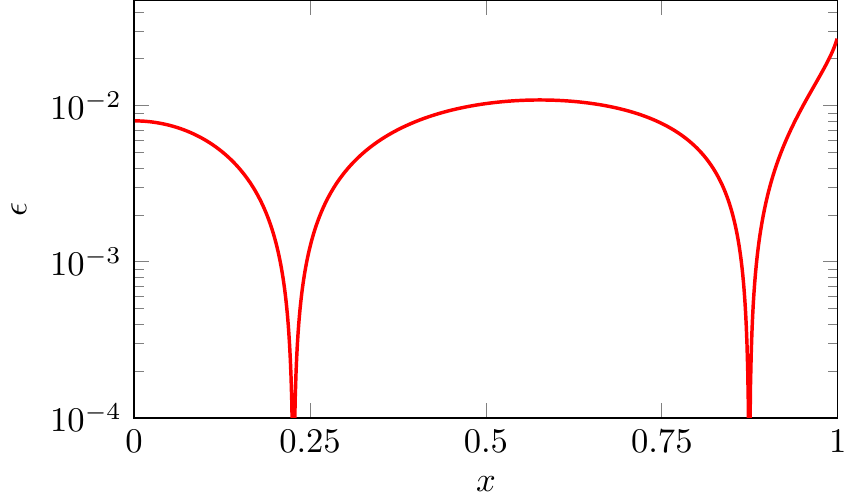}
        \label{subfig:resonance_b}
    }
    \caption{\textbf{Coherent resonance in thermodynamic transformations of two-level systems.} (a)~The ratio $V(\rho\|\gamma)/D(\rho\|\gamma)$ (encoding the resonance condition) for qubit states lying in the $xz$ plane of the Bloch sphere for a thermal state \mbox{$\gamma=\mathrm{diag}(0.95,0.05)$} (indicated by a white triangle). The white disk corresponds to the final state $\rho_2=\mathrm{diag}(0.75,0.25)$, while the dashed white line indicates a family of initial states $\rho_1(x)$ with diagonal $(0.85,0.15)$ and off-diagonal elements equal to $\sqrt{0.85\cdot 0.15}\cdot x$ for $x\in[0,1]$. (b)~Threshold transformation error $\epsilon$ required to achieve the asymptotic transformation rate $D(\rho_1(x)\|\gamma)/D(\rho_2\|\gamma)$ for finite number $n$ of transformed systems (i.e., $\epsilon$ such that the second-order correction term in Eq.~\eqref{eq:small_deviation} disappears). Resonance is obtained when the relative free energy fluctuations $V/D$ are the same for the initial state $\rho_1(x)$ and the final state $\rho_2$, i.e., when $\xi=1$. }
    \label{fig:resonance}
\end{figure*}


\subsection{Resonance phenomena}
\label{subsec:res}

One of the fundamental observations in the resource theory of thermodynamics is that all state transformations become reversible in the asymptotic limit~\cite{brandao2013resource}. Indeed, \cref{thm:rate_first} clearly states that for $n \rightarrow \infty$ the conversion rates $R$ and $R'$ for transformations $\rho_1^{\ot n}\rightarrow \rho_2^{\ot R n}$ and $\rho_2^{\ot n} \rightarrow \rho_1^{\ot R' n}$ become inversely proportional to each other, $R = 1/R'$. This is generally no longer true when we move outside of the idealised asymptotic scenario with $n\to\infty$. For example, in the thermodynamic protocols analysed in the previous section we have seen the deteriorating effect of finite-size transformations, i.e., due to the finite number of thermodynamically processed systems, the transformations are irreversible and lead to free energy dissipation that is related to the free energy fluctuations measured by $V(\rho\|\gamma)$~\cite{biswas2022fluctuation}. As a result, the performance of small quantum thermal machines may be seriously limited. Similar behaviour can be observed in the resource theory of pure-state entanglement or coherence.

Interestingly, it was recently found that these finite-size effects can be significantly mitigated by carefully engineering the resource conversion process~\cite{korzekwa2019avoiding}. More specifically, by appropriately tuning the initial and final states, so that the reversibility parameter $\xi=1$, the second-order correction to the optimal rate may vanish in the limit of zero transformation error. Thus, up to higher order terms, reversibility is restored. This intriguing phenomenon, termed resource resonance, was first predicted in Ref.~\cite{kumagai2016second} for pure-state entanglement transformations and then generalised to thermodynamic transformations between energy-incoherent states in Ref.~\cite{korzekwa2019avoiding}. The results we presented in this paper, allow us to extend the resource resonance phenomenon in three novel ways that we will now discuss.


\subsubsection{Coherent resonance}

For simplicity, let us focus on thermodynamic transformations between $n$ copies of a two-level system in a general state $\rho_1$ and $Rn$ copies of a two-level system in an energy-incoherent state $\rho_2$, assuming that the thermal Gibbs state $\gamma$ is the same for initial and final systems. Using \cref{thm:rate_smalldev} together with \cref{thm:thermo}, we get that the optimal transformation error $\epsilon$ for the asymptotic rate $R=D(\rho_1\|\gamma)/D(\rho_2\|\gamma)$ (i.e., avoiding dissipation) is given by
\begin{equation}
    \epsilon = S_{1/\xi} (0),
\end{equation}
which vanishes for $\xi=1$ and increases from 0 to 1/2 for $\xi>1$ and $\xi<1$. Without loss of generality, let us parameterise the initial and final states in the energy eigenbasis by:
\begin{equation}
    \!\!\!\rho_1(x)\!=\!\begin{pmatrix}
        p&\! x\sqrt{p(1\!-\!p)}\\
        x\sqrt{p(1\!-\!p)}&\! 1\!-\!p
    \end{pmatrix},~
    \rho_2\!=\!\begin{pmatrix}
        q&0\\
        0&1\!-\!q
    \end{pmatrix}\!,\!
\end{equation}
with $p,q,x\in[0,1]$. Then, for a fixed $p$ and $q$ (and given~$\gamma$), we can consider a family of initial states parameterised by $x$ (see \cref{subfig:resonance_a}). This corresponds to probabilistic mixtures of an energy-incoherent state \mbox{$p\ketbra{0}{0}+(1-p)\ketbra{1}{1}$} and a pure coherent superposition of energy eigenstates $\sqrt{p}\ket{0}+\sqrt{1-p}\ket{1}$, so that $x\in[0,1]$ smoothly connects between completely incoherent and completely coherent initial states. In \cref{subfig:resonance_b} we present the non-trivial dependence of the transformation error $\epsilon$ on the coherence level $x$, where we can observe two resonant values of $x$ for which error-free and dissipationless transformations (up to second-order asymptotics) are possible. This clearly illustrates that quantum coherence can play an important role in avoiding free energy dissipation in thermodynamic transformations of quantum states.


\subsubsection{Work-assisted resonance}

Looking at the optimal work-assisted rate from \cref{thm:small_work}, we see that the whenever $\xi'=1$, one can choose $w_2=0$ to make the second-order asymptotic correction vanish for zero transformation error $\epsilon$. Crucially, the value of $\xi'$ can be controlled by the amount $w_1$ of invested (or extracted) work per one copy of the system. By choosing
\begin{equation}
    w_1 = \frac{1}{\beta} \left( D(\rho_1\|\gamma_1)-\frac{V(\rho_1\|\gamma_1)}{V(\rho_2\|\gamma_2)}D(\rho_2\|\gamma_2)\right),
\end{equation}
which results in the optimal rate given by
\begin{equation}
    \label{eq:rate_work_rev}
    R=\frac{V(\rho_1\|\gamma_1)}{V(\rho_2\|\gamma_2)},
\end{equation}
one can perform an error-free and dissipationless transformation. In other words, the total initial state of the system and battery gets transformed to the total final state with zero-error and equal free energy content (up to second-order terms). This thus opens a way for bringing two states into resonance by investing or extracting work.

The work-assisted resonance can be understood by first noticing that the resonance condition can be seen as requiring the total fluctuations of the initial system to be equal to the total fluctuations of the final system, up to first-order in $n$. Without work-assistance this means that
\begin{equation}
    \label{eq:equal_fluct}
    V(\rho_1^{\ot n}\|\gamma_1^{\ot n}) = V(\rho_2^{\ot Rn}\|\gamma_2^{\ot Rn})
\end{equation}
and given the asymptotic value of $R$ it yields
\kk{\begin{equation}
    V(\rho_1\|\gamma_1) = \frac{D(\rho_1\|\gamma_1)}{D(\rho_2\|\gamma_2)}V(\rho_2\|\gamma_2),
\end{equation}}
which is exactly the original resonance condition $\xi=1$. Now, bringing the battery system does not change fluctuations (since at the initial and final time the battery is in a pure energy eigenstate with zero fluctuations), but affects the rate $R$. The work-assisted resonance condition is achieved by increasing or decreasing $R$ through an appropriate choice of $w$, so that \cref{eq:equal_fluct} is satisfied, which happens for $R$ given by \cref{eq:rate_work_rev}.


\subsubsection{Strong resonance}
\label{subsubsec:strong_res}

In Ref.~\cite{korzekwa2019avoiding} a resonance phenomenon was observed for transformations operating at the first-order asymptotic rate. Specifically, such transformations generically incur a constant error, but it was shown that if a resonance condition is met these errors are in fact exponentially suppressed. That result was built upon the small and moderate deviation results of Refs.~\cite{kumagai2016second,chubb2018beyond,chubb2019moderate}, but large and extreme deviation analyses had not been performed at the time that would allow for exponentially small errors to be analysed. By extending to large and extreme deviation analyses in this paper, it can in fact be seen that there exists an even stronger notion of resonance, which we term \emph{strong resonance}, in contrast to the \emph{weak resonance} of Ref.~\cite{korzekwa2019avoiding}, in which errors are not just exponentially suppressed, but entirely eliminated.

An illustration of weak and strong resonance is presented in \cref{fig:resonance}. Weak resonance corresponds to the second-order corrections in the small and moderate deviation rates, \cref{thm:rate_smalldev,thm:rate_moddev}, vanish, and occurs when
\begin{align}
    \frac{V\reli{1}}{D\reli{1}}=\frac{V\reli{2}}{D\reli{2}}.
\end{align}
Strong resonance corresponds to the situation in which the large and extreme deviation rates, \cref{thm:rate_largedev_lo,thm:rate_zero}, also collapse down to the first-order rate, in other words when
\begin{align}
    \operatorname*{arg\,min}\limits_{\alpha\in\overline{\mathbb R}}\frac{\Dm_\alpha\reli{1}}{D_\alpha\reli{2}}=1.
\end{align}

We present a numerical example of a set of states that exhibit both strong and weak resonance in \cref{app:res}, and discuss the relationship between weak and strong resonance in \cref{app:consistency}.

\begin{figure}[t]
    \centering
    \includegraphics[width=\linewidth]{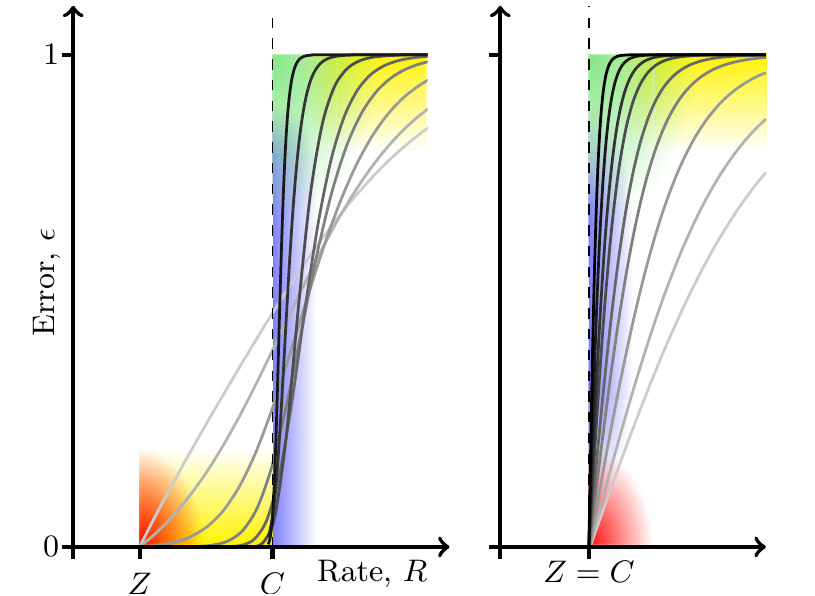}

    \caption{\textbf{Weak and strong resonance phenomena.} Left:~Weak resonance, in which the small and moderate regimes at rates $R<C$ collapse, but the large and extreme regimes persist, i.e.\ $Z<C$. Right:~Strong resonance, in which \emph{all} error regimes at rates below $C$ collapse, i.e.\ $Z=C$. See \cref{fig:sigmoid_rate} for an explanation of the various error regimes indicated, as well as the definitions of $Z$ and $C$.}

    \label{fig:res}
\end{figure}


\subsection{Entanglement transformations}
\label{subsec:ent}

Let us now make a few brief comments on \cref{thm:entanglement}. It is very important to note that related results have previously appeared in the literature. First, Ref.~\cite{kumagai2016second} derived the optimal second-order rates for pure-state bipartite entanglement transformations in the small-deviation regime using infidelity to measure transformation error. Later, Ref.~\cite{chubb2019moderate} extended these results to the moderate deviation regime. Finally, Ref.~\cite{Jensen_2019} investigated exact asymptotic transformations and derived optimal rates in the zero-error regime.  

Our work differs from these results in three ways. First, we extend the analysis to the previously unaddressed large deviation regime. This allows us to, e.g., predict a strong resonance phenomena for entanglement transformations. Second, our results hold for a different error measure (trace distance instead of infidelity). And third, probably most importantly, we propose a novel methodology employing dichotomies and hypothesis testing that allows us to easily characterise asymptotic rates in a unified manner across various error regimes, and avoid many arduous subtleties along the way. We believe that this approach brings a significant simplification and clarity as compared to the previous techniques. 

On the flip side, we need to mention that while infidelity measure between the Schmidt vectors $\v{p}_1$ and $\v{p}_2$ has a clear operational meaning (since it is precisely the infidelity between the corresponding entangled states $\psi_1$ and $\psi_2$), the use of the trace distance may be less useful. Still, one can directly relate the trace distance $\delta$ between $\v{p}_1$ and $\v{p}_2$ to the probability $P$ of distinguishing bipartite entangled states $\psi_1$ and $\psi_2$ locally by one party:~\mbox{$P=(1+\delta)/2$}.

Finally, we recall that it was proven in Ref.~\cite{du2015conditions} that the pure-state transformation laws in the resource theory of coherence~\cite{baumgratz2014quantifying}, i.e., conditions under which pure superpositions of distinguished basis states can be mapped to each other under incoherent operations, are also characterised by the majorisation relation. Thus, \cref{thm:entanglement} can be straightforwardly applied to describe optimal rates for pure-state coherence transformations (simply $\v{p}_1$ and $\v{p}_2$ need to represent occupations of the initial and target states in the distinguished basis).


\section{Derivations}
\label{sec:derive}

In this section we will give proofs of our results on the asymptotic analysis of the transformation rates between quantum dichotomies in several different error regimes. We will break this analysis down into three stages. In \cref{subsec:ht} we will review the relationship between Blackwell ordering and hypothesis testing, generalising  the existing analysis beyond the fully commuting case to allow for results where the input dichotomy is non-commuting, and partial results when the target dichotomy is also non-commuting. Critically, once established, this connection allows us to rather straightforwardly extend the existing asymptotic analyses of hypothesis testing to transformation rates between quantum dichotomies. In \cref{subsec:htass} we review the existing results around hypothesis testing, with some necessary technical extensions. Finally in \cref{subsec:rates} we put everything together, giving the final proofs of transformation rates in each error regime.


\subsection{Hypothesis testing and pinched hypothesis testing}
\label{subsec:ht}

The data processing inequality ensures that the approximate Blackwell ordering \mbox{$(\rho_1,\sigma_1)\succeq_{(\epsilon_\rho,\epsilon_\sigma)}(\rho_2,\sigma_2)$} implies
\begin{align}
    \beta_x\reli{1}\leq \beta_{x-\epsilon_\rho}\reli{2}+\epsilon_\sigma\quad\forall x\in(\epsilon_\rho,1).
\end{align}
Extending this, it is shown in Ref.~\cite{renes2016relative} that the two conditions are in fact equivalent for commuting states \mbox{$[\rho_1,\sigma_1]=[\rho_2,\sigma_2]=0$}. Thus, for such states the analysis of transformation rates can be entirely reduced to the analysis of hypothesis testing. Unfortunately, the situation for non-commuting quantum states is not so clean: it is known that such a hypothesis testing condition is \emph{not} generally a sufficient condition for Blackwell ordering~\cite{Matsumoto2014,jenvcova2012comparison,reeb2011hilbert,buscemi2012comparison}. 

While there is no known sufficient condition that can be phrased in terms of regular hypothesis testing, we instead consider a modified task we call \emph{pinched hypothesis testing}, which does provide such a sufficient condition for non-commuting input states. This condition does not, however, extend to non-commuting target states, and we leave this for future work. 

We will use $\pinch\cdot\tau$ to denote the pinching with respect to the eigenspaces of $\tau$. Specifically, it is defined by
\begin{align}
    \pinch X\tau:=\sum_\lambda \Pi_\lambda X\Pi_\lambda,
\end{align}
where $\Pi_\lambda$ are the eigenspace projectors of $\tau$, i.e.,\ \mbox{$\tau=\sum_\lambda \lambda\Pi_\lambda$}. The task of \emph{pinched hypothesis testing} is to distinguish between the states $\pinch{\rho}{\sigma}$ and $\sigma$, or between $\rho$ and $\pinch{\sigma}{\rho}$. Correspondingly, we define the \emph{left-pinched} and \emph{right-pinched} type-II hypothesis testing error as  
\begin{subequations}
\label{eq:pinched_beta}
    \begin{align}
    \betaL_x\rel{\rho}{\sigma}
    :=&
    \beta_x\rel{\pinch{\rho}{\sigma}}{\sigma},\\
    \betaR_x\rel{\rho}{\sigma}
    :=&
    \beta_x\rel\rho{\pinch\sigma\rho}.
    \end{align}
\end{subequations}
By the data-processing inequality we know that pinching cannot make states easier to distinguish, and thus the pinched error is at least the non-pinched error,
\begin{align}
\betaL_x\rel{\rho}{\sigma},\betaR_x\rel{\rho}{\sigma}\geq \beta_x\rel{\rho}{\sigma},
\end{align}
with equality if $\left[\rho,\sigma\right]=0$. Having defined the pinched error, we now show how it can be used to construct a \emph{sufficient} condition for non-commuting Blackwell ordering.

\begin{lemma}[Conditions for approximate Blackwell ordering for commuting second dichotomy]
    \label{lem:ht}
    Consider the approximate Blackwell ordering \mbox{$(\rho_1,\sigma_1)\succeq_{(\epsilon_\rho,\epsilon_\sigma)}(\rho_2,\sigma_2)$}. A \emph{necessary} condition for this ordering is given by
    \begin{align}
        \beta_x\reli{1}\leq \beta_{x-\epsilon_\rho}\reli{2}+\epsilon_\sigma\quad\forall x\in(\epsilon_\rho,1).
    \end{align}
    If the second dichotomy is commuting, $[\rho_2,\sigma_2]=0$, then a sufficient condition for this ordering is given by either
    \begin{subequations}
    \begin{align}
        \label{eq:betatildeordering_left}
        \betaL_x\reli{1}\leq \beta_{x-\epsilon_\rho}\reli{2}+\epsilon_\sigma\quad\forall x\in(\epsilon_\rho,1),
    \end{align}
    or by 
    \begin{align}
        \label{eq:betatildeordering_right}
        \betaR_x\reli{1}\leq \beta_{x-\epsilon_\rho}\reli{2}+\epsilon_\sigma\quad\forall x\in(\epsilon_\rho,1).
    \end{align}
    \end{subequations}
\end{lemma}
\begin{proof}
    As noted above, the necessary condition simply follows from the data processing inequality~\cite{renes2016relative}, so we need only prove the sufficient condition. We start by assuming that the pinched hypothesis testing inequality, \cref{eq:betatildeordering_left}, holds. Expanding out the definition of \mbox{$\betaL_x(\cdot\|\cdot)$}, this is equivalent to
    \begin{align}
       \!\!\! \beta_x\rel{\pinch{\rho_1}{\sigma_1}}{\sigma_1}\leq \beta_{x-\epsilon_\rho}\reli{2}+\epsilon_\sigma\quad\forall x\in(\epsilon_\rho,1).
    \end{align}
    Pinching a state causes it to commute, in the sense that $[\pinch{\cdot}{\sigma},\sigma]\equiv 0$. As such, the first dichotomy $(\pinch{\rho_1}{\sigma_1},\sigma_1)$ is commuting, and the second dichotomy $(\rho_2,\sigma_2)$ is also commuting by assumption. Applying Ref.~\cite[Thm.~2]{renes2016relative}, this in turn implies the Blackwell ordering on the pinched states,
    \begin{align}
        \left(\pinch{\rho_1}{\sigma_1},\sigma_1\right)\succeq_{(\epsilon_\rho,\epsilon_\sigma)}(\rho_2,\sigma_2).
    \end{align}
    Next, we want to argue that approximate Blackwell ordering has a data-processing property. By definition this ordering implies the existence of a channel $\mathcal E$ such that
    \begin{align}
        T\Bigl(\mathcal E\bigl(\pinch{\rho_1}{\sigma_1}\bigr),\rho_2\Bigr)\leq \epsilon_\rho
        ~~\text{and}~~
        T\left(\mathcal E(\sigma_1),\sigma_2\right)\leq \epsilon_\sigma.
    \end{align}
    If we define $\lefthat{\mathcal E}:=\mathcal E\circ \mathcal P_{\sigma_1}$, and recall that $\pinch{\sigma_1}{\sigma_1}=\sigma_1$, then these expressions can be rewritten
    \begin{align}
        T\left(\lefthat{\mathcal E}(\rho_1),\rho_2\right)\leq \epsilon_\rho
        ~~\text{and}~~
        T\left(\lefthat{\mathcal E}(\sigma_1),\sigma_2\right)\leq \epsilon_\sigma,
    \end{align}
    which in turn implies the required Blackwell ordering of the two dichotomies. For \cref{eq:betatildeordering_right} a similar argument can be given, with $\righthat{\mathcal E}:=\mathcal E\circ \mathcal P_{\rho_1}$.
\end{proof}

We now have both necessary and sufficient conditions for approximate Blackwell ordering of quantum dichotomies that are of the same form. Unlike the commuting case, however, these two conditions are no longer identical involving the pinched and non-pinched variants of hypothesis testing. As such, this will generally open up a gap between the upper and lower bounds that can be placed upon transformation rates using this technique, which makes this approach unsuitable in the single-shot setting. However, as we will see later in this section,  in the asymptotic setting the pinched and non-pinched variants of hypothesis testing have identical asymptotic behaviour in most error regimes, closing these gaps and allowing us to give optimal expressions of transformation rates beyond the first-order asymptotics. 


\subsection{Asymptotic analyses of hypothesis testing}
\label{subsec:htass}

In this subsection we want to review the relevant asymptotic analyses of hypothesis testing, putting these results into a common notation for easier use later, as well as extending these analyses to the pinched variant of the task where necessary. The cornerstone of asymptotic analysis of hypothesis testing is Stein's Lemma. While sufficient to give a first-order analysis of transformation rates between quantum dichotomies, we will see that we require refinements upon Stein's Lemma to go beyond first-order. We will start this section by describing Stein's Lemma, and giving some intuition for the different regimes of refinements thereto. We will then go through each error regime reviewing the refined asymptotic analysis in each, extending these analyses to the pinched variant of the problem as necessary.

Consider the task of distinguishing between two states $\rho^{\otimes n}$ and $\sigma^{\otimes n}$. To avoid technical issues, we will assume that $\sigma$ is of full support. Intuitively, each additional copy of the states should give us a constant amount of new information, allowing us to multiplicatively reduce the chance of failing to distinguish the two states, leading to exponentially decreasing hypothesis testing errors. In general there is a trade-off between the type-I and -II errors. A natural simplification of this more general question would be the following: if we constrain one of our errors to be constant, how does the other error decay? The answer is that the error decays exponentially, with that exponent being given by the relative entropy. This fact is known as Stein's Lemma, and will form the backbone of this subsection.

\begin{lemma}[Quantum Stein's Lemma~\cite{HiaiPetz1991,OgawaNagaoka1999}]
    \label{lem:ht_stein}
    For any \mbox{$\epsilon\in(0,1)$}
    \begin{align}
        \lim_{n\to\infty}-\frac 1n \log \beta_\epsilon\rel{\rho^{\otimes n}}{\sigma^{\otimes n}}=D\reli{}.
    \end{align}
\end{lemma}

\begin{figure*}[t]
    \centering
    \includegraphics[width=\linewidth]{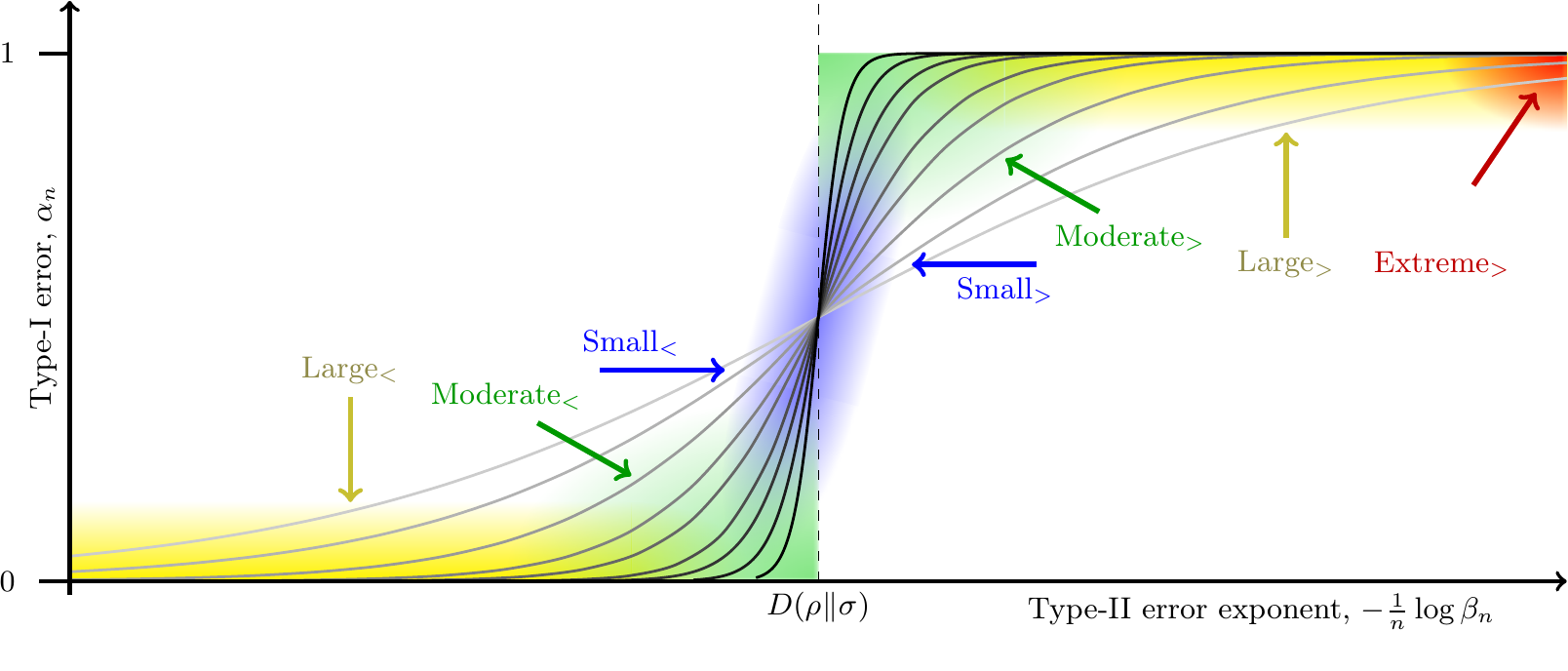}
    \def\yes{{\color{darkgreen}Yes}}
    \def\no{{\color{red}No}}
    \begin{tabular}{ccccccccc}
    \toprule
    \multicolumn{2}{c}{\bfseries Regime} & ~~~~ & \textbf{Type-I ($\alpha_n$)} & ~~~~ & \textbf{Type-II exponent ($-\frac 1n \log \beta_n$)} &~~&\textbf{Equal?} \\ \midrule
    Large$_<$ & \cref{lem:ht_largedev} && Approaching 0 exponentially && $\bigl[0,D(\rho\|\sigma)\bigr)$ &&\no\\
    Moderate$_<$ & \cref{lem:ht_moddev} && Approaching 0 subexponentially && $D(\rho\|\sigma)-\omega(1/\sqrt n)\cap o(1)$ &&\yes \\
    Small$_<$ & \cref{lem:ht_smalldev} && Constant, $<0.5$ && $D(\rho\|\sigma)-\Theta(1/\sqrt n)$ &&\yes \\ \midrule
    Small$_>$ & \cref{lem:ht_smalldev} && Constant, $>0.5$ && $D(\rho\|\sigma)+\Theta(1/\sqrt n)$ &&\yes \\
    Moderate$_>$ & \cref{lem:ht_moddev} && Approaching 1 subexponentially && $D(\rho\|\sigma)+\omega(1/\sqrt n)\cap o(1)$ &&\yes \\
    Large$_>$ & \cref{lem:ht_largedev} && Approaching 1 exponentialy && $\bigl(D(\rho\|\sigma),\infty\bigr)$ &&\yes \\
    Extreme$_>$ & \cref{lem:ht_extremedev} && Approaching 1 superexponentially && $\infty$ &&\yes \\ \bottomrule
    \end{tabular}

    \caption{\textbf{Trade-off between the optimal type-I and -II errors.} An illustrative sketch of the trade-off between the optimal type-I and -II errors of the hypothesis test between two states $\rho^{\otimes n}$ and $\sigma^{\otimes n}$ as $n$ grows. Here $\alpha_n$ is the optimal type-I error, $\beta_n$ the optimal type-II error, and $-\frac 1n \log \beta_n$ the type-II error exponent. Each of the grey curves correspond to a trade-off $(\alpha_n,\beta_n)$, for a given $n$, with darker curves correspond to growing $n$. The fact these curves approach a step at the relative entropy is equivalent to Stein's Lemma, \cref{lem:ht_stein}. Each of the coloured regions corresponds to a deviation regimes in which we will consider refinements to Stein's lemma in this subsection. In the table we present the scaling in each regime. For the details and explicit expressions for all of the scaling constants see the corresponding lemmas, \cref{lem:ht_smalldev,lem:ht_moddev,lem:ht_largedev,lem:ht_extremedev}. The final column denotes whether the asymptotics of the pinched and non-pinched variants of hypothesis testing are identical, which they are in all regimes by where both errors are exponentially decreasing.
    }

    \label{fig:sigmoid_ht}
\end{figure*}

As mentioned above, Stein's Lemma alone will only be sufficient to give first-order rates, and we will require more refined asymptotic analysis to go beyond this. In \cref{fig:sigmoid_ht} we present a sketch of the various error regimes we will consider. The idea is that as $n$ increases Stein's lemma states that the trade-off between the type-I error and type-II error exponent becomes a step at the relative entropy, and each of the refinements seek to quantify the rate of that convergence in different ways. Specifically, as shown in the table in \cref{fig:sigmoid_ht}, our analysis will be divided up based on the scaling of the type-I error considered. The two most important regimes will be the small deviation and large deviation regimes, in which the type-I error is a constant bounded away from 0/1, or exponentially approaching 0/1, respectively. This leaves us with two edge cases: the intermediate regime of subexponential decay is termed moderate deviation, and for completeness we also consider the regime in which the type-I error superexponentially approaches 1, which will be required for our analysis of transformation rates in the zero-error setting.

In \cref{app:consistency} we will non-rigorously discuss the interplay between these regimes and the consistency between these results, and in \cref{app:uni} we will shown how all of the below analyses can be strengthened to have a uniformity property which will be necessary in some of the proofs of transformations rates given in \cref{subsec:rates}.


\subsubsection{Small deviation}
\label{subsubsec:ht_small}

The first regime we consider is the small deviation. As stated in \cref{fig:sigmoid_ht}, in this regime the type-I error is a fixed constant between $0$ and $1$, and we want to know the asymptotic behaviour of the type-II error exponent. Stein's Lemma tells us that this exponent must approach the relative entropy, and the small deviation analysis, also known as the \emph{second-order expansion}, states that this convergence happens as $\Theta(1/\sqrt n)$.

\begin{lemma}[Small deviation analysis of hypothesis testing]
    \label{lem:ht_smalldev}
    For any constant $\epsilon\in(0,1)$, the hypothesis testing type-II errors scale as
    \begin{aligns}
        -\frac 1n\log\beta_\epsilon\rel{\rho^{\otimes n}}{\sigma^{\otimes n}}&\simeq D\rel{\rho}{\sigma}+\sqrt{\frac{V\rel{\rho}{\sigma}}n}\Phi^{-1}(\epsilon),\\
        -\frac 1n\log\betaL_\epsilon\rel{\rho^{\otimes n}}{\sigma^{\otimes n}}&\simeq D\rel{\rho}{\sigma}+\sqrt{\frac{V\rel{\rho}{\sigma}}n}\Phi^{-1}(\epsilon),
    \end{aligns}
    where $\simeq$ denotes equality up to terms $o(1/\sqrt n)$.
\end{lemma}
\begin{proof}
    The scaling of $\beta_\epsilon\rel{\rho^{\otimes n}}{\sigma^{\otimes n}}$ is directly a restatement of Ref.~\cite[Prop.~16]{TomamichelTan2015}, a result which originates in Ref.~\cite{li2014second,TomamichelHayashi2012}, so we are just left with showing the pinched variant $\betaL_\epsilon\rel{\rho^{\otimes n}}{\sigma^{\otimes n}}$ has the same scaling up to second-order. 
    
    Firstly, we note that $\betaL_\epsilon\rel{\cdot}{\cdot}\geq \beta_\epsilon\rel{\cdot}{\cdot}$, so the upper bound straightforwardly holds, leaving only the lower bound left to prove. For this we turn to Ref.~\cite{TomamichelHayashi2012}, specifically combining Equations (14,~20,~27) to give that, for any $0<\delta<\epsilon/3$,
    \begin{align}
        \betaL_\epsilon\rel{\rho}{\sigma}\leq \beta_{\epsilon-2\delta}\rel{\rho}{\sigma}\cdot \frac{2^8(\epsilon-\delta)\nu(\sigma)^2}{\delta^5(1-\epsilon+\delta)},
    \end{align}
    where $\nu(\sigma)$ denotes the number of unique eigenvalues of~$\sigma$. Notice that, for any finite dimensional $\sigma$, the number of eigenvalues of the tensor power $\sigma^{\otimes n}$ only scales polynomially, $\nu\left(\sigma^{\otimes n}\right)\leq n^{\nu(\sigma )}=n^{O(1)}$. Using this, we can now substitute $\rho\to\rho^{\otimes n}$ and $\sigma\to\sigma^{\otimes n}$, giving
    \begin{align}
        \log\betaL_\epsilon\rel{\rho^{\otimes n}}{\sigma^{\otimes n}}\leq \log\beta_{\epsilon-2\delta}\rel{\rho^{\otimes n}}{\sigma^{\otimes n}}+O(\log n).
    \end{align}
    Importantly, this logarithmic error is $o(\sqrt n)$ and can therefore be neglected to second-order. As such, we get the bound 
    \begin{aligns}
        -\frac 1n\log\betaL_\epsilon\rel{\rho^{\otimes n}}{\sigma^{\otimes n}}
        &\gtrsim -\frac 1n\log\beta_{\ctc\epsilon-2\delta}\rel{\rho^{\otimes n}}{\sigma^{\otimes n}},\\
        &\gtrsim D\reli{}
        +\sqrt{\frac{V\reli{}}n}\Phi^{-1}(\epsilon-2\delta).
    \end{aligns}
    As this holds for any $\delta\in (0,\epsilon/3)$, and $\Phi^{-1}$ is continuous on $(0,1)$, we can take $\delta\to 0^+$, giving 
    \begin{align}
       \!\!\! -\frac 1n\log\betaL_\epsilon\rel{\rho^{\otimes n}}{\sigma^{\otimes n}}
        &\!\gtrsim\! D\reli{}
        +\sqrt{\frac{V\reli{}}n}\Phi^{-1}(\epsilon),\!
    \end{align}
    as required.
\end{proof}


\subsubsection{Large deviation}
\label{subsubsec:ht_large}

The next most important error regime is that of large deviations, which is the regime in which both errors are exponentially approaching either 0 or 1. Stein's lemma suggests that as long as the type-II error exponent is less than the relative entropy, then the type-I error will also be exponentially decreasing with $n$; but if it exceeds the relative entropy, then we expect the type-I error to be exponentially increasing towards 1. 

The existing expressions of these results in the literature are all phrased in terms of these exponents directly. But, using this notation, the large deviation regime would need to be divided up into several different forms based on whether the errors are approaching 0 or 1, usually termed the \emph{error exponent} and \emph{strong converse exponent} regimes. Instead, we will combine all of these results in a single unified notation by concerning ourselves not with error \emph{probabilities}, but the error \emph{log odds}. This unified notation dramatically simplifies our later proofs which rely upon these bounds, and to our knowledge this formulation has not appeared elsewhere in the literature.

The idea to unify these regimes is to consider a `signed exponent'. For a quantity $p_n$ which is exponentially approaching 0, the exponent is given by $-\frac 1n \log p_n$, and similarly if $p_n$ is exponentially approaching 1 then the exponent is given by $-\frac 1n \log (1-p_n)$. The idea is to combine these two functions to give a single expression which can yield both exponents. Specifically, we will use the \emph{logit function} which is simply the difference between $\log p$ and $\log (1-p)$,
\begin{align}
L[p]:=\log\frac p{1-p}.
\end{align}
As required, now we can think of $\frac 1n L[p_n]$ as an exponent that covers both cases where $p_n$ is approaching 0 or 1 in the sign of this exponent. Specifically, {\ctc for any $\lambda>0$,}
\begin{aligns}
    \frac 1n L[p_n]&\to-\lambda & &\iff & -\frac 1n \log p_n&\to\lambda,\\
    \frac 1n L[p_n]&\to+\lambda & &\iff & -\frac 1n \log(1-p_n)&\to\lambda.
\end{aligns}
While there are other functions that have this property, one thing to note about the logit function specifically is that if $p$ is a probability then $L[p]$ is the associated \emph{log odds}. As we will see below, it turns out that the standard large deviation results can be more succinctly expressed in terms of the type-I and -II error \emph{log odds} instead of error \emph{probabilities}.

To allow us to express the large (and moderate) deviation results in terms of the \emph{log odds} we will define the optimal type-II log odds $\gamma_{\ctc x}(\rho\|\sigma)$---in analogy to the optimal type-II error probability $\beta_x\reli{}$---as the solution to the optimisation
\begin{aligns}
    \min_Q \quad & L\bigl[\Tr(\sigma Q)\bigr], \\
    \textnormal{subject to} \quad & 0 \leq Q \leq 1, \\
    & L\bigl[1-\Tr (\rho Q)\bigr] \leq x.
\end{aligns}
And as with the error probability we will also require the pinched variants, defined as
\begin{aligns}
    \gammaL_x(\rho\|\sigma):=\gamma_x\rel{\pinch{\rho}{\sigma}}{\sigma},\\
    \gammaR_x(\rho\|\sigma):=\gamma_x\rel\rho{\pinch\sigma\rho}.
\end{aligns}
In terms of error probabilities these definitions are equivalent to 
\begin{aligns}
    \gamma_x(\rho\|\sigma)&=L\left[ \beta_{L^{-1}[x]}(\rho\|\sigma) \right],\\
    \gammaL_x(\rho\|\sigma)&=L\left[ \betaL_{L^{-1}[x]}(\rho\|\sigma) \right],\\
    \gammaR_x(\rho\|\sigma)&=L\left[ \betaR_{L^{-1}[x]}(\rho\|\sigma) \right].
\end{aligns}

As argued above, a nice feature of this formulation is that we can describe all large deviation results in a single unified way. There are three different regimes of large deviation results, only two of which are captured in \cref{fig:sigmoid_ht}. Applying Stein's lemma, \cref{lem:ht_stein}, alongside a dual version where we swap the states, we can see that there is a regime in which both error probabilities decay exponentially, albeit with exponents no greater than the respectively relative entropies. If, however, one of the errors decays with an exponent greater than the relative entropy then the other error will in fact exponentially approach 1. This is illustrated in \cref{fig:logodds} using our log odds formulation.

\begin{figure}
    \includegraphics[width=\linewidth]{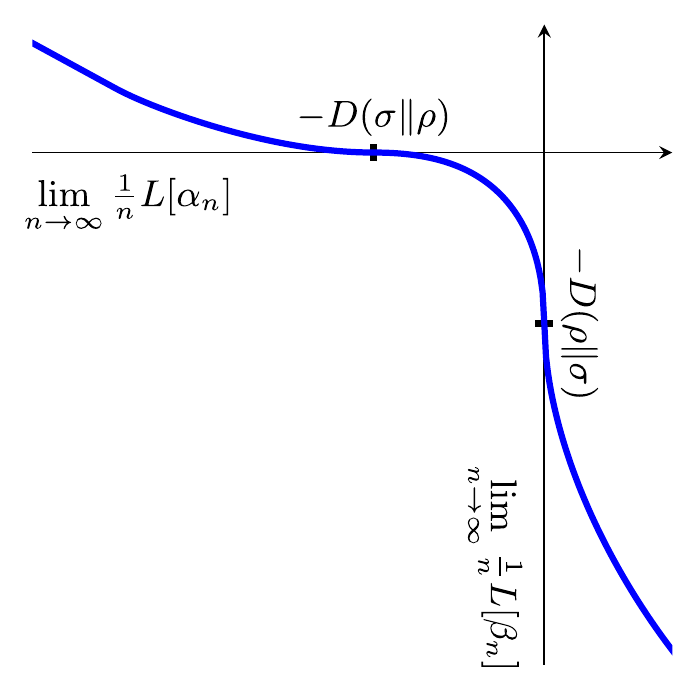}
    \caption{\textbf{Trade-off using log odds.} The trade-off between the type-I and -II error log odds per copy---$\lim\limits_{n\to\infty}\frac 1nL[\alpha_n]$ and $\lim\limits_{n\to\infty}\frac 1nL[\beta_n]$ respectively---for the hypothesis test between $\rho^{\otimes n}$ and $\sigma^{\otimes n}$, in the limit of growing $n$. The bottom-left quadrant corresponds to the regime in which both errors are decaying exponentially, with exponents bound by the relative entropies $D(\sigma\|\rho)$ and $D\reli{}$ respectively. The top-left regime corresponds to a type-I error which is decaying even more rapidly, causing the type-II error to instead increase towards 1, and the bottom-right the converse of this. This curve is generated by plotting $\Gamma_\lambda\reli{}$ from \cref{lem:ht_largedev} for two randomly generated $d=5$ qudit states.
    }
    \label{fig:logodds}
\end{figure}

Before we give the large deviation bound, we need several additional definitions that will be critical for the pinched case. Define the \emph{left-pinched} and \emph{right-pinched} R\'enyi relative entropies as
\begin{aligns}
    \label{eq:pinched_renyi_def1}
    \DL_\alpha\reli{} := \lim_{n\to\infty} \frac 1n D_\alpha\rel{\pinch{\rho^{\otimes n}}{\sigma^{\otimes n}}}{\sigma^{\otimes n}},\\
    \label{eq:pinched_renyi_def2}
    \DR_\alpha\reli{} := \lim_{n\to\infty} \frac 1n D_\alpha\rel{\rho^{\otimes n}}{\pinch{\sigma^{\otimes n}}{\rho^{\otimes n}}}.
\end{aligns}
We note that due to the duality property of the classical relative entropy $(1-\alpha)D_\alpha\rel pq=\alpha D_{1-\alpha}\rel qp$, we straightforwardly have
\begin{align}
(1-\alpha)\DL_\alpha\reli{}=\alpha\DR_{1-\alpha}\rel\sigma\rho.
\end{align}
We also note that for $\alpha\geq 0$ the left-pinched coincides with the sandwiched relative entropy and for $\alpha\leq 1$ the right-pinched coincides with the reverse sandwiched relative entropy~\cite{muller2013quantum,Tomamichel2016}, i.e.\ 
\begin{aligns}
    \DL_\alpha\reli{} &= \frac{1}{\alpha-1} \log\Tr\left(\left(\sqrt \rho\sigma^{\frac{1-\alpha}\alpha}\sqrt\rho\right)^{\alpha}\right)~\,~\forall\alpha\geq 0,\\
    \DR_\alpha\reli{} &= \frac{1}{\alpha-1} \log\Tr\left(\left(\sqrt \sigma\rho^{\frac\alpha{1-\alpha}}\sqrt\sigma\right)^{1-\alpha}\right)~\forall\alpha\leq 1,
\end{aligns}
but we know of no closed-form solution for either outside of these ranges. In \cref{app:pinch} we show the existence, and some important properties, of these relative entropies. A consequence of these expressions is that the regular relative entropy can be recovered by taking the limits of $\alpha$ going to $1$ and $0$ respectively,
\begin{align}
    D\reli{}=\lim_{\alpha\to 1}\DL_\alpha\reli{}=\lim_{\alpha\to 0}\frac{1-\alpha}{\alpha}\DR_\alpha\rel\sigma\rho.
\end{align}
In our results we will also need the counterpart quantity found when exchanging these limits, which we will denote $\Dstar\reli{}$, defined
\begin{align}
    \Dstar\reli{}:=\lim_{\alpha\to 1}\DR_\alpha\reli{}=\lim_{\alpha\to 0}\frac{1-\alpha}{\alpha}\DL_\alpha\rel\sigma\rho.
\end{align}
About this we note that the data-processing inequality gives $\Dstar\reli{}\leq D\reli{}$, with equality if $[\rho,\sigma]=0$.

With these definitions in hand, we can now present the large deviation bound on hypothesis testing.

\begin{lemma}[Large deviation analysis of hypothesis testing]
    \label{lem:ht_largedev}
    For any $\lambda\in \mathbb R$, define the asymptotic non-pinched/pinched log odds error per copy as
    \begin{aligns}
        \Gamma_\lambda(\rho\|\sigma)&:=\lim_{n\to\infty} \frac 1n \gamma_{\lambda n}\rel{\rho^{\otimes n}}{\sigma^{\otimes n}},\\
        \GammaL_\lambda(\rho\|\sigma)&:=\lim_{n\to\infty} \frac 1n \gammaL_{\lambda n}\rel{\rho^{\otimes n}}{\sigma^{\otimes n}},\\
        \GammaR_\lambda(\rho\|\sigma)&:=\lim_{n\to\infty} \frac 1n \gammaR_{\lambda n}\rel{\rho^{\otimes n}}{\sigma^{\otimes n}}.
    \end{aligns}
    Then, all of these limits exist, and each is given by
    \begin{aligns}
        \Gamma_\lambda(\rho\|\sigma)\!=\!&
        \begin{dcases}
                \sup_{t<0}\Dm_t\reli{}+\frac{t}{1-t}\lambda &
                \lambda<-D\rel{\sigma}{\rho},\\
                \inf_{0<t<1}\!\!\!-\Dp_t\reli{}-\frac{t}{1-t}\lambda &
                -D(\sigma\|\rho)< \lambda< 0,\\
                \sup_{t>1} -\Dm_t\reli{}+\frac{t}{1-t}\lambda &
                \lambda>0,
        \end{dcases}\\
        \GammaL_\lambda(\rho\|\sigma)\!=\!&
        \begin{dcases}
        \sup_{t<0}\DL_t\reli{}+\frac{t}{1-t}\lambda & 
                \lambda<-\Dstar\rel{\sigma}{\rho},\\
                \inf_{0<t<1}\!\!\!-\DL_t\reli{}-\frac{t}{1-t}\lambda & 
                -\Dstar(\sigma\|\rho)< \lambda< 0,\\
                \sup_{t>1} -\DL_t\reli{}+\frac{t}{1-t}\lambda & 
                \lambda> 0,
        \end{dcases}\\
        \GammaR_\lambda(\rho\|\sigma)\!=\!&
        \begin{dcases}
        \sup_{t<0}\DR_t\reli{}+\frac{t}{1-t}\lambda & 
                \lambda<-D\rel{\sigma}{\rho},\\
                \inf_{0<t<1}\!\!\!-\DR_t\reli{}-\frac{t}{1-t}\lambda & 
                -D(\sigma\|\rho)< \lambda< 0,\\
                \sup_{t>1} -\DR_t\reli{}+\frac{t}{1-t}\lambda & 
                \lambda> 0,
        \end{dcases}
    \end{aligns}
    including the edge cases
    \begin{aligns}
        \Gamma_{-D\rel\sigma\rho}\reli{}&=0, & \Gamma_0\reli{}&=-D\reli{},\\
        \GammaL_{-\Dstar\rel\sigma\rho}\reli{}&=0,~~ & \GammaL_0\reli{}&=-D\reli{},\\
        \GammaR_{-D\rel\sigma\rho}\reli{}&=0, & \GammaR_0\reli{}&=-\Dstar\reli{},
    \end{aligns}
    and the limits
    \begin{align}
        \Gamma_{\pm\infty}\reli{}=\GammaL_{\pm\infty}\reli{}=\GammaR_{\pm\infty}\reli{}=\mp\infty.
    \end{align}
\end{lemma}

\begin{proof}
    Start by dividing the range of $\lambda$ into three parameter regions, corresponding to the quadrants of \cref{fig:logodds}, 
    \begin{aligns}
        \mathcal R_L:=&\bigl(-\infty,-\kk{D(\sigma\|\rho)}\bigr),\\
        \mathcal R_M:=&\bigl(\kk{-D(\sigma\|\rho)},0\bigr),\\
        \mathcal R_R:=&(0,\infty).
    \end{aligns}
    We note that the region $\mathcal R_M$ corresponds to the so-called `error exponent' regime where both the type-I and type-II errors are exponentially decreasing, whereas $\mathcal R_L$/$\mathcal R_R$ are referred to as the `strong converse' regime where one error is exponentially decaying but the other exponentially approaching $1$. Regions $\mathcal R_M$ and $\mathcal R_L$ have been previously studied, in fact the expressions for $\Gamma_\lambda$ on $\mathcal R_M$ {\ctc can be derived from} Refs.~\cite{Hayashi07,Nagaoka06}, and the expressions for $\Gamma_\lambda$ on $\mathcal R_L$ {\ctc can be derived from } Refs.~\cite{MosonyiOgawa13,MosonyiOgawa14}\footnote{\ctc It should be noted that the results in Refs.~\cite{Hayashi07,Nagaoka06,MosonyiOgawa13,MosonyiOgawa14} are given with the roles of type-I and -II errors reversed, but by swapping the states $\rho\leftrightarrow\sigma$ and inverting these bounds they can be shown to be equivalent to the expressions for $\Gamma_\lambda$ given above.}. 
    
    By swapping the states we can extend the result in $\mathcal R_L$ to $\mathcal R_R$. Recall that $\Gamma_{\gamma}\reli{}$ quantifies the optimal type-II error possible for a given type-I error. As swapping the states corresponds to exchanging the two error types, this means that $\Gamma_\gamma\rel{\sigma}{\rho}$ must therefore quantify the optimal type-I error given for a given type-II error. {\ctc As such, we can think of the two functions}
    \begin{align}
        \lambda\mapsto \Gamma_\lambda(\rho\|\sigma)
        \qquad\text{and}\qquad
        \lambda\mapsto \Gamma_\lambda(\sigma\|\rho),
    \end{align}
    {\ctc as both describing the trade-off between two types of error, as functions of the type-I and -II errors respectively, and are therefore inverses.\footnote{\ctc Strictly speaking, these functions are not necessarily \emph{inverses}, but only \emph{quasi-inverses}, which are functions which act as inverses on each others domain/co-domain, i.e.\ functions $f$ and $g$ such that \mbox{$f\circ g\circ f=f$} and \mbox{$g\circ f\circ g = g$}.}} So, to find an expression for $\Gamma_\lambda$ on $\mathcal R_R$, we simply need to find the inverse of the flipped version on ${\ctc \lambda\leq D\reli{}}$. To this end, we define
    \begin{aligns}
        f(\lambda):=&\sup_{t>1}\frac{1-t}{t}\left[\lambda+\Dm_t(\rho\|\sigma) \right],\\
        g(\lambda):=&\sup_{s>1}\frac{s}{1-s}\lambda-\Dm_s(\rho\|\sigma).
    \end{aligns}
    By utilising the identity
    \begin{align}
        \Dm_\alpha\rel{\rho}{\sigma}&=\frac{\alpha}{1-\alpha}\Dm_{1-\alpha}\rel{\sigma}{\rho},
    \end{align}
    we can see that $f(\lambda)=\Gamma_\lambda(\sigma\|\rho)$ for $\lambda \leq {\ctc-D\reli{}}$, and we now want to argue that $g$ is its inverse. Consider composition of $f$ and $g$, which gives
    \begin{aligns}
        (f\circ g)(\lambda)&=\sup_{t>1}\inf_{s>1}\frac{1-t}{t}\left[ \Dm_t\reli{}-\Dm_s\reli{}+\frac{s}{1-s}\lambda \right],\\
        (g\circ f)(\lambda)&=\sup_{s>1}\inf_{t>1}\frac{s}{1-s}\frac{1-t}{t}\left[\lambda+\Dm_t\reli{}\right]-\Dm_s\reli{}.
    \end{aligns}
    If we let $t=s$ and $s=t$ in the inner optimisations in each of these,  we get $(f\circ g)(\lambda)\geq \lambda$ and $(g\circ f)(\lambda)\geq \lambda$, respectively. Using these, together with the fact that $f$ is monotonically {\ctc non-increasing}, we have
    \begin{aligns}
        f \circ g \circ f=(f \circ g) \circ f\geq f,\\
        f \circ g \circ f=f \circ (g \circ f)\leq f.
    \end{aligns}
    Thus, we have that $f\circ g\circ f=f$, and so $f$ and $g$ are {\ctc quasi-}inverses. As $f$ is $\Gamma_\lambda(\sigma\|\rho)$ on $\lambda{\ctc \leq -D\reli{}}$, then $g$ must correspond to $\Gamma_\lambda\reli{}$ for $\lambda\in \mathcal R_R$, as required.

    Next we turn to the pinched variants $\GammaL_\lambda$/$\GammaR_\lambda$---we will start with $\GammaL_\lambda$. Here, we want to consider the hypothesis testing between the pinched state $\pinch{\rho^{\otimes n}}{\sigma^{\otimes n}}$ and~$\sigma^{\otimes n}$. As these states are no longer iid, this can be considered as a hypothesis testing problem between two correlated states. Thankfully, the problem of extending the previously mentioned hypothesis testing analyses to the case of correlated states has been considered. Specifically, Ref.~\cite[Thm.~4.8]{HiaiMosonyiOgawa2007} considers the error exponent regime ($\mathcal R_M$) and Ref.~\cite[Cor~IV.6]{MosonyiOgawa14} the strong converse regime ($\mathcal R_L/\mathcal R_R$). In both cases, it has been shown that, if the regularised relative entropy exists and is differentiable (see \cref{app:pinch}), then the standard iid results can be extended where the single-copy relative entropy is replaced with this regularised quantity. For our case, looking at pinched states, this means that the change when going from non-pinched to pinched hypothesis testing is
    \begin{align}
        \Dp\reli{},\Dm\reli{} \to \DL\reli{}.
    \end{align}
    Making this substitution, we get the required expression for $\GammaL_\lambda$ on $\mathcal R_L$, $\mathcal R_M$, and $\mathcal R_R$. An analogous argument can also be made for $\GammaR_\lambda$.
    
    The only values of $\lambda$ left to consider are the edge cases and limits. In each case these follow from the monotonicity of $\Gamma_\lambda$/$\GammaL_\lambda$/$\GammaR_\lambda$ in $\lambda$. The edge cases
    \begin{align}
        \Gamma_{-D\rel\sigma\rho}\reli{}&=0, & \Gamma_0\reli{}&=-D\reli{}
    \end{align}
    correspond to \cref{lem:ht_stein} (and its state-reversed analogue), and similarly
    \begin{aligns}
        \GammaL_{-\Dstar\rel\sigma\rho}\reli{}&=0, & \GammaL_0\reli{}&=-D\reli{},\\
        \GammaR_{-D\rel\sigma\rho}\reli{}&=0, & \GammaR_0\reli{}&=-\Dstar\reli{},
    \end{aligns}
    to the pinched variants of Stein's Lemma. Lastly the limits
    \begin{align}
        \Gamma_{\pm\infty}\reli{}=\GammaL_{\pm\infty}\reli{}=\GammaR_{\pm\infty}\reli{}=\mp\infty
    \end{align}
    follow from the fact that \mbox{$\beta_0\rel{\cdot}{\cdot}=\betaL_0\rel{\cdot}{\cdot}=\betaR_0\rel{\cdot}{\cdot}=1$} and \mbox{$\beta_1\rel{\cdot}{\cdot}=\betaL_1\rel{\cdot}{\cdot}=\betaR_1\rel{\cdot}{\cdot}=0$} generically for any states of full support.
\end{proof}


\subsubsection{Moderate deviation}
\label{subsubsec:ht_mod}

In the small deviation regime we considered type-I errors which were constant, and in the large we considered type-I errors which were exponentially approaching 0/1. This leaves a gap for errors which are approaching 0 or 1, but doing so sub-exponentially. This regime is referred to as \emph{moderate deviations}~\cite{dembo98,chubb2017moderate,cheng17}. As with large deviations it will be advantageous to express these results in term of type-I/-II log odds. 

Where small deviations correspond to type-I log odds which are constant in $n$, and large deviation to any log odds which scales as $\pm \lambda n$, moderate will refer to any log odds which scales as $\pm \lambda n^a$ where $\lambda>0$ and $a\in(0,1)$. We note that in some other papers considering moderate deviations---such as Refs.~\cite{chubb2017moderate,cheng17}---these results are considered more generally for any sequence $x_n$ such that $\lim_{n\to\infty} x_n=0$ and $\lim_{n\to\infty} nx_n=\pm \infty$. We will restrict to this polynomial subset of such sequences primarily for notational convenience, but note that all of the below results can be extended to these more general moderate sequences.

\begin{lemma}[Moderate deviation analysis of hypothesis testing]
    \label{lem:ht_moddev}
    For any $\lambda>0$ and $a\in(0,1)$ the type-II log odds of error scales as
    \begin{aligns}
        \frac 1n \gamma_{\pm \lambda n^a}\rel{\rho^{\otimes n}}{\sigma^{\otimes n}}&\simeq-D(\rho\|\sigma)\mp
        \sqrt{2 V\reli{}\lambda n^{a-1}},\\
        \frac 1n \gammaL_{\pm \lambda n^a}\rel{\rho^{\otimes n}}{\sigma^{\otimes n}}&\simeq-D(\rho\|\sigma)\mp
        \sqrt{2 V\reli{}\lambda n^{a-1}},
    \end{aligns}
    where $\simeq$ denotes equality up to terms scaling as $o\left(\sqrt{n^{a-1}}\right)$.
\end{lemma}
\begin{proof}
    The non-pinched result is just a restatement of the hypothesis testing result from Ref.~\cite[Thm.~1]{chubb2017moderate}. For the pinched quantity we will use a similar argument to that used in the small deviation case of \cref{lem:ht_smalldev}. 
    
    The data-processing inequality $\gammaL_x\reli{}\geq \gamma_x\reli{}$ trivially gives us a lower bound on the scaling of $\gammaL$. For the upper bound we return to the inequality
    \begin{align}
        \betaL_\epsilon\rel{\rho}{\sigma}\leq \beta_{\epsilon-2\delta}\rel{\rho}{\sigma}\cdot \frac{2^8(\epsilon-\delta)\nu(\sigma)^2}{\delta^5(1-\epsilon+\delta)}
    \end{align}
    considered previously in the proof of \cref{lem:ht_smalldev}. We want to use this bound in the two moderate regimes, in which the type-I error is approaching 0/1. To do this, consider the two equalities when we substitute \mbox{$\delta=\epsilon/4$} for the case of $\epsilon$ approaching $0$, and \mbox{$\delta=1-\epsilon$} for the case of $\epsilon$ approaching $1$ (with the added requirement that $\epsilon>3/4$), giving 
    \begin{aligns}
        \betaL_\epsilon\rel{\rho}{\sigma} &\leq \beta_{\epsilon/2}\rel{\rho}{\sigma}\cdot \frac{3\cdot2^{16}\nu(\sigma)^2}{\epsilon^4(1-3\epsilon/4)},\\
        \betaL_\epsilon\rel{\rho}{\sigma} &\leq \beta_{3\epsilon-2}\rel{\rho}{\sigma}\cdot \frac{2^{7}(2\epsilon-1)\nu(\sigma)^2}{(1-\epsilon)^6}.
    \end{aligns}

    Now we want to make the substitutions $\rho\to\rho^{\otimes n}$ and $\sigma\to\sigma^{\otimes n}$. In the former case we will substitute \mbox{$\epsilon\to L^{-1}[-\lambda n^a]$}, and in the latter \mbox{$\epsilon\to L^{-1}[+\lambda n^a]$}. Taking logarithms, and recalling that the number of unique eigenvalues $\nu\left(\sigma^{\otimes n}\right)$ scales only polynomially with $n$, this gives us the upper bounds expressed in terms of log odds per copy of
    \begin{align}
        \frac 1n \gammaL_{\pm \lambda n^a}\rel{\rho^{\otimes n}}{\sigma^{\otimes n}}&\leq \frac 1n \gamma_{\pm\lambda n^a}\rel{\rho^{\otimes n}}{\sigma^{\otimes n}}+O(n^{a-1}).
    \end{align}
    Lastly, as $a<1$ we have that $n^{a-1}=o\left(\sqrt{n^{a-1}}\right)$, allowing us to neglect this error term, and thus conclude that the pinched and non-pinched log odds per copy must scale identically up to $\simeq$, as required.
\end{proof}


\subsubsection{Extreme deviation}
\label{subsubsec:ht_extreme}

Now that we have dealt with type-I errors that do not approach 0 or 1, approach them sub-exponentially, and approach them exponentially, we are left with one final case: when the error approaches 0 or 1 \emph{super}-exponentially. We will see that if we consider super-exponentially scaling errors, the problem of hypothesis testing becomes `boring', in the sense that we get a simple linear trade-off between the two types of error, as the error constraints are too strict for any meaningful trade-off. While boring in and of itself, the analysis of hypothesis testing in this regime is needed for technical reasons in the proof of the zero-error transformation rates between quantum dichotomies to come in \cref{subsubsec:rate_extreme}.

One silver lining of the boringness of this regime is that we do not need to consider an asymptotic number of states and can start with a single-shot statement. The quantum Neymann-Pearson lemma~\cite{helstrom1969quantum,holevo1972analog} states that the trade-off between type-I and type-II hypothesis testing error can be characterised by Neyman-Pearson tests of the form
\begin{align}
    Q_t:=\left\lbrace \rho-t\sigma>0 \right\rbrace,
\end{align}
for $t\geq 0$, where $\lbrace M>0\rbrace$ denotes the projector onto the eigenspaces of $M$ corresponding to positive eigenvalues. Specifically, the claim is that the optimal trade-off between the errors is either given by a test of the form $Q_t$, or, when $t$ corresponds to a value at which $Q_t$ changes rank, a convex combination of $\lim_{s\to t^-}Q_s$ and $Q_t$. 

{\ctc
Recall that the $\alpha\to\pm\infty$ limiting cases of the minimal divergence are given by the max-divergence~\cite[\S 4.2]{Tomamichel2016},
\begin{aligns}
    \Dm_{+\infty}\reli{}&
    = \log \lambda_{\max}\left(\sigma^{-1/2}\rho\sigma^{-1/2}\right),\\
    \Dm_{-\infty}\reli{}&
    = -\log\lambda_{\max}\left(\rho^{-1/2}\sigma\rho^{-1/2}\right).
\end{aligns}
}

\begin{lemma}[Single-shot extreme deviation analysis of hypothesis testing]
    \label{lem:ht_extremedev_singleshot}
    For any $x\leq \lambda_{\min}(\rho)$
    \begin{aligns}
        \beta_{1-x}\reli{} &= x\cdot \exp\left(-\Dm_{+\infty}\rel{\rho}{\sigma}\right),\\
        1-\beta_x\reli{} &= x\cdot \exp\left(-\Dm_{-\infty}\reli{}\right),
    \end{aligns}
    and
    \begin{aligns}
        \betaR_{1-x}\reli{} &= x\cdot \exp\Bigl(-D_{+\infty}\rel\rho{\pinch\sigma\rho}\Bigr),\\
        1-\betaR_x\reli{} &= x\cdot \exp\Bigl(-D_{-\infty}\rel\rho{\pinch\sigma\rho}\Bigr).
    \end{aligns}
    Also, for any $x\leq \lambda_{\min}\left(\pinch\rho\sigma\right)$
    \begin{aligns}
        \betaL_{1-x}\reli{} &= x\cdot \exp\Bigl(-D_{+\infty}\rel{\pinch\rho\sigma}{\sigma}\Bigr),\\
        1-\betaL_x\reli{} &= x\cdot \exp\Bigl(-D_{-\infty}\rel{\pinch\rho\sigma}{\sigma}\Bigr).
    \end{aligns}  
\end{lemma}
\begin{proof}
    Firstly, we note that we can rewrite the Neymann-Pearson test $Q_t$ as
    \begin{align}
        Q_t = \left\lbrace \rho^{1/2}\left(I-t\rho^{-1/2}\sigma\rho^{-1/2}\right)\rho^{1/2}> 0\right\rbrace.
    \end{align}
    Clearly $Q_0=I$, but from this we can also see that $Q_t=I$ for any $t<t^*$, where
    \begin{align}
        t^*:=&1/\lambda_{\max}(\rho^{-1/2}\sigma\rho^{-1/2})
        =\exp\left(\Dm_{-\infty}\reli{}\right).    
    \end{align}
    As such, we can see that the first non-trivial projective Neyman-Pearson test is given by $Q_{t^*}$. Since $Q_{t}$ are necessarily not full rank for any $t>t^*$, they must have a type-I error of at least $\lambda_{\min}(\rho)$. Thus, to get a type-I error of $0<x<\lambda_{\min}(\rho)$ we must consider a test that is a convex combination of $Q_0=I$ and $Q_{t^*}$.

    {\ctc Asssume for the moment that $\rho^{-1/2}\sigma\rho^{-1/2}$ has a non-degenerate maximal eigenvector (we will return to this below), with eigenvalue $1/t^*$ and eigenvector $\ket\psi$.} Then, this first non-trivial projective test is 
    \begin{align}
        \ctc
        Q_{t^*}=I-\frac{\rho^{-1/2}\ketbra{\psi}{\psi}\rho^{-1/2}}{\braopket{\psi}{\rho^{-1}}{\psi}}.
    \end{align}
    So a test $Q$ which is a convex combination of $Q_0$ and $Q_{t^*}$, and has type-I error of $x$, takes the form
    \begin{align}
        Q:=I-x\cdot{\ctc \rho^{-1/2}\ketbra\psi\psi\rho^{-1/2}}
        .
    \end{align}
    As $0\leq x< \lambda_{\min}(\rho)$, we have that this is a valid test, and the type-I error is simply given by $x$ as required, 
    \begin{align}
        1-\Tr (\rho Q)=x.
    \end{align}
    For the type-II we get the desired expression,
    \begin{aligns}
        \beta_x(\rho\|\sigma)
        &=\Tr (\sigma Q)\\
        &=1-x\cdot{\ctc \Tr(\sigma \rho^{-1/2}\ketbra\psi\psi\rho^{-1/2})}\\
        &=1-x\cdot{\ctc \bra{\psi}{\rho^{-1/2}\sigma\rho^{-1/2}}\ket{\psi}}\\
        &=1-x/t^*\\
        &=1-x\cdot \exp\left(-\Dm_{-\infty}\reli{}\right).
    \end{aligns}
    For $x>1-\lambda_{\min}(\rho)$ a similar argument can be given \emph{mutatis mutandis} by considering the last non-trivial test, which gives
    \begin{align}
        \beta_x\reli{}=(1-x)\exp\left(-\Dm_{+\infty}\reli{}\right).
    \end{align}
    Finally, the pinched results trivially follow from the non-pinched variants by making the substitution $\sigma\to\pinch\sigma\rho$ and $\rho\to\pinch\rho\sigma$, respectively.

    {\ctc In the above we assumed that $\rho^{-1/2}\sigma\rho^{-1/2}$ has a non-degenerate maximal eigenvalue---an assumption to which we now return. The idea now is to show that we can perturb the state $\sigma$ by an arbitrarily small amount to break any such degeneracy. Specifically, consider letting $\ket\psi$ be an arbitrary maximal eigenvector of $\rho^{-1/2}\sigma\rho^{-1/2}$, and define the (unnormalised) state $\sigma_\epsilon$ as
    \begin{align}
        \sigma_\epsilon := \sigma + \epsilon \rho^{1/2}\proj\psi\rho^{1/2}. 
    \end{align}
    We can see that this breaks the degeneracy as
    \begin{align}
        \rho^{-1/2}\sigma_\epsilon\rho^{-1/2} = \rho^{-1/2}\sigma\rho^{-1/2} + \epsilon\proj\psi,
    \end{align}
    allowing us to apply the above proof to give expressions for $\beta_x\rel{\rho}{\sigma_\epsilon}$. Next, we can see that $\beta_x\rel{\rho}{\sigma_\epsilon}\to\beta_x\rel{\rho}{\sigma}$ as $\epsilon\to 0^+$ as the difference is bounded by the trace norm,
    \begin{aligns}
        \abs{\beta_x\rel{\rho}{\sigma_\epsilon}-\beta_x\reli{}}
        &\leq \norm{\sigma_\epsilon-\sigma}_{\Tr}\\
        &=\epsilon\braopket{\psi}{\rho}{\psi}\\
        &\leq \epsilon.
    \end{aligns}
    }
\end{proof}

Applying this single-shot analysis to the case of an asymptotically large number of copies of each state, we can get asymptotic expressions for the log odds per copy that are comparable with the large deviation expressions of \cref{lem:ht_largedev}, for both non-pinched and pinched hypothesis testing problems.

\begin{lemma}[Extreme deviation analysis of hypothesis testing]
    \label{lem:ht_extremedev}
    For any $\lambda > -\log \lambda_{\min}(\rho)$:
    \begin{aligns}
        \Gamma_{\pm\lambda}\reli{}&=\mp\lambda-\Dm_{\pm\infty}\reli{},\\
        \GammaL_{\pm\lambda}\reli{}&=\mp\lambda-\DL_{\pm\infty}\reli{},\\
        \GammaR_{\pm\lambda}\reli{}&=\mp\lambda-\DR_{\pm\infty}\reli{},
    \end{aligns}
    where we recall that
    \begin{aligns}
        \Gamma_{\lambda}\reli{}:=&\lim_{n\to \infty} \frac 1n \gamma_{+\lambda n}\reli{},\\
        \GammaL_{\lambda}\reli{}:=&\lim_{n\to \infty} \frac 1n \gammaL_{+\lambda n}\reli{},\\
        \GammaR_{\lambda}\reli{}:=&\lim_{n\to \infty} \frac 1n \gammaR_{+\lambda n}\reli{}.
    \end{aligns}
\end{lemma}
\begin{proof}
    For the non-pinched case, we can simply apply the single-shot result, \cref{lem:ht_extremedev_singleshot}, to the states $\rho^{\otimes n}$ and~$\sigma^{\otimes n}$, using the additivity of the $\Dm_{\pm\infty}$, and expressing the type-I and -II errors in terms of log-odds.

    We will do similarly for the pinched case. Firstly, we use the pinching inequality
    \begin{align}
        \pinch{\rho^{\otimes n}}{\sigma^{\otimes n}}\geq \frac{\rho^{\otimes n}}{\abs{\mathrm{spec}(\sigma^{\otimes n})}}\geq \frac{\rho^{\otimes n}}{n^d},
    \end{align}
    and so
    \begin{align}
        \log \lambda_{\min}\left(\pinch{\rho^{\otimes n}}{\sigma^{\otimes n}}\right) \geq n\lambda_{\min}(\rho)-O(\log n).
    \end{align}
    Thus, if we have a \emph{strict} inequality $\lambda >-\log\lambda_{\min}(\rho)$, then
    \begin{align}
        n\lambda \gtev -\log \lambda_{\min}\left(\pinch{\rho^{\otimes n}}{\sigma^{\otimes n}}\right).
    \end{align}
    Given this, we can now substitute the pinched states into \cref{lem:ht_extremedev_singleshot}, which give the desired expressions for $\GammaL_{\pm \lambda}$ and $\GammaR_{\pm \lambda}$.
\end{proof}


\subsection{Transformation rates}
\label{subsec:rates}

In this section, we will take the asymptotic analysis of hypothesis testing from the previous section and extend it to transformation rates between quantum dichotomies. To be more concrete, for some sequence of errors $\epsilon_n$ and fixed states $\rho_1,\rho_2,\sigma_1,\sigma_2$, recall the definition of $R_n^*(\epsilon_n)$ as the maximum $R_n$ such that
\begin{align}
    (\rho^{\otimes n}_1, \sigma^{\otimes n}_1) \succeq_{(\epsilon_n,0)} 
    (\rho^{\otimes R_nn}_2, \sigma^{\otimes R_nn}_2).
\end{align}
We will be studying the scaling of $R_n^*(\epsilon_n)$ for various scaling regimes of $\epsilon_n$. While we will restrict below to just the case of such one-sided errors, we cover how these techniques can be extended, and what the resulting rates are, in the more general regime of two-sided errors in \cref{app:two-side}.

To spare the reader from being subjected to the phrase `for sufficiently large $n$' \textit{ad nauseam}, we use notation $\ltev$ and $\gtev$ to denote eventual inequalities for the following proofs. In other words, we will use $a_n\ltev b_n$ and $b_n\gtev a_n$ as shorthand for
\begin{align}
    \exists N:~a_n<b_n~\forall n\geq N.
\end{align}
We note that if the quantities are functions, then this will just denote pointwise eventual inequality, i.e., $f_n(x)\ltev g_n(x)$ is shorthand for 
\begin{align}
    \forall x~\exists N(x):~f_n(x)<g_n(x)~\forall n\geq N(x),
\end{align}
and not
\begin{align}
    \exists N:~f_n(x)<g_n(x)~\forall x,\forall n\geq N.
\end{align}
Upgrading from such pointwise inequalities to uniform inequalities will be important in the achievability proofs to come, requiring uniform versions of the lemmas in \cref{subsec:htass}, which are presented in \cref{app:uni}.

We start with the first-order rate, since all of the remaining results in this section are refinements of this first-order rate. Moreover, all of the remaining proofs will follow a general approach that extends the proof below. We will quantify the optimal transformation rates regime by providing both upper and lower bounds, referred to as the \emph{optimality} and \emph{achievability} bounds, respectively. In all cases these bounds will follow from \cref{lem:ht}, which provides both necessary and sufficient conditions for the existence of a transformation in terms of hypothesis testing quantities. 

The first-order transformation rate is captured by the following theorem.

\ratefirst*
\begin{proof}
    We start with optimality. Consider a rate \mbox{$R>D\reli{1}/D\reli{2}$}. As $\epsilon\in(0,1)$, we also have that \mbox{$(1\pm \epsilon)/2\in(0,1)$}, allowing us to apply \cref{lem:ht_stein}. On the input side this gives
    \begin{align}
        \lim_{n\to\infty }-\frac 1n \log \beta_{\frac{1+\epsilon}2}\rel{\rho_1^{\otimes n}}{\sigma_1^{\otimes n}}=D\reli{1},
    \end{align}
    and on the target side
    \begin{aligns}
        &\lim_{n\to\infty }-\frac 1n \log \beta_{\frac{1-\epsilon}2}\rel{\rho_2^{\otimes Rn}}{\sigma_2^{\otimes Rn}}\notag\\
        &\qquad\qquad=\lim_{m\to\infty }-\frac Rm \log \beta_{\frac{1-\epsilon}2}\rel{\rho_2^{\otimes m}}{\sigma_2^{\otimes m}}\\
        &\qquad\qquad=RD\reli{2}\\
        &\qquad\qquad>D\reli{1}.
    \end{aligns}
    If we let $x=(1+\epsilon)/2$, then
    \begin{align}
        \beta_{x}\rel{\rho_1^{\otimes n}}{\sigma_1^{\otimes n}}
        \gtev
        \beta_{x-\epsilon}\rel{\rho_2^{\otimes Rn}}{\sigma_2^{\otimes Rn}},
    \end{align}
    and so by \cref{lem:ht} this means that transformation at a rate of $R$ is eventually \emph{not} possible. Thus, $R$ provides an upper bound for asymptotic optimal transformation rate. As this argument holds for \emph{any} \mbox{$R>D\reli{1}/D\reli{2}$} this means
    \begin{align}
        \limsup_{n\to\infty} R_n^*(\epsilon)\leq \frac{D\reli{1}}{D\reli{2}}.
    \end{align}
    
    Next, we proceed to proving the achievability for commuting target dichotomies, $[\rho_2,\sigma_2]=0$. Consider a rate \mbox{$r<D\reli{1}/D\reli{2}$}. To show a rate is achievable, we need to consider the pinched hypothesis testing of the input dichotomy. Specifically, we will use the limits from \cref{lem:ht_smalldev}
    \begin{aligns}
        \lim_{n\to\infty }-\frac 1n \log \betaL_{\epsilon}\rel{\rho_1^{\otimes n}}{\sigma_1^{\otimes n}}&=D\reli{1},\\
        \lim_{n\to\infty }-\frac 1n \log \beta_{1-\epsilon}\rel{\rho_2^{\otimes rn}}{\sigma_2^{\otimes rn}}&=rD\reli{2}.
    \end{aligns}
    As \mbox{$r<D\reli{1}/D\reli{2}$}, combining these gives
    \begin{align}
        \betaL_{\epsilon}\rel{\rho_1^{\otimes n}}{\sigma_1^{\otimes n}}
        \ltev
        \beta_{1-\epsilon}\rel{\rho_2^{\otimes rn}}{\sigma_2^{\otimes rn}}.
    \end{align}
    But this is only true for specific errors, while \cref{lem:ht} requires such an inequality for \emph{all} $x$ in a range. How do we span this gap? For the first-order problem (and high-error cases of moderate and large deviations) this is easily solved by simply using the monotonicity of $\beta_x\rel{\cdot}{\cdot}$, $\betaL_x\rel{\cdot}{\cdot}$ and $\betaR_x\rel{\cdot}{\cdot}$ as functions of $x$ for fixed states. We can use this to relax the preceding inequality to
    \begin{align}
        \betaL_{x}\rel{\rho_1^{\otimes n}}{\sigma_1^{\otimes n}}
        \ltev
        \beta_{x-\epsilon}\rel{\rho_2^{\otimes rn}}{\sigma_2^{\otimes rn}}~~\forall x\in(\epsilon,1).
    \end{align}
    We note that because this set of inequalities (parameterised by $x$) follows from the previous $x$-independent inequality, there is no issue of uniformity, i.e.,~there exists an $x$-independent $N$ such that this holds for $n\geq N$. As such, we can apply \cref{lem:ht}, which gives that transformation at a rate of $r$ \emph{is} eventually achievable, and thus
    \begin{align}
        \liminf_{n\to\infty} R_n^*(\epsilon)\geq \frac{D\reli{1}}{D\reli{2}}.
    \end{align}
\end{proof}

As with hypothesis testing in the previous subsection, we will now spend the rest of this subsection giving our refinements on this first-order result for different regimes of the scaling of the error $\epsilon_n$. A summary of these different regimes is given in \cref{fig:sigmoid_rate}.


\subsubsection{Small deviation}
\label{subsubsec:rate_small}

We start with the small deviation regime, in which the error is a constant $\epsilon\in(0,1)$. The proof is broadly similar to that of the first-order rate in \cref{thm:rate_first}, but more care has to be taken to capture the second-order contribution, especially on the achievability side. Also, recall that \cref{eq:reversibility} defined the reversibility parameter as
\begin{align}
    \xi:=\frac{V(\rho_1\|\sigma_1)}{D(\rho_1\|\sigma_1)} \bigg/ \frac{V(\rho_2\|\sigma_2)}{D(\rho_2\|\sigma_2)}.
\end{align}
Then we have the following.

\ratesmalldev*
\begin{proof}
    Consider a small slack parameter $\delta>0$ and define the rate $R_n$ as
    \begin{align}
        R_n:=\frac{D(\rho_1\|\sigma_1)+\sqrt{V\reli{1}/n}\cdot S_{1/\xi}^{-1}(\epsilon)+\delta/\sqrt n}{D(\rho_2\|\sigma_2)}.
    \end{align}
    Recall that by \cref{lem:sesqui} the term $S_{1/\xi}^{-1}(\epsilon)$ can be expressed as a minimum, 
    \begin{align}
        S_{1/\xi}^{-1}(\epsilon)=\min_{x\in(\epsilon,1)} \left[\Phi^{-1}(x)-\sqrt{1/\xi}\cdot \Phi^{-1}(x-\epsilon)\right].
    \end{align}
    Let $x^*>\epsilon$ denote the value of $x$ at which this minimum is attained, such that
    \begin{align}
        R_n=&\frac{D(\rho_1\|\sigma_1)}{D(\rho_2\|\sigma_2)}
        +\frac{\delta}{\sqrt n D\reli{2}}
        +\sqrt{\frac{V\reli{1}}{n D^2\reli{2}}}\Phi^{-1}(x^*)
        \notag\\
        &\quad-\sqrt{\frac{V\reli{2}D\reli{1}}{nD^3\reli{2}}} \Phi^{-1}(x^*-\epsilon).
    \end{align}
    Next we turn to \cref{lem:ht_smalldev}. For the input dichotomy this gives
    \begin{align}
        &-\frac 1n\log\beta_{x^*}\rel{\rho_1^{\otimes n}}{\sigma_1^{\otimes n}}
        \notag\\
        &\qquad\qquad\simeq D\reli{1}+\sqrt{\frac{V\reli{1}}n}\Phi^{-1}(x^*),
    \end{align}
    and for the target we can substitute in $R_n$ to get
    \begin{aligns}
        \!\!\!&-\frac 1n\log\beta_{x^*-\epsilon}\rel{\rho_2^{\otimes R_nn}}{\sigma_2^{\otimes R_nn}}\notag\\
         \!\!\!&\qquad\simeq R_nD\reli{2}+\sqrt{\frac{R_nV\reli{2}}n}\Phi^{-1}(x^*-\epsilon)\\
         \!\!\!&\qquad= D\reli{1}+\sqrt{\frac{V\reli{1}}n}\Phi^{-1}(x^*)+\frac{\delta}{\sqrt n}\\
         \!\!\!&\qquad\simeq -\frac 1n\log\beta_{x^*}\rel{\rho_1^{\otimes n}}{\sigma_1^{\otimes n}}+\frac{\delta}{\sqrt n}.
    \end{aligns}
    Thanks to this $\delta>0$ term, we can therefore conclude that
    \begin{align}
        \beta_{x^*}\rel{\rho_1^{\otimes n}}{\sigma_1^{\otimes n}}
        \gtev
        \beta_{x^*-\epsilon}\rel{\rho_2^{\otimes R_nn}}{\sigma_2^{\otimes R_nn}},
    \end{align}
    and so by \cref{lem:ht} the transformations at the rate $R_n$ are eventually \emph{not} possible. As this is true for all $\delta>0$, this then upper bounds the optimal rate
    \begin{align}
        R_n^*(\epsilon)\lesssim \frac{D(\rho_1\|\sigma_1)+\sqrt{V\reli{1}/n}\cdot S_{1/\xi}^{-1}(\epsilon)}{D(\rho_2\|\sigma_2)}.
    \end{align}

    Now we turn to achievability. Once again consider a small slack parameter $0<\delta<\epsilon/2$, and define the rate $r_n$ as
    \begin{align}
        r_n:=&\frac{D\reli{1}}{D\reli{2}}-\frac{\delta}{\sqrt n D\reli{2}}\\\notag
        &~+\frac{1}{\sqrt n}\min_{y\in[\epsilon,1]} \Biggl[
        \frac{\sqrt{V\reli{1}}}{D\reli{2}}
        \Phi^{-1}(y-\delta)\\\notag
        &~\qquad\qquad-\sqrt{\frac{V\reli{2}D\reli{1}}{D\reli{2}^3}}\Phi^{-1}(y-\epsilon+\delta)\Biggr].
    \end{align}
     Again, for \kk{convenience}, let $y^*$ denote a the minimiser in the above optimisation over $y$. If we let $x\in(\epsilon,1)$ then applying \cref{lem:ht_smalldev} to the input gives
    \begin{align}
        &-\frac 1n\log \betaL_{x-\delta}\rel{\rho_1^{\otimes n}}{\sigma_1^{\otimes n}}\notag\\
        &\qquad\simeq D\reli{1}
        +\sqrt{\frac{V\reli{1}}{n}}\cdot\Phi^{-1}(x-\delta),
    \end{align}
    and to the target gives
    \begin{aligns}
        &-\frac 1n\log \beta_{x-\epsilon+\delta}\rel{\rho_2^{\otimes r_nn}}{\sigma_2^{\otimes r_nn}}\notag\\
        &\quad\qquad\simeq r_nD\reli{2}
        \\\notag&\qquad\qquad~~
        +\sqrt{\frac{r_nV\reli{2}}{n}}\cdot\Phi^{-1}(x-\epsilon+\delta),\\
        &\quad\qquad\simeq D\reli{1}-\delta/\sqrt n
        \\\notag&\qquad\qquad
        +\sqrt{\frac{V\reli{2}D\reli{1}}{nD\reli{2}}}\Phi^{-1}(x-\epsilon+\delta)
        \\\notag&\qquad\qquad
        +\sqrt{\frac{V\reli{1}}{n}}\Phi^{-1}(y^*-\delta)
        \\\notag&\qquad\qquad
        -\sqrt{\frac{V\reli{2}D\reli{1}}{nD\reli{2}}}\Phi^{-1}(y^*-\epsilon+\delta)
        .
    \end{aligns}
    Combining these we have
    \begin{align}
        &-\frac 1n\log \beta_{x-\epsilon+\delta}\rel{\rho_2^{\otimes r_nn}}{\sigma_2^{\otimes r_nn}}\notag\\
        &\quad\qquad\simeq -\frac 1n\log \betaL_{x-\delta}\rel{\rho_1^{\otimes n}}{\sigma_1^{\otimes n}}-\delta/\sqrt n
        \\\notag&\qquad\qquad
        ~~-\sqrt{\frac{V\reli{1}}{n}}\Phi^{-1}(x-\delta)
        \\\notag&\qquad\qquad
        ~~+\sqrt{\frac{V\reli{2}D\reli{1}}{nD\reli{2}}}\Phi^{-1}(x-\epsilon+\delta)
        \\\notag&\qquad\qquad
        ~~-\sqrt{\frac{V\reli{2}D\reli{1}}{nD\reli{2}}}\Phi^{-1}(y^*-\epsilon+\delta)
        \\\notag&\qquad\qquad
        ~~+\sqrt{\frac{V\reli{1}}{n}}\Phi^{-1}(y^*-\delta)
        .
    \end{align}
    Recalling that $y^*$ was defined as the minimiser over just such an expression, we get
    \begin{align}
        &-\frac 1n\log \beta_{x-\epsilon+\delta}\rel{\rho_2^{\otimes r_nn}}{\sigma_2^{\otimes r_nn}}\notag\\
        &\qquad\qquad\lesssim -\frac 1n\log \betaL_{x-\delta}\rel{\rho_1^{\otimes n}}{\sigma_1^{\otimes n}}-\delta/\sqrt n.
    \end{align}
    Because of $\delta>0$, this in turn implies the eventual inequality
    \begin{align}
        \betaL_{x-\delta}\rel{\rho_1^{\otimes n}}{\sigma_1^{\otimes n}}
        \ltev \beta_{x-\epsilon+\delta}\rel{\rho_2^{\otimes r_nn}}{\sigma_2^{\otimes r_nn}}.
    \end{align}
    Finally we can relax out the slack parameter, once again using the fact that $\beta_x$ and $\betaL_x$ are monotone decreasing as functions of $x$, giving 
    \begin{align}
        \label{eq:type2ineq}
        \betaL_{x}\rel{\rho_1^{\otimes n}}{\sigma_1^{\otimes n}}
        \leq \beta_{x-\epsilon}\rel{\rho_2^{\otimes r_nn}}{\sigma_2^{\otimes r_nn}}.
    \end{align}
    
    We note that in the above proof we have skipped over the issue of uniformity. While we did show that \cref{eq:type2ineq} holds for all $x$ eventually, it still remains to be seen that it eventually holds for all $x$, the latter of which would be required to apply \cref{lem:ht}. To put it less confusingly, we have shown that there exists an $N$ such that \cref{eq:type2ineq} holds for all $n>N$, but have not ruled out the possibility that $N$ depends on~$x$. Such a dependence might mean that there is no $x$-independent $N$ beyond which this expression holds for all $x$, which is what would be needed by \cref{lem:ht}. However, if we swap out \cref{lem:ht_smalldev} with its uniform version (\cref{lem:ht_smalldev_uni}, presented and proven in \cref{app:uni}), then we do indeed get an $N$ that is independent of $x$ (though still dependent on $\rho_1,\rho_2,\sigma_1,\sigma_2,\epsilon,\delta$ of course), removing this issue. Given this, we can now conclude that \cref{eq:type2ineq} eventually holds for all $x\in(\epsilon,1)$. Having dealt with this uniformity issue, we can we can return to \cref{lem:ht}, which allows us to conclude that transformation at the rate $r_n$ \emph{is} eventually possible. This held true for all small $\delta>0$. Taking the limit $\delta\to 0^+$, and using the continuity of $\Phi^{-1}$ on $(0,1)$, we can see that this rate does indeed limit to the desired expression, lower bounding the optimal rate as 
    \begin{align}
        R_n^*(\epsilon_n)\gtrsim \frac{D(\rho_1\|\sigma_1)+\sqrt{V\reli{1}/n}\cdot S_{1/\xi}^{-1}(\epsilon)}{D(\rho_2\|\sigma_2)}.
    \end{align}
\end{proof}


\subsubsection{Large deviation}
\label{subsubsec:rate_large}

Now we turn to the large deviation regime, in which we consider errors which are exponentially approaching either $0$ or $1$, which we refer to a low and high-error.
The general structure of the proof follows that of the small deviation case, but will be split into two sub-regimes: high and low-error. In high-error we have that $\epsilon_n$ is exponentially close to $1$, and so the region $x\in(\epsilon_n,1)$ is quite small. This allows us to get an optimal expression for the transformation rate with a single application of the large deviation analysis of hypothesis testing \cref{lem:ht_largedev}, similar to the proof of the first-order rate \cref{thm:rate_first}. For the low-error case, however, the region $x\in(\epsilon_n,1)$ is quite large, requiring us to consider hypothesis testing for a whole interval of possible error exponents. Here the proof will more closely follow that of the small deviation case, \cref{thm:rate_smalldev}, running into the same uniformity issues. As the high-error proof is simpler, we shall start there.

\ratelargedevhi*
\begin{proof}
    Consider a rate $R$ such that
    \begin{align}
        R> \inf_{\substack{t_1>1\\0<t_2<1}}\frac{\Dm_{t_1}\reli{1}+\left(\frac{t_1}{t_1-1}+\frac{t_2}{1-t_2}\right)\lambda}{\Dp_{t_2}\reli{2}}.
        \label{eq:rate_largedev_hi_rate}
    \end{align}
    Rearranging this gives
    \begin{align}
        &\sup_{t_1>1}-\Dm_{t_1}\reli{1}+\frac{t_1}{1-t_1}\lambda\notag\\
        &\qquad\qquad>~ 
        \inf_{0<t_2<1}-R\Dp_{t_2}\reli{2}+\frac{t_2}{1-t_2}\lambda.
    \end{align}
    Recalling the definition of $\Gamma$ from \cref{lem:ht_largedev}, this is equivalent to
    \begin{align}
        \Gamma_{+\lambda}\reli{1}>R\Gamma_{-\lambda/R}\reli{2}.
        \label{eq:rate_largedev_hi_opt}
    \end{align}
    The idea now is to connect this rate to the large deviation hypothesis testing quantities from \cref{lem:ht_largedev}. Consider hypothesis testing of the input/target with a type-I error log odds of $\pm\lambda$. \cref{lem:ht_largedev} gives for the input dichotomy
    \begin{align}
        \lim_{n\to\infty} \frac 1n \gamma_{+\lambda n}\rel{\rho_1^{\otimes n}}{\sigma_1^{\otimes n}}=\Gamma_{+\lambda}\reli{1},
    \end{align}
    and for the target dichotomy
    \begin{aligns}
        &\lim_{n\to\infty} \frac 1n \gamma_{-\lambda n}\rel{\rho_2^{\otimes Rn}}{\sigma_1^{\otimes Rn}}\notag
        \\
        &\qquad\qquad\qquad
        =\lim_{m\to\infty} \frac Rm \gamma_{-\lambda m/R}\rel{\rho_2^{\otimes m}}{\sigma_1^{\otimes m}}\\
        &\qquad\qquad\qquad=R\Gamma_{-\lambda/R}\reli{2}.
    \end{aligns}
    We can put the above limits back in terms of the type-II error \emph{probability} as
    \begin{subequations}
    \label{eq:rate_largedev_hi_beta}
    \begin{align}
        \lim_{n\to\infty}\frac 1n L\left[ \beta_{L^{-1}[-\lambda n]}\rel{\rho_1^{\otimes n}}{\sigma_1^{\otimes n}} \right]&=\Gamma_{+\lambda}\reli{1},\\
        \lim_{n\to\infty}\frac 1n L\left[ \beta_{L^{-1}[+\lambda n]}\rel{\rho_2^{\otimes Rn}}{\sigma_2^{\otimes Rn}} \right]&=R\Gamma_{-\lambda/R}\reli{2},
    \end{align}
    \end{subequations}
    where we recall that $L[x]:=\log\frac{x}{1-x}$. As $\epsilon_n$ is exponentially approaching $1$ with an exponent of $\lambda$, $(1\pm \epsilon_n)/2$ are exponentially approaching $0$ and $1$ respectively, both also with an exponent of $\lambda$. Putting this in terms of log odds as we did in \cref{subsubsec:ht_large}, this means that 
    \begin{align}
        \lim_{n\to\infty}\frac 1nL\left[\frac{1\pm \epsilon_n}2\right]= \pm \lambda,
    \end{align}
    Using the uniformity of the large deviation analysis of hypothesis testing shown in \cref{lem:ht_largedev_uni} then implies that \cref{eq:rate_largedev_hi_beta} can be extended to 
    \begin{aligns}
        \lim_{n\to\infty}\frac 1n L\left[ \beta_{\frac{1+\epsilon_n}2}\rel{\rho_1^{\otimes n}}{\sigma_1^{\otimes n}} \right]&= \Gamma_{+\lambda}\reli{1},\\
        \lim_{n\to\infty}\frac 1n L\left[ \beta_{\frac{1-\epsilon_n}2}\rel{\rho_2^{\otimes Rn}}{\sigma_2^{\otimes Rn}} \right]&= R\Gamma_{-\lambda/R}\reli{2}.
    \end{aligns}
    Recalling back to \cref{eq:rate_largedev_hi_opt}, and using the monotonicity of $L\left[\cdot\right]$, then we have
    \begin{align}
        \beta_{x_n}\rel{\rho_1^{\otimes n}}{\sigma_1^{\otimes n}} \gtev
        \beta_{x_n-\epsilon_n}\rel{\rho_2^{\otimes Rn}}{\sigma_2^{\otimes Rn}}.
    \end{align}
    for $x_n:=\frac{1+\epsilon_n}{2}$. By \cref{lem:ht} this means that transformation at a rate of $R$ is eventually \emph{not} possible. As this held for \emph{any} $R$ above satisfying \cref{eq:rate_largedev_hi_rate} this means that this implies a corresponding upper bound on the optimal rate,
    \begin{align}
        \limsup_{n\to\infty }R_n^*(\epsilon_n)\leq \inf_{\substack{t_1>1\\0<t_2<1}}\frac{\Dm_{t_1}\reli{1}+\left(\frac{t_1}{t_1-1}+\frac{t_2}{1-t_2}\right)\lambda}{\Dp_{t_2}\reli{2}},
    \end{align}
    as required.
    
    Next consider a rate $r$ such that
    \begin{align}
        r< \inf_{\substack{t_1>1\\0<t_2<1}}\frac{\DL_{t_1}\reli{1}+\left(\frac{t_1}{t_1-1}+\frac{t_2}{1-t_2}\right)\lambda}{D_{t_2}\reli{2}}.
        \label{eq:rate_largedev_hi_rate2}
    \end{align}
    Similar to optimality, this can be rearranged into the inequality
    \begin{align}
        \GammaL_{+\lambda}\reli{1}<r\Gamma_{-\lambda/r}\reli{2},
    \end{align}
    and so
    \begin{align}
        \gammaL_{+\lambda n} \rel{\rho_1^{\otimes n}}{\sigma_1^{\otimes n}} \ltev \gamma_{-\lambda n} \rel{\rho_2^{\otimes rn}}{\sigma_2^{\otimes rn}},
    \end{align}
    or equivalently
    \begin{align}
        \betaL_{L^{-1}\left[+\lambda n\right]} \rel{\rho_1^{\otimes n}}{\sigma_1^{\otimes n}} \ltev \beta_{L^{-1}\left[-\lambda n\right]} \rel{\rho_2^{\otimes rn}}{\sigma_2^{\otimes rn}}.
    \end{align}
    Recalling that $\epsilon_n=1-\exp(-\lambda n)$, we have
    \begin{aligns}
        \lim_{n\to\infty} \frac 1n L\left[1-\epsilon_n\right]&=-\lambda,\\
        \lim_{n\to\infty} \frac 1n L\left[\epsilon_n\right]&=+\lambda,
    \end{aligns}
    and therefore we eventually have
    \begin{align}
        \betaL_{\epsilon_n} \rel{\rho_1^{\otimes n}}{\sigma_1^{\otimes n}} \ltev \beta_{1-\epsilon_n} \rel{\rho_2^{\otimes rn}}{\sigma_2^{\otimes rn}}.
    \end{align}
    Lastly we use monotonicity of the type-II error, which allows us to relax this to the more broader inequality
    \begin{align}
        \betaL_{x} \rel{\rho_1^{\otimes n}}{\sigma_1^{\otimes n}} \ltev \beta_{x-\epsilon_n} \rel{\rho_2^{\otimes rn}}{\sigma_2^{\otimes rn}},
    \end{align}
    for all $x\in(\epsilon_n,1)$. We note that unlike with the proof of the small deviation rate, \cref{lem:ht_smalldev}, there is no concern about the uniformity of the asymptotic analysis of hypothesis testing, as we have only need to apply the large deviation analysis \cref{lem:ht_largedev} for a single error exponent~$\lambda$. 
    This diversion aside, we can now apply \cref{lem:ht}, which allows us to conclude that the rate $r$ \emph{is} eventually achievable. As this was true for any rate of the form \cref{eq:rate_largedev_hi_rate2} this implies a corresponding lower bound on the optimal rate,
    \begin{align}
        \liminf_{n\to\infty }R_n^*(\epsilon_n)\geq \inf_{\substack{t_1>1\\0<t_2<1}}\frac{\Dm_{t_1}\reli{1}+\left(\frac{t_1}{t_1-1}+\frac{t_2}{1-t_2}\right)\lambda}{D_{t_2}\reli{2}}.
    \end{align}
\end{proof}

Next, we turn to the trickier case of low-error. Before stating and proving the result, we will need some definitions. As we saw in the high-error case, the proof came down to satisfying inequalities of the form
\begin{aligns}
    \Gamma_{+\lambda}\rel{\rho_1^{\otimes n}}{\rho_2^{\otimes n}}
    &\leq 
    R\Gamma_{-\lambda/R}\rel{\rho_2^{\otimes Rn}}{\rho_2^{\otimes Rn}},\\
    \GammaL_{+\lambda}\rel{\rho_1^{\otimes n}}{\rho_2^{\otimes n}}
    &\leq 
    R\Gamma_{-\lambda/R}\rel{\rho_2^{\otimes Rn}}{\rho_2^{\otimes Rn}}
    ,\\
    \GammaR_{+\lambda}\rel{\rho_1^{\otimes n}}{\rho_2^{\otimes n}}
    &\leq 
    R\Gamma_{-\lambda/R}\rel{\rho_2^{\otimes Rn}}{\rho_2^{\otimes Rn}}
    .
\end{aligns}
In the low-error case, we will need inequalities of the form
\begin{aligns}
    \Gamma_{+\mu}\rel{\rho_1^{\otimes n}}{\rho_2^{\otimes n}}
    &\leq 
    R\Gamma_{+\mu/R}\rel{\rho_2^{\otimes Rn}}{\rho_2^{\otimes Rn}},\\
    \GammaL_{+\mu}\rel{\rho_1^{\otimes n}}{\rho_2^{\otimes n}}
    &\leq 
    R\Gamma_{+\mu/R}\rel{\rho_2^{\otimes Rn}}{\rho_2^{\otimes Rn}}
    ,\\
    \GammaR_{+\mu}\rel{\rho_1^{\otimes n}}{\rho_2^{\otimes n}}
    &\leq 
    R\Gamma_{+\mu/R}\rel{\rho_2^{\otimes Rn}}{\rho_2^{\otimes Rn}}
    .
\end{aligns}
Importantly, for the low-error case, we will need to satisfy these not just for a single $\mu$, but for all $-\lambda\leq \mu\leq \lambda$ (see the below proof for details). As such, it will be helpful to define the rates which saturate the above inequalities. Specifically, let $\overline r(\mu)$, $\lefthat r(\mu)$ and $\righthat r(\mu)$ denote the largest rates satisfying these inequalities, in the non-pinched/left-pinched/right-pinched cases respectively. By expanding the definitions of $\Gamma$/$\GammaL$/$\GammaR$, one can come up with explicit formulations of these rates, which share their piece-wise structure, specifically
\begin{aligns}
    \overline r(\mu):=&\begin{dcases}
        r_1(\mu) & \mu<-D\rel{\sigma_1}{\rho_1},\\
        \overline r_2(\mu) & -D\rel{\sigma_1}{\rho_1}<\mu<0,\\
        r_3(\mu) & \mu>0,
    \end{dcases}\\
    \lefthat r(\mu):=&\begin{dcases}
        \lefthat r_1(\mu) & \mu<-\Dstar\rel{\sigma_1}{\rho_1},\\
        \lefthat r_2(\mu) & -\Dstar\rel{\sigma_1}{\rho_1}<\mu<0,\\
        r_3(\mu) & \mu>0,
    \end{dcases}\\
    \righthat r(\mu):=&\begin{dcases}
        r_1(\mu) & \mu<-D\rel{\sigma_1}{\rho_1},\\
        \righthat r_2(\mu) & -D\rel{\sigma_1}{\rho_1}<\mu<0,\\
        \righthat r_3(\mu) & \mu>0,
    \end{dcases}
\end{aligns}
where 
\begin{aligns}
    r_1(\mu):=&\sup_{t_2<0}\inf_{t_1<0}\frac{-\Dm_{t_1}\reli{1}+\left(\frac{t_1}{t_1-1}-\frac{t_2}{t_2-1}\right)\mu}{-\Dm_{t_2}\reli{2}},\\
    \overline r_2(\mu):=&\inf_{0<t_2<1}\sup_{0<t_1<1}\frac{\Dp_{t_1}\reli{1}+\left(\frac{t_1}{1-t_1}-\frac{t_2}{1-t_2}\right)\mu}{\Dp_{t_2}\reli{2}},\\
    r_3(\mu):=&\sup_{t_2>1}\inf_{t_1>1}\frac{\Dm_{t_1}\reli{1}+\left(\frac{t_1}{t_1-1}-\frac{t_2}{t_2-1}\right)\mu}{\Dm_{t_2}\reli{2}},
\end{aligns}
and
\begin{aligns}
    \lefthat r_1(\mu):=&\sup_{t_2<0}\inf_{t_1<0}\frac{-\DL_{t_1}\reli{1}+\left(\frac{t_1}{t_1-1}-\frac{t_2}{t_2-1}\right)\mu}{-D_{t_2}\reli{2}},\\
    \lefthat r_2(\mu):=&\inf_{0<t_2<1}\sup_{0<t_1<1}\frac{\DL_{t_1}\reli{1}+\left(\frac{t_1}{1-t_1}-\frac{t_2}{1-t_2}\right)\mu}{D_{t_2}\reli{2}},\\
    \righthat r_2(\mu):=&\inf_{0<t_2<1}\sup_{0<t_1<1}\frac{\DR_{t_1}\reli{1}+\left(\frac{t_1}{1-t_1}-\frac{t_2}{1-t_2}\right)\mu}{D_{t_2}\reli{2}},\\
    \righthat r_3(\mu):=&\sup_{t_2>1}\inf_{t_1>1}\frac{\DR_{t_1}\reli{1}+\left(\frac{t_1}{t_1-1}-\frac{t_2}{t_2-1}\right)\mu}{D_{t_2}\reli{2}}.
\end{aligns}
Lastly, similar to $\overline r(\mu)$, we define $\widecheck r(\mu)$ as
\begin{align}
    \widecheck r(\mu):=&\begin{dcases}
        r_1(\mu) & \mu<-D\rel{\sigma_1}{\rho_1},\\
        \widecheck r_2(\mu) & -D\rel{\sigma_1}{\rho_1}<\mu<0,\\
        r_3(\mu) & \mu>0,
    \end{dcases}
\end{align}
where
\begin{align}
    \widecheck r_2(\mu):=&\inf_{0<t_2<1}\sup_{0<t_1<1}\frac{\Dm_{t_1}\reli{1}+\left(\frac{t_1}{1-t_1}-\frac{t_2}{1-t_2}\right)\mu}{D_{t_2}\reli{2}}.
\end{align}

These definitions in hand, we can turn to the low-error rate.

\ratelargedevlo*
\begin{proof}
    We once again start with optimality. Let \mbox{$0<\delta<\lambda$} be a small constant, and consider a rate $R$ such that
    \begin{align}
        R>\min_{-\lambda+\delta\leq \mu\leq \lambda-\delta}\overline r(\mu).
    \end{align}
    This means that there exists a $-\lambda< \mu^*<\lambda$ such that $R>\overline r(\mu^*)$, and therefore that
    \begin{align}
        \Gamma_{\mu^*}\reli{1}>R\Gamma_{\mu^*/R}\reli{2}.
    \end{align}
    By a similar chain of reasoning to that used in the proof of \cref{thm:rate_largedev_hi}, this implies
    \begin{align}
        \beta_{L^{-1}[\mu^*n]}\rel{\rho_1^{\otimes n}}{\sigma_1^{\otimes n}} \gtev \beta_{L^{-1}[\mu^*n]}\rel{\rho_2^{\otimes Rn}}{\sigma_2^{\otimes Rn}}.
    \end{align}
    Let $x_n:=L^{-1}[\mu^* n]$. As $\mu^*>-\lambda$ we have that $x_n$ dominates over $\epsilon_n$, and so the monotonicity of $\beta_x$ allows us to relax this to 
    \begin{align}
        \beta_{x_n}\rel{\rho_1^{\otimes n}}{\sigma_1^{\otimes n}}\gtev \beta_{x_n-\epsilon_n}\rel{\rho_2^{\otimes Rn}}{\sigma_2^{\otimes Rn}}.
    \end{align}
    By \cref{lem:ht} this means that transformation at a rate of $R$ is eventually \emph{not} possible. If we now take $\delta\to 0^+$, this gives a upper bound on the optimal rate of
    \begin{align}
        \limsup_{n\to\infty} R_n^*(\epsilon_n)\leq \min_{-\lambda\leq \mu\leq \lambda}\overline r(\mu).
    \end{align}
    
    Now to achievability. Let $\delta>0$ be a small constant, and consider a rate $r$ such that
    \begin{align}
        r<\min_{-\lambda-\delta\leq\mu\leq\lambda+\delta}\lefthat r(\mu).
    \end{align}
    This means that
    \begin{align}
        \GammaL_{\mu}\reli{1}<r\Gamma_{\mu/r}\reli{2},
    \end{align}
    and thus by \cref{lem:ht_largedev}
    \begin{align}
        \betaL_{L^{-1}[\mu n]}\rel{\rho_1^{\otimes n}}{\sigma_1^{\otimes n}} \ltev \beta_{L^{-1}[\mu n]}\rel{\rho_2^{\otimes Rn}}{\sigma_2^{\otimes Rn}},
    \end{align}
    for any $-\lambda-\delta\leq\mu\leq\lambda+\delta$. Note that applying \cref{lem:ht_largedev} only gives this convergence pointwise, but if we swap this out for the uniform version (\cref{lem:ht_largedev_uni} given in \cref{app:uni}), then this can be strengthened to a uniform statement. Specifically, we get that for sufficiently large $n$, this holds for all $\mu$ such that $\abs{\mu}\leq \lambda+\delta$ in that range. Recalling that $\epsilon_n:=\exp(-\lambda n)$, and therefore corresponds to a log odds per copy of $-\lambda$, we can see that for any probability $y_n\in(\epsilon_n/2,1-\epsilon_n/2)$ we have
    \begin{align}
        \frac 1n L(y_n)\in [-\lambda,\lambda]\subset(-\lambda-\delta,\lambda+\delta)
    \end{align}
    for sufficiently large $n$. As such, we have
    \begin{align}
        \betaL_{y_n}\rel{\rho_1^{\otimes n}}{\sigma_1^{\otimes n}} \ltev \beta_{y_n}\rel{\rho_2^{\otimes Rn}}{\sigma_2^{\otimes Rn}},
    \end{align}
    for all such $y_n$. Using the monotonicity of $\beta_x$ and $\betaL_x$ allows us to relax this to
    \begin{align}
        \betaL_{y_n-\epsilon_n/2}\rel{\rho_1^{\otimes n}}{\sigma_1^{\otimes n}} \ltev \beta_{y_n+\epsilon_n/2}\rel{\rho_2^{\otimes Rn}}{\sigma_2^{\otimes Rn}}.
    \end{align}
    Lastly, we shift this by $x_n:=y_n-\epsilon/2$, which yields
    \begin{align}
        \betaL_{y_n-\epsilon_n/2}\rel{\rho_1^{\otimes n}}{\sigma_1^{\otimes n}} \ltev \beta_{y_n+\epsilon_n/2}\rel{\rho_2^{\otimes Rn}}{\sigma_2^{\otimes Rn}},
    \end{align}
    for $x_n\in(\epsilon_n,1)$. As in deriving this inequality we employed not just the pointwise \cref{lem:ht_largedev}, but the uniform \cref{lem:ht_largedev_uni}, we therefore have it uniformly, which allows us to utilise \cref{lem:ht}. This in turn tells us that transformation at rate $r$ \emph{is} eventually possible. Taking $\delta\to0^+$, this yields the corresponding lower bound on the optimal rate of 
    \begin{align}
        \liminf_{n\to\infty} R_n^*(\epsilon_n)\geq \min_{-\lambda\leq \mu\leq\lambda}\lefthat r(\mu).
    \end{align}

    Repeating the above argument for $\righthat r(\mu)$ also gives an achievability bound
    \begin{align}
        \liminf_{n\to\infty} R_n^*(\epsilon_n)\geq \min_{-\lambda\leq \mu\leq\lambda}\righthat r(\mu).
    \end{align}
    Combining these gives
    \begin{align}
        \liminf_{n\to\infty} R_n^*(\epsilon_n)\geq \min_{-\lambda\leq \mu\leq\lambda}\max\left\lbrace \lefthat r(\mu),\righthat r(\mu)\right\rbrace.
    \end{align}
    By applying \cref{lem:maxpinch} it can be shown that
    \begin{align}
    \max\lbrace \lefthat r(\mu), \righthat r(\mu) \rbrace=\widecheck r(\mu),
    \end{align}
    giving the final achievability bound
    \begin{align}
        \liminf_{n\to\infty} R_n^*(\epsilon_n)\geq \min_{-\lambda\leq \mu\leq\lambda}\widecheck r(\mu).
    \end{align}
\end{proof}


\subsubsection{Moderate deviation}
\label{subsubsec:rate_moderate}

So far we have considered constant error and exponentially decaying error, which leaves a gap of errors which decay sub-exponentially, known as the moderate deviation regime. Much like the large deviation case, this will contain a slightly easier high-error case and a slightly trickier low-error case, and the proof will follow as a streamlined version of the proof used for \cref{thm:rate_largedev_lo,thm:rate_largedev_hi}. Recall from \cref{eq:reversibility} that the reversibility parameter is defined as
\begin{align}
    \xi:=\frac{V(\rho_1\|\sigma_1)}{D(\rho_1\|\sigma_1)} \bigg/ \frac{V(\rho_2\|\sigma_2)}{D(\rho_2\|\sigma_2)}.
\end{align}
Then we have the following:
 \ratemoddev*
\begin{proof}
    We begin with the more involved low-error case of $R_n^*(\epsilon_n)$, returning to the high-error case of $R_n^*(1-\epsilon_n)$ at the end of the proof. As is customary, we start with optimality. Let \mbox{$0<\lambda'<\lambda$} be a constant, and consider a rate $R_n$ defined
    \begin{align}
        R_n:=\frac{D(\rho_1\|\sigma_1)
        -\abs{1-\xi^{-1/2}}\sqrt{2\lambda' V\reli{1}n^{a-1}}
        }{D(\rho_2\|\sigma_2)}.
    \end{align}
    Applying \cref{lem:ht_moddev} to the input state, we get
    \begin{align}
        \frac 1n \gamma_{\pm \lambda n^{a}}\rel{\rho_1^{\otimes n}}{\sigma_1^{\otimes n}}\simeq -D\reli{1}\mp\sqrt{2\lambda V\reli{1}n^{a-1}},
    \end{align}
    and for the target state we have
    \begin{aligns}
        &\frac1n \gamma_{\pm \lambda n^{a}}\rel{\rho_2^{\otimes R_nn}}{\sigma_2^{\otimes R_nn}}\notag\\
        &\qquad\simeq -R_nD\reli{1}\mp\sqrt{2\lambda R_nV\reli{2}n^{a-1}}.
    \end{aligns}
    Taking a difference of these and expanding out $R_n$ gives
    \begin{aligns}
        &\frac 1n \gamma_{\pm \lambda n^{a}}\rel{\rho_1^{\otimes n}}{\sigma_1^{\otimes n}}-\frac1n \gamma_{\pm \lambda n^{a}}\rel{\rho_2^{\otimes R_nn}}{\sigma_2^{\otimes R_nn}}\notag\\
        &\qquad \simeq -\sqrt{2\lambda 'V\reli{1}n^{a-1}}\abs{1-\xi^{-1/2}}\\
        &\qquad\qquad\mp
        \sqrt{2\lambda V\reli{1}n^{a-1}}\left[1-\xi^{-1/2}\right]. \notag
    \end{aligns}
    So if we take $s:=\mathrm{sgn}(\xi^{-1/2}-1)$ to be the sign which makes this second term positive, then $\lambda'<\lambda$ tells us that this term must asymptotically dominate, and as such we can conclude
    \begin{align}
        \frac 1n \gamma_{s\lambda n^{a}}\rel{\rho_1^{\otimes n}}{\sigma_1^{\otimes n}}\gtev \frac1n \gamma_{s\lambda n^{a}}\rel{\rho_2^{\otimes R_nn}}{\sigma_2^{\otimes R_nn}}.
    \end{align}
    Recalling that $\epsilon_n:=\exp(-\lambda n^a)$, this means that $\epsilon_n$ has an asymptotic log odds per copy of $-\lambda n^a$, as does $2\epsilon_n$. Using this, we can re-express the above in terms of the type-II error probabilities as
    \begin{aligns}
        \xi>1:& & \beta_{2\epsilon_n}\rel{\rho_1^{\otimes n}}{\sigma_1^{\otimes n}}&\gtev \beta_{\epsilon_n}\rel{\rho_2^{\otimes R_nn}}{\sigma_2^{\otimes R_nn}},\\
        \xi<1:& & \beta_{1-\epsilon_n}\rel{\rho_1^{\otimes n}}{\sigma_1^{\otimes n}}&\gtev \beta_{1-2\epsilon_n}\rel{\rho_2^{\otimes R_nn}}{\sigma_2^{\otimes R_nn}}.
    \end{aligns}
    By \cref{lem:ht} this means that transformation at a rate of $R_n$ is asymptotically not possible, and thus that
    \begin{align}
        R^*_n(\epsilon_n)&\lesssim \frac{
        D(\rho_1\|\sigma_1)
        -\abs{1-\xi^{-1/2}}\sqrt{2\lambda V\reli{1}n^{a-1}}
        }{D(\rho_2\|\sigma_2)}.
    \end{align}
    
    The achievability proof follows similarly. The idea is that the bottleneck will once appear at log odds of $\pm \lambda n^{a-1}$, and to satisfy both the rate will require an absolute value around the term $1-\xi^{1/2}$. For this to give achievability we will need to use the uniform version of the moderate deviation analysis of hypothesis testing (\cref{lem:ht_moddev_uni}, presented in \cref{app:uni}).
    
    As for the high-error case this sign issue does not arise. In this case we can follow an approach similar to the achievability proof of \cref{thm:rate_first}. By using \cref{lem:ht_moddev} we can show that 
    \begin{align}
        \beta_{\frac{1+\epsilon_n}2}\rel{\rho_1^{\otimes n}}{\sigma_1^{\otimes n}}\gtev \beta_{\frac{1-\epsilon_n}2}\rel{\rho_2^{\otimes R_n n}}{\sigma_1^{\otimes R_n n}}
    \end{align}
    for an appropriately chosen rate $R_n$ which will yield the optimality bound, and for achievability we first show
    \begin{align}
        \betaL_{\epsilon_n}\rel{\rho_1^{\otimes n}}{\sigma_1^{\otimes n}}\ltev \beta_{1-\epsilon_n}\rel{\rho_2^{\otimes r_n n}}{\sigma_1^{\otimes r_n n}},
    \end{align}
    and then use monotonicity to extend this to the ordering required by \cref{lem:ht}.
\end{proof}


\subsubsection{Extreme deviation}
\label{subsubsec:rate_extreme}

The argument for the zero-error case follows similarly to the low-error large deviation case.

\ratezero*
\begin{proof}
    As this is the zero-error case, the optimality side is pretty straightforward. Any additive and data-processing quantity $Q\rel{\cdot}{\cdot}$ puts a single-shot bound on the largest possible transformation rate for all $n$ of the form
    \begin{align}
        R_n^*(\epsilon) \leq \frac{Q\reli{1}}{Q\reli{2}}.
    \end{align}
    If we consider the minimal relative entropies $\Dm_\alpha$, this gives
    \begin{align}
        R_n^*(0) \leq 
        \min_{\alpha\in\overline{\mathbb R}} \frac{\Dm_\alpha\reli{1}}{\Dm_\alpha\reli{2}}.
    \end{align}
    In the case of coherent outputs one could include other possible monotones $Q$, which could constrain the zero-error rate further. 

    Now, we turn to the tricky part, achievability. Consider a rate constant $r$ such that 
    \begin{align}
        r<\inf_{\alpha\in\mathbb R} \frac{\DL_\alpha\reli{1}}{D_\alpha\reli{2}}.
    \end{align}
    We note that this rate is almost of the form we want, but involves the pinched relative entropy and not the minimal, and is therefore suboptimal---we will return to this. We want to prove that
    \begin{align}
        \betaL_x\rel{\rho_1^{\otimes n}}{\sigma_1^{\otimes n}}\stackrel{!}<\beta_x\rel{\rho_2^{\otimes rn}}{\sigma_2^{\otimes rn}}
    \end{align}
    eventually holds for all $x$. To do this, we will need to combine both the extreme and large deviation analysis.

    First, we start with high errors. Noticing that \mbox{$r<\DL_{+\infty}\reli{1}/D_{+\infty}\reli{2}$} and recalling that $\DL_\alpha$ is defined as the pinched-and-regularised relative entropy, this means that
    \begin{align}
        r\ltev\frac{
        D_{+\infty}
        \rel{\pinch{\rho_1^{\otimes n}}{\sigma_1^{\otimes n}}}{\sigma_1^{\otimes n}}}{nD_{+\infty}\reli{2}}.
    \end{align}
    The pinching inequality gives 
    \begin{align}
        \lambda_{\min}\left(\pinch{\rho_1^{\otimes n}}{\sigma_1^{\otimes n}}\right) \geq\frac{\lambda_{\min}(\rho)^n}{\abs{\mathrm{spec}(\sigma^{\otimes n})}} \geq \lambda_{\min}^{dn}(\rho),
    \end{align}
    so any \mbox{$x\leq \min\lbrace\lambda_{\min}^d(\rho_1),\lambda_{\min}^{r}(\rho_2)\rbrace^n$}
    satisfies
    \begin{align}
        x\leq \lambda_{\min}\left(\pinch{\rho_1^{\otimes n}}{\sigma_1^{\otimes n}}\right)
        ~~\text{and}~~
        x\leq \lambda_{\min}\left(\sigma_1^{\otimes rn}\right),
    \end{align}
    and so we can apply \cref{lem:ht_extremedev_singleshot} to both states, giving 
    \begin{aligns}
        \betaL_{1-x}\rel{\rho_1^{\otimes n}}{\sigma_1^{\otimes n}}
        &=x\exp\left(-D_{+\infty}
        \rel{\pinch{\rho_1^{\otimes n}}{\sigma_1^{\otimes n}}}{\sigma_1^{\otimes n}}\right)\\
        &\ev{<}x\exp\left(-rnD_{+\infty}\reli{2}\right)\\
        &=x\exp\left(-\Dm_{+\infty}\rel{\rho_2^{\otimes rn}}{\sigma_2^{\otimes rn}}\right)\\
        &=\beta_{1-x}\rel{\rho_2^{\otimes rn}}{\sigma_2^{\otimes rn}}.
    \end{aligns}
    Similarly, for the low-error case, we can use \mbox{$r<\DL_{-\infty}\reli{1}/D_{-\infty}\reli{2}$}, which gives for sufficiently large $n$ that
    \begin{aligns}
        1-\betaL_{x}\rel{\rho_1^{\otimes n}}{\sigma_1^{\otimes n}}
        &=x\exp\left(-D_{-\infty}
        \rel{\pinch{\rho_1^{\otimes n}}{\sigma_1^{\otimes n}}}{\sigma_1^{\otimes n}}\right)\\
        &\ev>x\exp\left(-rn\Dm_{-\infty}\reli{2}\right)\\
        &=1-\beta_{x}\rel{\rho_2^{\otimes rn}}{\sigma_2^{\otimes rn}}.
    \end{aligns}
    
    For the remaining range of $x$, we resort to the method used in the large deviation regime. As \mbox{$r<\DL_{\alpha}\reli{1}/D_{\alpha}\reli{2}$}, we have
    \begin{align}
        \GammaL_{\lambda}\reli{1}<r\Gamma_{\lambda r}\reli{2}
    \end{align}
    for all $\lambda$. Using \cref{lem:ht_largedev}, this means that
    \begin{align}
        \betaL_{L^{-1}[\lambda n]}\rel{\rho_1^{\otimes n}}{\sigma_1^{\otimes n}}
        \ltev \beta_{L^{-1}[\lambda n]}\rel{\rho_2^{\otimes rn}}{\sigma_2^{\otimes rn}}.
    \end{align}
    This is only a pointwise convergence which is insufficient for achievability but, similarly to the proof of \cref{thm:rate_largedev_lo}, we can leverage the uniform analysis of \cref{lem:ht_largedev_uni} to show that this inequality must eventually hold \emph{uniformly} for $\lambda$ on a closed interval. If we specifically consider the interval
    \begin{align}
        \abs\lambda\leq \max\left\lbrace -d\log\lambda_{\min}(\rho_1),-r\log\lambda_{\min}(\rho_2) \right\rbrace+1,
    \end{align}
    then this overlaps with the extreme deviation cases, and thus we have that for sufficiently large $n$
    \begin{align}
        \betaL_{x}\rel{\rho_1^{\otimes n}}{\sigma_1^{\otimes n}}
        < \beta_{x}\rel{\rho_2^{\otimes rn}}{\sigma_2^{\otimes rn}}
    \end{align}
    holds for all $x\in(0,1)$. Applying \cref{lem:ht} gives that transformation at rate $r$ is eventually possible, and so
    \begin{align}
        \liminf_{n\to\infty}R_n^*(0)\geq \inf_{\alpha\in\mathbb R} \frac{\DL_\alpha\reli{1}}{D_\alpha\reli{2}}.
        \label{eq:rate_zero_ach1}
    \end{align}

    {\ctc Similarly, if we consider the right-pinching we also get}
    \begin{align}
        \liminf_{n\to\infty}R_n^*(0)
        &\geq \inf_{\alpha\in\mathbb R} \frac{\DR_\alpha\reli1}{D_\alpha\reli2}.
    \end{align}
    Now combining both achievability results for left- and right-pinching, and recalling \cref{lem:maxpinch}, gives
    \begin{align}
        \liminf_{n\to\infty}R_n^*(0)
        &\geq \max\left\lbrace
        \inf_{\alpha\in\mathbb R} \frac{\DL_\alpha\reli1}{D_\alpha\reli2},
        \inf_{\alpha\in\mathbb R} \frac{\DR_\alpha\reli1}{D_\alpha\reli2}
        \right\rbrace,
    \end{align}
    {\ctc as required. Note this is quite close to the achievability which, due to \cref{lem:maxpinch}, can be rewritten as
    \begin{align}
        \limsup_{n\to\infty}R_n^*(0)
        &\leq \inf_{\alpha\in\mathbb R}\max\left\lbrace
        \frac{\DL_\alpha\reli1}{D_\alpha\reli2},
        \frac{\DR_\alpha\reli1}{D_\alpha\reli2}
        \right\rbrace.
    \end{align}
    }
\end{proof}

In \cref{thm:thermo} it was noted that all of the achievability results in this paper were, in the thermodynamic setting, achievable with only thermal operations \emph{except} \cref{thm:rate_zero}, which requires Gibbs-preserving maps. While all the achievability results in this paper leverage pinching, which is itself a thermal operation (see \cref{app:thermal}), the problem arose in this final step involving pinching either the first or second state. In the case where switching the pinching is unnecessary, then this is a thermal operation, but that is not generally the case.

Instead of a rate-based statement, we can also phrase this zero-error statement in terms of \emph{eventual} Blackwell ordering~\cite{Jensen_2019,MuPomattoStrackTamuz2019,FarooqFritzHaapasaloTomamichel2023}, in line with some of the existing papers looking at similar zero-error transformation questions. For a pair of dichotomies we define a notion of \emph{eventual} Blackwell ordering as an ordering which appears for a sufficiently large number of copies, i.e.\ \mbox{$\left(\rho_1^{\otimes n},\sigma_1^{\otimes n}\right)\ev \succeq \left(\rho_2^{\otimes n},\sigma_2^{\otimes n}\right)$} is a shorthand for 
\begin{align}
    \exists N:~\left(\rho_1^{\otimes n},\sigma_1^{\otimes n}\right)\succeq \left(\rho_2^{\otimes n},\sigma_2^{\otimes n}\right)~\forall n\geq N.
\end{align}

\begin{cor}[Eventual Blackwell ordering]
    Consider a pair of dichotomies $(\rho_1,\sigma_1)$ and $(\rho_2,\sigma_2)$. If the target is commuting, $[\rho_2,\sigma_2]=0$, and
    \begin{aligns}
        \ctc \DL_\alpha(\rho_1\|\sigma_1) 
        &> 
        \ctc D_\alpha(\rho_2\|\sigma_2)~~\forall\alpha\in\overline{\mathbb R},\\
        &\qquad\text{\ctc or}\notag\\
        \ctc \DR_\alpha(\rho_1\|\sigma_1) 
        &>
        \ctc D_\alpha(\rho_2\|\sigma_2)~~\forall\alpha\in\overline{\mathbb R},
    \end{aligns}
    then $\left(\rho_1^{\otimes n},\sigma_1^{\otimes n}\right)\ev \succeq \left(\rho_2^{\otimes n},\sigma_2^{\otimes n}\right)$. Moreover, if $\left(\rho_1^{\otimes n},\sigma_1^{\otimes n}\right)\ev \succeq \left(\rho_2^{\otimes n},\sigma_2^{\otimes n}\right)$ then this implies the inequalities
    \begin{align}
        \Dm_\alpha(\rho_1\|\sigma_1)&\geq \Dm_\alpha(\rho_2\|\sigma_2)~~\forall\alpha\in\overline{\mathbb R},
    \end{align}
    even for non-commuting targets $[\rho_2,\sigma_2]\neq 0$.
\end{cor}
\begin{proof}
    This follows directly from considering the $R=1$ cases of \cref{thm:rate_zero}. The inequalities
    \begin{aligns}
        \ctc \DL_\alpha(\rho_1\|\sigma_1) 
        &> 
        \ctc D_\alpha(\rho_2\|\sigma_2)~~\forall\alpha\in\overline{\mathbb R},\\
        &\qquad\text{\ctc or}\notag\\
        \ctc \DR_\alpha(\rho_1\|\sigma_1) 
        &>
        \ctc D_\alpha(\rho_2\|\sigma_2)~~\forall\alpha\in\overline{\mathbb R},
    \end{aligns}
    give that the zero-error rate is strictly greater than unity, $R_n^*(0)\ev> 1$, and the inequalities 
    \begin{align}
        \Dm_\alpha(\rho_1\|\sigma_1)\geq \Dm_\alpha(\rho_2\|\sigma_2)~~\forall\alpha\in\overline{\mathbb R}
    \end{align}
    all follow from the data-processing inequality of the minimal R\'enyi relative entropy.
\end{proof}

Lastly, for completeness, we consider the case of a super-exponentially high-error, wherein the asymptotic transformation rate is unbounded.

\rateextreme*
\begin{proof}
    Consider any constant rate $r$. As $\epsilon_n$ is super-exponentially approaching $1$ it must dominate any other expression approaching $1$ exponentially, specifically
    \begin{align}
        \epsilon_n\gtev 1-\left(\lambda_{\min}(\rho_1)/2\right)^{n}
        ~~~\text{and}~~~
        \epsilon_n\gtev 1-\lambda_{\min}(\rho_2)^{rn}.    
    \end{align}
   Thus, we have that $1-\epsilon_n\ltev (\lambda_{\min}(\rho_1)/2)^n$ and \mbox{$1-\epsilon_n\ltev \lambda_{\min}^{rn}(\rho_2)$}. Applying \cref{lem:ht_extremedev_singleshot} to the input gives
    \begin{align}
        \frac{\betaL_{\epsilon_n}\rel{\rho_1^{\otimes n}}{\sigma_1^{\otimes n}} }{1-\epsilon_n}
        &=\exp\left(-D_{+\infty}\rel{\pinch{\rho^{\otimes n}}{\sigma^{\otimes n}}}{\sigma^{\otimes n}}\right),
    \end{align}
    and to the target gives
    \begin{align}
        \frac{1-\beta_{1-\epsilon_n}\rel{\rho_2^{\otimes rn}}{\sigma_2^{\otimes rn}}}{1-\epsilon_n}&=\exp\left(-rn\Dm_{-\infty}\reli{2}\right).
    \end{align}
    As $n\to\infty$ these type-II errors approach $0$ and $1$ respectively, and so
    \begin{align}
        \betaL_{\epsilon_n}\rel{\rho_1^{\otimes n}}{\sigma_1^{\otimes n}}\ltev \beta_{1-\epsilon_n}\rel{\rho_2^{\otimes rn}}{\sigma_2^{\otimes rn}}.
    \end{align}
    Using monotonicity of $x\mapsto \beta_x(\cdot\|\cdot)$ allows us to relax this to 
    \begin{align}
        \betaL_{x}\rel{\rho_1^{\otimes n}}{\sigma_1^{\otimes n}}< \beta_{x-\epsilon_n}\rel{\rho_2^{\otimes rn}}{\sigma_2^{\otimes rn}}
    \end{align}
    for $x\in(\epsilon_n,1)$. So, by \cref{lem:ht} this means that transformation at the rate $r$ is eventually achievable. As this entire argument worked for any constant $r$, this therefore means that the optimal rate must diverge,
    \begin{align}
        \liminf_{n\to\infty}R_n^*(\epsilon_n)=\infty.
    \end{align}
\end{proof}


\section{Conclusions and outlook}
\label{sec:outlook}

In this work we have analysed one of the central problems of the theory of quantum statistical inference, namely that of comparing informativeness of two quantum dichotomies (which is directly related to transforming the first dichotomy into the second one). By focusing on the asymptotic version of the problem, we were able to solve it in various error regimes under the assumption that the second dichotomy is commutative. More precisely, we found optimal transformation rates between many copies of pairs of quantum states in the small, moderate, large and zero-error regimes. We then employed the obtained results to derive new thermodynamic laws for quantum systems prepared in coherent superpositions of energy eigenstates. Thus, for the first time, we were able to analyse the optimal performance of thermodynamic protocols with coherent inputs beyond the thermodynamic limit, and discussed new resonance phenomena that allow one to mitigate thermodynamic dissipation by, e.g., employing quantum coherence.

We believe that the success of employing quantum statistical inference techniques to accurately describe quantum thermodynamic transformations strongly motivates further exploration of the connections between the two frameworks. We propose the following three avenues. First, one of the problems within the resource-theoretic approach to quantum thermodynamics is the lack of techniques for addressing the regimes of non-independent systems. Interestingly, Refs.~\cite{HiaiMosonyiOgawa2007,MosonyiOgawa14} \kk{suggest} that the hypothesis testing approach can be effective for studying ensembles composed of weakly correlated states. Potentially, such techniques can be adapted to study the thermodynamic state transformation problem outside of the usual uncorrelated setting. Second, one could use quantum statistical inference techniques to develop explicit thermodynamic protocols. Indeed, one of the criticisms of the resource-theoretic approach is that many of its consequences are implicit, i.e., one often shows the existence of protocols but their explicit form is usually not possible to infer. However, as observed in Ref.~\cite{renes2016relative}, the hypothesis testing approach allows one to construct explicit thermal operations starting from the optimal measurement in the related hypothesis testing problem. We expect that investigating the explicit form of optimal thermodynamic protocols in different asymptotic regimes can lead to interesting new insights on the nature of fundamental limitations imposed by thermodynamic laws on dissipation, reversibility, work processes, etc. And third, it is a long-standing problem to connect the resource-theoretic approach to thermodynamics with more standard approaches~\cite{guarnieri2019quantum}. We believe that an especially interesting connection might exist between the resource-theoretic approach and so-called slow-driving protocols~\cite{scandi2020quantum}. Here, we note that both approaches use similar statistical and geometric techniques. For example, the optimal  thermodynamic protocols in the slow-driving regime can be quantified using the so-called Kubo-Mori metric, which is also related to the problem of hypothesis testing \cite{jarzyna2020geometric}. Exploring these intrinsic similarities might improve our understanding of quantum thermodynamics. 

On a more technical side, we think that a very interesting avenue for further research is to try to generalise our results so that they also apply to non-commutative output dichotomies~\cite{Misra2016}. This would open a way to study fully quantum laws of thermodynamics, where both initial and final states could be given by superpositions of different energy eigenstates. While we think that \cref{conj:quantum} may be true, this does not necessarily mean that the transformation rates in the fully coherent regime would be simple generalisations of the current results. The reason for that is that proving \cref{conj:quantum} would only guarantee such a simple generalisation of the rate under transformations with the so-called Gibbs-preserving operations~\cite{faist2015gibbs}, and not under thermal operations. In fact, we believe that in a fully quantum regime, there may be a gap between the rates achievable with these two sets of free operations (especially for the zero-error case).

Another technical generalisation of our result that we find highly interesting is to study transformations between \emph{multichotomies}, i.e., multipartite transformations from $m$ states to $m$ states, with dichotomies being the special case of $m=2$. The classical zero-error case of this has recently been analysed in Ref.~\cite{FarooqFritzHaapasaloTomamichel2023}, and the quantum and/or nonzero-error cases are natural generalisations worthy of study. Physically, such a result could help understanding transformations of quantum systems under the constraint of the symmetry. This is because for a symmetry group~$G$, the existence of a $G$-covariant quantum channel mapping the initial state to the final one is equivalent to the existence of an unconstrained channel mapping the orbit of the initial state to the orbit of the final state, with the orbit being generated by the symmetry elements of~$G$~\cite{marvianthesis}. 

Finally, there are two aspects of the resonance phenomena described in this paper that we believe deserve more attention. First, we think it would be very interesting to find the equivalent of the resonance phenomenon in more traditional approaches to thermodynamics, beyond the resource-theoretic treatment. In other words, we would like to investigate whether such a potential reduction of free energy dissipation may appear in actual physical processes when the parameters are tuned appropriately. Second, one could look for similar resonance effects in other resource theories. In particular, we note that pure state interconversion conditions in the resource theory of $U(1)$-asymmetry~\cite{gour2008resource} are ruled by a generalisation of the majorisation partial order (called cyclic majorisation in Ref.~\cite{szymanskithesis}). Since the resonance appeared for standard majorisation (as we have seen for pure bipartite entanglement transformations in this paper), the resource theory of $U(1)$-asymmetry seems to be a good candidate to look for novel resource resonance effects.

 \bigskip

\begin{acknowledgments}

   We would like to thank Francesco Buscemi for helpful pointers to the literature. KK acknowledges financial support from the Foundation for Polish Science through the TEAM-NET project (contract no. POIR.04.04.00-00-17C1/18-00). PLB/CTC/JMR acknowledge the Swiss National Science Foundation for financial support through the NCCRs SwissMAP and QSIT, and CTC/JMR acknowledge their support through Sinergia Grant CRSII5\_186364. MT acknowledges the hospitality of the Pauli Center who supported a long-term visit to ETH Z\"urich. He is also supported by NUS startup grants R-263-000-E32-133 and R-263-000-E32-731.

\end{acknowledgments}

\renewcommand{\v}[1]{\vv{#1}}
%

\renewcommand{\v}[1]{\ensuremath{\boldsymbol #1}}

\onecolumngrid
\appendix



\section[Sesquinormal distribution properties]{\texorpdfstring{Proof of \cref{lem:sesqui}}{Sesquinormal distribution properties}}
\label{app:sesqui}
\sesqui*
\begin{proof}
To prove that $S_\nu$ is a valid cdf, we need to prove that it is continuous, monotone {\ctc non-decreasing}, and has the limits 
\begin{equation}
    \lim_{\mu\to-\infty}S_\nu(\mu)=0 \quad\mathrm{and}\quad \lim_{\mu\to+\infty}S_\nu(\mu)=1.    
\end{equation}
We will see that continuity and the limits both follow from the closed-form below. For monotonicity we can use the fact that the total variation distance of two distributions is unchanged if we shift them along $\mathbb R$,
\begin{equation}
    S_\nu(\mu+\epsilon)
    =\inf_{A\geq \Phi_{0,1}}T\left(A,\Phi_{\mu+\epsilon,\nu}\right)=\inf_{A\geq \Phi_{-\epsilon,1}}T\left(A,\Phi_{\mu,\nu}\right)\geq \inf_{A\geq \Phi_{0,1}}T\left(A,\Phi_{\mu,\nu}\right)=S_\nu(\mu).
\end{equation}
Next, we turn to a closed form of $S_\nu$. Recall the definition,
\begin{align} 
    S_\nu(\mu):=\frac 12 \inf_{A\geq \Phi}\int_{\mathbb R }\abs{A'(x)-\phi_{\mu,\nu}(x)}\mathrm dx.
\end{align}
To find a closed-form, we will suggest a candidate $A$, evaluate its total variation distance, and then construct a lower bound to show that this is optimal among distributions with $A{\ctc\geq}\Phi$. We split into two cases based on whether $\nu\lessgtr 1$.

\paragraph{Case $0<\nu< 1$:}
Start by recalling that the total variation distance between two measures is the largest possible difference in probability that they assign to an event, i.e.,
\begin{align}
    S_\nu(\mu)=\inf_{A\geq \Phi}\sup_{R\subseteq\mathbb R} \int_R\left(A'(x)-\phi_{\mu,\nu}(x)\right)\mathrm dx.
\end{align}
Next, consider the set of $x$ such that \mbox{$\Phi(x)\geq \Phi_{\mu,\nu}(x)$} \emph{and} \mbox{$\phi(x)\geq \phi_{\mu,\nu}(x)$}. For $\nu<1$ this region is given precisely by \mbox{$x\leq X$} where
\begin{align}
    X=\frac{\mu-\sqrt\nu \sqrt{\mu^2+(\nu-1)\ln\nu}}{1-\nu}.
\end{align}
Thus, we can lower bound the total variation distance by considering the region $R=(-\infty,X]$, which gives
\begin{aligns}
    S_{\nu}(\mu)
    \geq& \inf_{A\geq \Phi}\int_{-\infty}^{X} \bigl( A'(x)-\phi_{\mu,\nu}(x) \bigr)\,\mathrm dx\\
    =&\inf_{A\geq \Phi}A(X)-\Phi_{\mu,\nu}(X)\\
    \geq& \Phi(X)-\Phi_{\mu,\nu}(X)\\
    =& \Phi\left(\frac{\mu-\sqrt\nu \sqrt{\mu^2+(\nu-1)\ln\nu}}{1-\nu}\right)
    -\Phi_{\mu,\nu}\left( \frac{\mu-\sqrt\nu\sqrt{\mu^2+(\nu-1)\ln\nu}}{1-\nu} \right)\\
    =& \Phi\left(\frac{\mu-\sqrt\nu\sqrt{\mu^2+(\nu-1)\ln\nu}}{1-\nu}\right)
    -\Phi\left( \frac{\sqrt{\nu}\mu-\sqrt{\mu^2+(\nu-1)\ln\nu}}{1-\nu} \right).
\end{aligns}
Moreover, it can be seen that by taking \mbox{$A(x):=\max\lbrace \Phi(x),\Phi_{\mu,\nu}(x)\rbrace$} we can saturate this lower bound, proving it to be optimal among all cdfs such that $A\geq \Phi$.

\paragraph{Case $\nu> 1$:}
For $\nu>1$ we can do a similar proof to $\nu<1$. Here we are interested in the region in which $\Phi(x)\geq \Phi_{\mu,\nu}(x)$ \emph{and} $\phi(x)\leq \phi_{\mu,\nu}(x)$, which is now given by $x\geq X$, with $X$ defined as before. Now, looking at the lower bound given by $R=[X,\infty)$, we get the same as previously
\begin{aligns}
    S_{\nu}(\mu)
    \geq& \inf_{A\geq \Phi}\int_{X}^{\infty} \bigl( \phi_{\mu,\nu}(x)-A'(x) \bigr)\,\mathrm dx\\
    =&\inf_{A\geq \Phi}\bigl(1-\Phi_{\mu,\nu}(X)\bigr)-\bigl(1-A(X)\bigr)\\
    =&\inf_{A\geq \Phi}A(X)-\Phi_{\mu,\nu}(X)\\
    \geq& \Phi(X)-\Phi_{\mu,\nu}(X)\\
    =& \Phi\left(\frac{\mu-\sqrt\nu\sqrt{\mu^2+(\nu-1)\ln\nu}}{1-\nu}\right)
    -\Phi\left( \frac{\sqrt{\nu}\mu-\sqrt{\mu^2+(\nu-1)\ln\nu}}{1-\nu} \right).
\end{aligns}
Once again, optimality of this bound is implied by the fact it is still saturated by \mbox{$A(x):=\max\lbrace \Phi(x),\Phi_{\mu,\nu}(x)\rbrace$}.

Now, by taking limits of the closed form we can see that
\begin{align}
    \lim_{\nu\to 0^+} S_\nu(\mu)=\lim_{\nu\to\infty}S_\nu(\sqrt \nu\mu)=\Phi(\mu)
    \qquad\text{and}\qquad
    \lim_{\nu\to 1} S_\nu(\mu)=\max\lbrace 2\Phi(\mu/2)-1,0\rbrace,
\end{align}
and by substituting $\nu\to1/\nu$ it can be straightforwardly seen that this expression has the duality property
\begin{align}
    S_\nu(\mu)=S_{1/\nu}(\mu/\sqrt \nu).
\end{align}

Having a closed form of the cdf, we turn to the inverse cdf. To start with, consider an arbitrary $\epsilon\in(0,1)$, and define 
    \begin{align}
        \mu:=\inf_{x\in(\epsilon,1)}\sqrt{\nu}\Phi^{-1}(x)-\Phi^{-1}(x-\epsilon).
    \end{align}
    To prove the form of the inverse cdf $S_\nu^{-1}(\epsilon)=\mu$ it suffices to show the inverse expression $S_\nu(\mu)=\epsilon$. We start by noting that the function $f(x):=\sqrt{\nu}\Phi^{-1}(x)-\Phi^{-1}(x-\epsilon)$ is bounded for any $x\in(\epsilon,1)$ and diverges to $+\infty$ for either $x\to\epsilon^+$ and $x\to 1^-$, and as such the infimum is in fact a minimum, so
    \begin{align}
        \mu=\min_{x\in(\epsilon,1)}\sqrt{\nu}\Phi^{-1}(x)-\Phi^{-1}(x-\epsilon).
    \end{align}
    Thus, there must exist a $y\in{\ctc (\epsilon,1)}$ at which the infimum is attained, i.e., $f(y)=\mu$. By the Interior Extremum Theorem, we must have that this is a stationary point, $f'(y)=0$. The Inverse Function Rule allows us to evaluate the derivative of $f$ to be
    \begin{align}
    f'(x)=\frac{\sqrt\nu}{\phi\left(\Phi^{-1}(x)\right)}-\frac{1}{\phi\left(\Phi^{-1}(x-\epsilon)\right)},
    \end{align}
    and thus $f'(y)=0$ reduces to
    \begin{align}
        \left[\Phi^{-1}(y)\right]^2+\ln\nu=\left[\Phi^{-1}(y-\epsilon)\right]^2.
    \end{align}
    To get rid of this shifted Gaussian term, we can use a substitution
    \begin{align}
        \Phi^{-1}(y-\epsilon)=\sqrt\nu\Phi^{-1}(y)-\mu, \label{eqn:shift_gauss}
    \end{align}
    which gives us a quadratic expression for $\Phi^{-1}(y)$,
    \begin{align}
        \left[\Phi^{-1}(y)\right]^2+\ln\nu=\left[\sqrt\nu\Phi^{-1}(y)-\mu\right]^2
        ,
    \end{align}
    with a pair of solutions
    \begin{align}
        \Phi^{-1}(y)=-\frac{\mu\sqrt\nu\pm\sqrt{\mu^2+(\nu-1)\ln\nu}}{1-\nu}.
    \end{align}
    
    This, however, still has a lingering $\pm$ ambiguity. {\ctc If we rearrange \cref{eqn:shift_gauss} to make $\epsilon$ the subject, we get
    \begin{align}
        \epsilon=y-\Phi\left(\sqrt\nu\Phi^{-1}(y)-\mu\right).
    \end{align}
    Substituting our pair of solutions into this expression gives}
    \begin{aligns}
        \epsilon
        &={\ctc \Phi\left(\Phi^{-1}(y)\right)-\Phi\left(\sqrt\nu\Phi^{-1}(y)-\mu\right)   } \\
        &={\ctc \Phi\left(-\frac{\mu\sqrt\nu \pm \sqrt{\mu^2+(\nu-1)\ln\nu}}{1-\nu}\right)-\Phi\left(-\sqrt\nu\frac{\mu\sqrt\nu \pm \sqrt{\mu^2+(\nu-1)\ln\nu}}{1-\nu}-\mu\right)   } \\
        &=\Phi\left(\frac{\mu\mp\sqrt{\nu}\sqrt{\mu^2+(\nu-1)\ln\nu}}{1-\nu}\right)
        -\Phi\left(\frac{\sqrt\nu\mu\mp\sqrt{\mu^2+(\nu-1)\ln\nu}}{1-\nu}\right).
    \end{aligns}
    Now, we notice that $\sqrt{\mu^2+(\nu-1)\ln\nu}\geq \mu$. Using this, we can see that the positive solution for $\Phi^{-1}(x)$ will correspond to $\epsilon\leq 0$, and thus the minimising $x$ must correspond to the negative solution, i.e., 
    \begin{align}
        \epsilon
        &=\Phi\left(\frac{\mu-\sqrt{\nu}\sqrt{\mu^2+(\nu-1)\ln\nu}}{1-\nu}\right)
        -\Phi\left(\frac{\sqrt\nu\mu-\sqrt{\mu^2+(\nu-1)\ln\nu}}{1-\nu}\right).
    \end{align}
    Finally, we now have $\epsilon=S_\nu(\mu)$, as required.

\end{proof}

\begin{lemma}[Asymptotic expansions]
For $\mu\to \infty$ the sesquinormal cdf can be expanded as
\begin{aligns}
    \ln\left[S_{\nu}(-\mu)\right]&\approx -\frac 12 \left(\frac{\mu}{1-\sqrt\nu}\right)^2,\\
    \ln\left[1-S_{\nu}(\mu)\right]&\approx -\frac 12 \left(\frac{\mu}{1+\sqrt\nu}\right)^2.
\end{aligns}
Similarly, for $\epsilon\to 0^+$ the sesquinormal inverse cdf can be expanded as
\begin{aligns}
    S_\nu^{-1}(\epsilon)&\approx \abs{1-\sqrt \nu}\sqrt{2\log 1/\epsilon},\\
    S_\nu^{-1}(1-\epsilon)&\approx (1+\sqrt \nu)\sqrt{2\log 1/\epsilon}.
\end{aligns}
\end{lemma}
\begin{proof}
    The expansions of the cdf can be found simply by expanding the closed-form expression in \cref{lem:sesqui} to leading order in $\mu$, specifically using the approximation
    \begin{align}
        \sqrt{\mu^2+(\nu-1)\ln\nu}\approx \abs{\mu}.
    \end{align}
    Using this approximation and the $x\to \infty$ expansion $\ln\left[1-\Phi(x)\right]\approx -x^2/2$ gives the cdf expansions. By inverting this we can equivalently get the inverse cdf expansions as well.
\end{proof}


\section{Pinched relative entropy}
\label{app:pinch}

In this appendix, we will show the existence and properties of the pinched R\'{e}nyi relative entropies. Suppose the states $\rho,\sigma$ are fixed and full rank, and define
\begin{align}
    f_n(\alpha):=\frac 1n D_\alpha\rel{\pinch{\rho^{\otimes n}}{\sigma^{\otimes n}}}{\sigma^{\otimes n}}.
\end{align}
The left-pinched R\'enyi relative entropy is, as we shall see below, defined as $\DL_\alpha:=\lim_{n\to\infty} f_n(\alpha)$, with $\DR_\alpha$ defined similarly. Its existence and properties will be given below in \cref{thm:pinched}. The first thing we will note is that, while we do not know of a closed form solution for $\DL_\alpha$ and $\DR_\alpha$ in general, they are known to reduce to the sandwiched and reverse sandwiched entropies for $\alpha\geq 0$ and $\alpha\leq 1$ respectively~\cite[Prop.~4.12]{Tomamichel2016},
\begin{aligns}
    \forall \alpha\geq 0\quad \DL_\alpha\reli{}&=\frac{1}{\alpha-1}\log\Tr\left(\left(\sqrt\rho\sigma^{\frac{1-\alpha}\alpha}\sqrt\rho\right)^\alpha\right),\\
    \forall \alpha\leq 1\quad \DR_\alpha\reli{}&=\frac{1}{\alpha-1}\log\Tr\left(\left(\sqrt\sigma\rho^{\frac{\alpha}{1-\alpha}}\sqrt\sigma\right)^{1-\alpha}\right),
\end{aligns}
and so inherit the desired properties within these ranges. As such, we will focus on showing that these properties extend beyond these ranges where we lack closed-form expressions. We start by showing that $f_n$, $f_n'$, $f_n''$ are uniformly bounded.

\begin{lemma}
    \label{lem:pinch_bounded}
    For all $n$ and $\alpha\leq 0$, $f_n(\alpha)$ is non-positive and bounded by the minimal entropy
    \begin{align}
        0\geq f_n(\alpha)\geq \Dm_\alpha\reli{}.
    \end{align}
    Moreover, there exist uniform (i.e.\ independent of $n$ and $\alpha$) bounds on the value and first two derivatives of $f_n$,
    \begin{align}
        \abs{f_n(\alpha)}\leq C_0,\qquad
        \abs{f_n'(\alpha)}\leq C_1,\qquad
        \abs{f_n''(\alpha)}\leq C_2.
    \end{align}    
\end{lemma}
\begin{proof}
Firstly, the non-positivity of $f_n$ follows from the fact that $D_\alpha$ is non-positive for $\alpha\leq 0$. The lower bound on~$f_n$ follows from the data-processing inequality (recalling that the DPI is reversed for $\alpha\leq 0$), and additivity of the minimal relative entropy,
\begin{align}
f_n=\frac 1n \Dm_\alpha\rel{\pinch{\rho^{\otimes n}}{\sigma^{\otimes n}}}{\sigma^{\otimes n}}\geq \frac 1n \Dm_\alpha\rel{\rho^{\otimes n}}{\sigma^{\otimes n}}=\Dm_\alpha\reli{}.
\end{align}
Furthermore, given that $\Dm_\alpha\reli{}\geq \Dm_{-\infty}\reli{}$ for all $\alpha$, we have a uniform lower bound on $f_n$, i.e.,\ \mbox{$C_0:=\Dm_{-\infty}\reli{}$}.

Next, we turn to the derivatives. Before applying it to our states, we start by looking at what form the derivatives of the (classical) R\`enyi relative entropy take in the abstract, say for two classical distributions $p$ and $q$. For notational simplicity we are going to assume all logarithms below are natural to avoid factors of $\ln b$. Given that we are only concerned with non-positive $\alpha$ and are not concerned with $\alpha=1$, we can switch to looking at the unnormalised variant of the R\'enyi relative entropy of the form $(\alpha-1)D_\alpha\rel pq$. Taking derivatives of this gives
\begin{aligns}
    (\alpha-1)D_\alpha\rel pq &= \log \sum_i p_i^{\alpha}q_i^{1-\alpha},\\
    \left((\alpha-1)D_\alpha\rel pq\right)' &= \frac{\sum_i \ln\frac{p_i}{q_i}\cdot p_i^{\alpha}q_i^{1-\alpha}}{\sum_i p_i^{\alpha}q_i^{1-\alpha}},\\
    \left((\alpha-1)D_\alpha\rel pq\right)'' &= 
    \frac 
    { 
        \left(\sum_i \ln^2\frac{p_i}{q_i}\cdot p_i^{\alpha}q_i^{1-\alpha}\right)\cdot
        \left(\sum_i p_i^{\alpha}q_i^{1-\alpha}\right)^2
        -
        \left(\sum_i \ln\frac{p_i}{q_i}\cdot p_i^{\alpha}q_i^{1-\alpha}\right)^2 
    }
    {\left(\sum_i p_i^{\alpha}q_i^{1-\alpha}\right)^2}.
\end{aligns}
Conveniently, the first and second derivatives take the form of moments. Specifically, if we consider the distribution,
\begin{align}
    w_i:=\frac{p_i^{\alpha}q_i^{1-\alpha}}{\sum_j p_j^{\alpha}q_j^{1-\alpha}},
\end{align}
then the derivatives become the mean and variance of $\ln\frac{p}{q}$ with respect to $w$,
\begin{aligns}
    \left((\alpha-1)D_\alpha\rel pq\right)' &= \sum_i w_i \ln\frac{p_i}{q_i},\\
    \left((\alpha-1)D_\alpha\rel pq\right)'' &= \sum_i w_i\ln^2\frac{p_i}{q_i}-\left( \sum_i w_i\ln\frac{p_i}{q_i} \right)^2.
\end{aligns}
So now we can uniformly bound both in terms of $\abs{\ln p_i/q_i}\leq \max\lbrace -\ln \min_i p_i,-\ln \min_i q_i\rbrace$, specifically
\begin{aligns}
    \abs{\left((\alpha-1)D_\alpha\rel pq\right)'} &\leq \max\lbrace -\ln \min_i p_i,-\ln \min_i q_i\rbrace,\\
    \abs{\left((\alpha-1)D_\alpha\rel pq\right)''} &\leq \max\lbrace -\ln \min_i p_i,-\ln \min_i q_i\rbrace^2.
\end{aligns}

Now we want to return to $f_n$, wherein $p=\pinch{\rho^{\otimes n}}{\sigma^{\otimes n}}$ and $q=\sigma^{\otimes n}$. Before that, we need to deal with the pinching. Specifically, if we use the the pinching inequality
\begin{align}
    \pinch{\rho^{\otimes n}}{\sigma^{\otimes n}} \geq \frac{\rho^{\otimes n}}{\abs{\mathrm{spec}(\sigma^{\otimes n})}}\geq \frac{\rho^{\otimes n}}{n^d},
\end{align}
we can see that
\begin{align}
    -\frac 1n \ln \lambda_{\min}\left( \pinch{\rho^{\otimes n}}{\sigma^{\otimes n}} \right)
    &\leq -\frac 1n \ln \frac{\lambda_{\min}^n(\rho)}{n^d}
    =\frac{d\ln n}{n}-\ln\lambda_{\min}(\rho)
    \leq d-\ln\lambda_{\min}(\rho),
\end{align}
and thus
\begin{align}
    \max\left\lbrace
    -\frac 1n \ln \lambda_{\min}\left( \pinch{\rho^{\otimes n}}{\sigma^{\otimes n}} \right) ,
    -\ln \lambda_{\min}\left(\sigma\right)
    \right\rbrace\leq M,
\end{align}
where $M:=\max\left\lbrace d-\ln\lambda_{\min}(\rho), -\ln\lambda_{\min}(\sigma) \right\rbrace$. 

Now we return to $f_n$. We take the unnormalised version of this,
\begin{align}
    (\alpha-1)f_n(\alpha)=\frac{1}{n} (\alpha-1)D\rel{\pinch{\rho^{\otimes n}}{\sigma^{\otimes n}}}{\sigma^{\otimes n}},
\end{align}
so by the above arguments the derivatives can be bounded
\begin{aligns}
    \abs{\left((\alpha-1)f_n(\alpha)\right)'} &\leq  \max\left\lbrace
    -\frac 1n \ln \lambda_{\min}\left( \pinch{\rho^{\otimes n}}{\sigma^{\otimes n}} \right) ,
    -\ln \lambda_{\min}\left(\sigma\right)
    \right\rbrace\leq M,\\
    \abs{\left((\alpha-1)f_n(\alpha)\right)''} &\leq  \max\left\lbrace
    -\frac 1n \ln \lambda_{\min}\left( \pinch{\rho^{\otimes n}}{\sigma^{\otimes n}} \right) ,
    -\ln \lambda_{\min}\left(\sigma\right)
    \right\rbrace^2\leq M^2.
\end{aligns}

We thus have that the derivatives of the $(\alpha-1)f_n(\alpha)$ are bounded, so all that is left is to show is that this necessarily extends to $f_n(\alpha)$ itself. By straightforward algebraic manipulation, we can write the derivatives of the latter quantity in terms of those of the former, specifically
\begin{aligns}
f_n'(\alpha)&= \frac{\left((\alpha-1)f_n(\alpha)\right)'-f_n(\alpha)}{\alpha-1},\\
f_n''(\alpha)&= \frac{\left((\alpha-1)f_n(\alpha)\right)''-2f_n'(\alpha)}{\alpha-1}.
\end{aligns}
Given that $\alpha\leq 0$, and is therefore gapped away from $\alpha=1$, this causes no issues. Specifically, if we take $C_1:=M+C_0$ and $C_2:=M^2+2C_1$, then the desired uniform bounds hold as required.
\end{proof}

Before attacking the existence and properties of $\DL_\alpha$, we need one final theorem that allows us to leverage these uniform bounds to extend properties of $\lbrace f_n\rbrace_n$ through to $\DL_\alpha$. This theorem is a corollary of the Arzel\`a-Ascoli Theorem. 

\begin{lemma}[Arzel\`a-Ascoli Theorem~{\cite[Cor.~11.6.11]{Lebl2018}}]
\label{lem:aa}
    Let $\lbrace f_n\rbrace_n$ be a sequence of differentiable functions on a compact domain which are uniformly bounded, and whose derivative is also uniformly bounded. Then, there exists a uniformly convergent subsequence $\lbrace f_{m_n}\rbrace_n$.
\end{lemma}

With this in hand, we turn to proving the properties of the pinched relative entropies.

\begin{theorem}[Properties of the pinched relative entropy]
    \label{thm:pinched}
    Define the left-pinched relative entropy as
    \begin{align}
         \DL_\alpha\reli{}:=& \lim_{n\to\infty}\frac 1n D_\alpha\rel{\pinch{\rho^{\otimes n}}{\sigma^{\otimes n}}}{\sigma^{\otimes n}}.
    \end{align}
    For full rank states and $\alpha\in\overline{\mathbb R}$, the pinched relative entropy has the following properties:
    \begin{itemize}
        \item Existence: $\DL_\alpha\reli{}$ exists.
        \item (Non-)positivity: $\DL_\alpha\reli{}$ is non-negative for $\alpha\geq 0$, and non-positive for $\alpha\leq 0$.
        \item Subminimality: $\DL_\alpha\reli{}\leq \Dm_\alpha\reli{}$ for $\alpha\geq 0$, and $\DL_\alpha\reli{}\geq\Dm_\alpha\reli{}$ for $\alpha\leq 0$. 
        \item Differentiability: $\alpha\mapsto \DL_\alpha\reli{}$ is differentiable.
    \end{itemize}
    Moreover, all of these properties also extend to the right-pinched relative entropy
    \begin{align}
         \DR_\alpha\reli{}:=& \lim_{n\to\infty}\frac 1n D_\alpha\rel{\rho^{\otimes n}}{\pinch{\sigma^{\otimes n}}{\rho^{\otimes n}}}
         .
    \end{align}
\end{theorem}
\begin{proof}
    The sandwiched relative entropy has all of the above properties~\cite{Tomamichel2016}, and coincides with the pinched relative entropy with $\alpha\geq 0$, so we need only show that these properties hold for $\alpha\leq 0$ as well.

    If we consider the composition of pinching a composite system, we have
    \begin{align}
        \left( \mathcal P_X \otimes \mathcal P_Y \right)\left( \mathcal P_{X\otimes Y}(A)\right)=
        \left( \mathcal P_X \otimes \mathcal P_Y \right)\left(A\right).
    \end{align}
    This, together with the data-processing inequality, gives
    \begin{aligns}
    D_\alpha\rel{\pinch{\rho^{\otimes (n+m)}}{\sigma^{\otimes (n+m)}}}{\sigma^{\otimes (n+m)}}
    &\leq D_\alpha\rel{(\mathcal P_{\sigma^{\otimes n}}\otimes\mathcal P_{\sigma^{\otimes m}})\left(\rho^{\otimes (n+m)}\right)}{\sigma^{\otimes (n+m)}},\\
    &=D_\alpha\rel{\pinch{\rho^{\otimes n}}{\sigma^{\otimes n}}}{\sigma^{\otimes n}} + D_\alpha\rel{\pinch{\rho^{\otimes m}}{\sigma^{\otimes m}}}{\sigma^{\otimes m}},
    \end{aligns}
    or in other words $(n+m){\ctc f_{n+m}(\alpha)}\geq n{\ctc f_n(\alpha)}+m{\ctc f_m(\alpha)}$. {\ctc Applying Fekete's Superadditive lemma~\cite{Fekete}, this superaddivity implies that $f_n(\alpha)$ is convergent in $n$ for each $\alpha$.} 

    Next, we want to apply Arzel\`a-Ascoli Theorem (\cref{lem:aa}). \cref{lem:pinch_bounded} gives us the uniform boundedness required to apply \cref{lem:aa} to $\lbrace f_n\rbrace_n$, which gives that there exists a uniformly convergent subsequence $\lbrace f_{a_n}\rbrace_n$. But, as we already have established that $\lbrace f_n\rbrace_n$ is also convergent, this implies that this convergence is uniform. Next, using the uniform bound on $\lbrace f''_n\rbrace_n$ from \cref{lem:pinch_bounded}, we can also apply \cref{lem:aa} to $\lbrace f_n'\rbrace_n$, which gives a uniformly convergent subsequence $\lbrace f_{b_n}'\rbrace_n$. As $\lbrace f_{b_n}\rbrace_n$ and $\lbrace f_{b_n}'\rbrace_n$ are both uniformly convergent, we can commute through the limit and the derviative. Using this together with the convergence of $\lbrace f_n\rbrace_n$, we can see that $\lbrace f_n'\rbrace_n$ must also be (uniformly) convergent, and thus that $\DL_\alpha$ is differentiable,
    \begin{align}
        \lim_{n}f_{b_n}'(\alpha)=\left(\lim_{n}f_{b_n}(\alpha)\right)'=\left(\lim_{n}f_{n}(\alpha)\right)'=\left(\DL_\alpha\right)'.
    \end{align}

    To extend all of these properties to the right-pinched relative entropy we can simply use the identity
    \begin{align}
    \DR_{\alpha}\reli{}=\frac{\alpha}{1-\alpha}\DL_{1-\alpha}\rel\sigma\rho.
    \end{align}
    We might suspect that the $1-\alpha$ denominator causes issues around $\alpha=1$, but as $\DR_\alpha$ reduces to the reverse sandwiched relative entropy for $\alpha>1/2$ it therefore inherits the above properties within that range.
\end{proof}

Next we prove a nice relationship between the two pinched relative entropies and the minimal relative entropy.
\begin{lemma}
    \label{lem:maxpinch}
    The maximum (in magnitude) of the pinched R\'{e}nyi relative entropies corresponds to the minimal relative entropy,
    \begin{align}
        \max\left\lbrace 
        \abs{\DL_\alpha\reli{}},
        \abs{\DR_\alpha\reli{}} 
        \right\rbrace=\abs{\Dm_\alpha\reli{}}.
    \end{align}
\end{lemma}
\begin{proof}
    From the subminimality property of the pinched relative entropy (see \cref{thm:pinched}), we have
    \begin{align}
        \abs{\Dm_\alpha\reli{}}\geq \max\left\lbrace 
        \abs{\DL_\alpha\reli{}},
        \abs{\DR_\alpha\reli{}} 
        \right\rbrace.
    \end{align}
    Next, Ref.~\cite[Prop.~4.12]{Tomamichel2016} gives that $\DL_\alpha$ corresponds to the sandwiched entropy for $\alpha\geq 0$, which in turn corresponds to the minimal entropy for $\alpha>1/2$~\cite[Sec.~4.3]{Tomamichel2016}. By duality this means that $\DR_\alpha$ corresponds to the reverse sandwiched relative entropy for $\alpha\leq 1$, and also to the minimal for $\alpha\leq 1/2$. Thus, this inequality is satisfied for all $\alpha$ as required.
\end{proof}

\section{Two-sided error}
\label{app:two-side}

In \cref{\ratetheorems} we only considered transformations involving an error on the first state in the dichotomy. One reason this was done is because such transformations are the relevant transformations for the resource theoretic applications of concern (see \cref{subsec:thermo,sec:frame_ent}). Another is that, as we will see, the more general problem in which we allow errors on both states is no more rich. In this appendix we will give a summary of the asymptotic rate scalings for two non-zero-errors. In lieu of giving rigorous proofs of these rates, we will instead mention how things change from the proofs of \cref{\ratetheorems}.

One of the reason that two non-zero-errors do not give a much richer problem is that there exist errors for which the rate becomes infinite. To be clear, we do not mean that the rate diverges as $n\to\infty$ (\textit{a la} \cref{thm:rate_extreme}), but instead a situation where the rate is unbounded for a finite $n$. This occurs because, if the errors as sufficiently large, the Blackwell order breaks down in its entirety:
\begin{lemma}[Breakdown of Blackwell ordering]
    \label{lem:initial}
    If $\beta_{\epsilon_\rho}\reli{}\leq \epsilon_\sigma$, then $(\rho,\sigma)$ Blackwell dominates \emph{all} dichotomies, i.e.,\ 
    \begin{align}
        \beta_{\epsilon_\rho}\reli{1}\leq \epsilon_\sigma
        \qquad\iff\qquad
        (\rho,\sigma)\succeq_{(\epsilon_\rho,\epsilon_\sigma)} (\rho',\sigma')~\forall \rho',\sigma'.
    \end{align}
\end{lemma}
\begin{proof}
    Using the definition of $\beta_x$, we have that there exists a test $Q$ such that 
    \begin{align}
        \Tr ((I-Q)\rho_1)\leq \epsilon_\rho 
        \qquad\text{and}\qquad
        \Tr (Q\sigma_1)\leq \epsilon_\sigma.
    \end{align}
    Consider a measure-and-prepare channel based on that very test, specifically
    \begin{align}
        \mathcal E(\tau):=\rho_2 \Tr (Q \tau) + \sigma_2 \Tr ((I-Q) \tau). 
    \end{align}
    Applying this channel, we can easily see it has the desired error properties for any output dichotomy,
    \begin{aligns}
        T\left( \mathcal E(\rho_1),\rho_2 \right)&=T(\rho_2,\sigma_2)\cdot \Tr((I-Q)\rho_1) \leq \epsilon_\rho,\\
        T\left( \mathcal E(\sigma_1),\sigma_2 \right)&=T(\rho_2,\sigma_2)\cdot \Tr (Q\sigma_1) \leq \epsilon_\sigma.
    \end{aligns}
    As for the reverse direction, this simply follows from using the data-processing inequality for $\beta_x$ and the output states $\rho'=\proj{0}$ and $\sigma'=\proj{1}$, as $\beta_x\rel{\ketbra{0}{0}}{\ketbra{1}{1}}\equiv 0$.
\end{proof}

So now let us move on to transformation rates. Similar to the one-sided error case, let $R_n^*(\epsilon_{n}^{(\rho)},\epsilon_{n}^{(\sigma)})$ denote the largest $R_n$ such that
\begin{align}
    \left(\rho_1^{\otimes n},\sigma_1^{\otimes n}\right)
    \succeq_{(\epsilon_{n}^{(\rho)},\epsilon_{n}^{(\sigma)})}
    \left(\rho_2^{\otimes R_nn},\sigma_2^{\otimes R_nn}\right),
\end{align}
where we note that $R_n^*(\epsilon_{n}):=R_n^*(\epsilon_{n},0)$.

\definecolor{color1}{RGB}{191,255,191}
\definecolor{color2}{RGB}{223,223,255}
\definecolor{color3}{RGB}{255,255,159}
\definecolor{color4}{RGB}{255,191,191}

\newcommand{\cca}[1]{\multicolumn{1}{>{\columncolor{color1}[\tabcolsep]}c}{#1}}
\newcommand{\ccb}[1]{\multicolumn{1}{>{\columncolor{color2}[\tabcolsep]}c}{#1}}
\newcommand{\ccc}[1]{\multicolumn{1}{>{\columncolor{color3}[\tabcolsep]}c}{#1}}
\newcommand{\ccd}[1]{\multicolumn{1}{>{\columncolor{color4}[\tabcolsep]}c}{#1}}

\begin{table}[ht]
\begin{tabular}{c|ccccccc}
\textbf{}& \textbf{Zero-error} & \textbf{Large$_<$} & \textbf{Moderate$_<$} & \textbf{Small} & \textbf{Moderate$_>$} & \textbf{Large$_>$} & \textbf{Extreme} \\ \hline
\rowcolor{color1}[\tabcolsep]\cellcolor{white}\textbf{Zero-error} & \cref{thm:rate_zero} & \cref{thm:rate_largedev_lo} & \cref{thm:rate_moddev} & \cref{thm:rate_smalldev} & \cref{thm:rate_moddev} & \cref{thm:rate_largedev_hi} & \cref{thm:rate_extreme} 
\\
\textbf{Large$_<$} & \cca{\cref{thm:rate_largedev_lo}} & \ccb{\cref{lem:twoside_largelo}} & \multicolumn{5}{>{\columncolor{color3}[\tabcolsep]}c}{\cref{lem:twoside_unchanged}} 
\\
\textbf{Moderate$_<$} & \cca{\cref{thm:rate_moddev}} & \ccc{} &
\multicolumn{5}{>{\columncolor{color4}[\tabcolsep]}c}{} 
\\
\textbf{Small} & \cca{\cref{thm:rate_smalldev}} & \ccc{} & \multicolumn{5}{>{\columncolor{color4}[\tabcolsep]}c}{} 
\\
\textbf{Moderate$_>$} & \cca{\cref{thm:rate_moddev}} & \ccc{\cref{lem:twoside_unchanged}} &
\multicolumn{5}{>{\columncolor{color4}[\tabcolsep]}c}{\cref{lem:breakdown}}
\\
\textbf{Large$_>$} & \cca{\cref{thm:rate_largedev_hi}} & \ccc{} & 
\multicolumn{5}{>{\columncolor{color4}[\tabcolsep]}c}{} 
\\
\textbf{Extreme} & \cca{\cref{thm:rate_extreme}} & \ccc{} &   \multicolumn{5}{>{\columncolor{color4}[\tabcolsep]}c}{} 
\end{tabular}
\caption{\textbf{Summary of the two-sided error results.} The green results, \cref{\ratetheorems}, are the one-sided error results presented in \cref{sec:results}. The red region, \cref{lem:breakdown}, denotes where the Blackwell order breaks down resulting in eventually infinite transformation rates. The yellow region, \cref{lem:twoside_unchanged}, denotes the regimes in which the one-sided rates hold until a critical error exponent is reached, beyond which the Blackwell order once again breaks down. Finally the blue region, \cref{lem:twoside_largelo}, denotes the sole regime in which there is a non-trivial change in the transformation rate from the one-sided error case.}
\label{tbl:twoside}
\end{table}

In \cref{\ratetheorems} we dealt with the cases where one error was exactly zero, splitting the results up by the scaling of the other error into 7 different regimes (small, moderate low/high, large low/high, extreme low/high).
Na\"ively, one might think we then need to consider 49 different regimes for the general two-sided error problem (see \cref{tbl:twoside}). However, \cref{lem:initial} will allow us to instantly rule out 25 of these regimes in which neither error is exponentially small:
\begin{lemma}[Rate breakdown]
    \label{lem:breakdown}
    Suppose that neither $\epsilon_n^{(\rho)}$ nor $\epsilon_n^{(\sigma)}$ is exponentially bounded, i.e.
    \begin{align}
        \lim_{n\to\infty} \frac 1n \log\epsilon_n^{(\rho)}
        =\lim_{n\to\infty} \frac 1n \log\epsilon_n^{(\rho)}
        =0.
    \end{align}
    Then, the rate is eventually infinite, $R_n^*(\epsilon_n^{(\rho)},\epsilon_n^{(\sigma)})\ev =+\infty$.
\end{lemma}
\begin{proof}
    The idea here is to show that if neither error is exponentially shrinking, then eventually we see a breakdown of the Blackwell ordering in the sense of \cref{lem:initial}. As neither error is exponententially decaying, then we can take any arbitrarily small constant $\delta>0$ and have that
    \begin{align}
        \epsilon_n^{(\rho)}\gtev \exp(-\delta n)
        \qquad\text{and}\qquad
        \epsilon_n^{(\sigma)}\gtev \exp(-\delta n).
    \end{align}
    From \cref{lem:ht_largedev} we have
    \begin{align}
        \gamma_{-\delta}\rel{\rho_1^{\otimes n}}{\sigma_1^{\otimes n}}=\Gamma_{-\delta}\reli{1}.
    \end{align}
    Given that $\Gamma_0\reli{}=-D\reli{}<0$, and $\Gamma_\lambda$ is continuous in $\lambda$, then for sufficiently small $\delta$ then we will also have $\Gamma_{-\delta}\reli{}<-\delta$. In terms of the type-II error probability,
    \begin{align}
        \beta_{\exp(-\delta n)}\rel{\rho_1^{\otimes n}}{\sigma_1^{\otimes n}}\ltev\exp(-\delta n).
    \end{align}
    Lastly we can use the monotonicity of $\beta_x$, which gives
    \begin{align}
        \beta_{\epsilon_n^{(\rho)}}\rel{\rho_1^{\otimes n}}{\sigma_1^{\otimes n}}
        \ev\leq
        \beta_{L^{-1}[-\delta n]}\rel{\rho_1^{\otimes n}}{\sigma_1^{\otimes n}}
        \ltev \exp(-\delta n)
        \ltev
        \epsilon_n^{(\sigma)},
    \end{align}
    and so for sufficiently large $n$ \cref{lem:initial} applies, and thus the rate become infinite, $R_n^*(\epsilon_n^{(\rho)},\epsilon_n^{(\sigma)})\ev =+\infty$.
\end{proof}

As such, the only regimes left are the cases where one error is exponentially small and the other is non-zero. For simplicity, we will assume that the second error is the exponentially small error for the rest of this appendix,
\begin{align}
    \epsilon_n^{(\sigma)}:=\exp(-n\lambda_\sigma),
\end{align}
and will discuss how this non-zero $\epsilon_n^{(\sigma)}$ modifies the results of \cref{\ratetheorems} for different regimes of $\epsilon_n^{(\rho)}$. As we shall see below, in the large deviation low-error regime we get a non-trivial change in the asymptotic rate (\cref{lem:twoside_largelo}), but in all other regimes we get that the one-sided error results hold unchanged up to a critical value of $\lambda_n^{(\rho)}$, beyond which we see a breakdown similar to \cref{lem:breakdown} resulting in an eventually infinite rate (\cref{lem:twoside_unchanged}). A classification of the 49 two-sided error regimes is given in \cref{tbl:twoside}.

\subsection{High errors}

We will start with the small and moderate deviation results. In these cases, we will see that as long as the exponent of the second error, $\lambda_\sigma$, is above a certain critical exponent then these regimes are left unchanged. But if it crosses, we also get a complete breakdown. 

\begin{lemma}[Unchanged two-sided rates]
    \label{lem:twoside_unchanged}
    If $\epsilon_n^{(\rho)}$ is in the small or moderate deviation regimes, \mbox{$e^{\omega(n)}\leq \epsilon_n^{(\rho)}\leq 1-e^{O(n)}$}, and $\lambda_\sigma>D\reli{1}$, then the small and moderate deviation results of \cref{thm:rate_smalldev,thm:rate_moddev} remain unchanged, and if $\lambda_\sigma<D\reli{1}$ then \mbox{$R_n^*(\epsilon_n^{(\rho)},\epsilon_n^{(\sigma)})\ev=+\infty$}.
    
    If $\epsilon_n^{(\rho)}$ is in the high-error large deviation regime, $\epsilon_n^{(\rho)}:=1-\exp(-n\lambda_n^{(\rho)})$, and $\lambda_\sigma>-\Gamma_{\lambda_{\rho}}\reli{1}$, then the high-error large deviation results of \cref{thm:rate_largedev_hi} remain unchanged, and if $\lambda_\sigma<-\Gamma_{\lambda_{\rho}}\reli{1}$ then $R_n^*(\epsilon_n^{(\rho)},\epsilon_n^{(\sigma)})\ev=+\infty$.
    
    If $\epsilon_n^{(\rho)}$ is in the extreme deviation regime, $\epsilon_n^{(\rho)}=1-\exp(\omega(n))$, then \cref{thm:rate_extreme} remains unchanged for any $\lambda_n^{(\sigma)}$.
\end{lemma}
\begin{proof}[Proof sketch]
    We start with the small and moderate cases. When $\lambda_\sigma<D\reli{1}$ we can once again use \cref{lem:initial}. Specifically, the first-order contributions of \cref{lem:ht_smalldev,lem:ht_moddev} give that
    \begin{align}\label{eq:twoside_smallmod}
        \lim_{n\to\infty}-\frac 1n\log \beta_{\epsilon_n^{(\rho)}}\rel{\rho_1^{\otimes n}}{\sigma_1^{\otimes n}}=D\reli{1},
    \end{align}
    for any $\epsilon_n^{(\rho)}$ that is not exponentially approaching either 0 or 1. So, if $\lambda_\sigma<D\reli{1}$, then $\epsilon_n^{(\sigma)}$ is decaying with a smaller exponent and must dominate this expression, specifically
    \begin{align}
        \beta_{\epsilon_n^{(\rho)}}\rel{\rho_1^{\otimes n}}{\sigma_1^{\otimes n}}\ev < \epsilon_n^{(\sigma)}.
    \end{align}
    Thus, by \cref{lem:initial}, the Blackwell order breaks down and \mbox{$R_n^*(\epsilon_n^{(\rho)},\epsilon_n^{(\sigma)})\ev=+\infty$}.
    
    Next, we want to argue that for $\lambda_\sigma>D\reli{1}$ the results of  \cref{thm:rate_smalldev,thm:rate_moddev} remain unchanged. Clearly, allowing errors on the second state can only \emph{increase} the optimal transformation rate, and so to demonstrate this rate remains unchanged we need only show that the upper bound (optimality) remains unchanged. The optimality bound of \cref{thm:rate_smalldev,thm:rate_moddev} comes from applying \cref{lem:ht}, which bounds the rate $R_n$ by 
    \begin{align}
        \forall x\in\left(\epsilon_n^{(\rho)},1\right):~~\beta_{x}\rel{\rho_1^{\otimes n}}{\sigma_1^{\otimes n}}
        \leq \beta_{x-\epsilon_n^{(\rho)}}\rel{\rho_2^{\otimes R_n n}}{\sigma_2^{\otimes R_n n}}.
    \end{align}
    In the presence of two-sided errors this changes to
    \begin{align}
    \label{eq:twoside_smallmod2}
        \forall x\in\left(\epsilon_n^{(\rho)},1\right):~~\beta_{x}\rel{\rho_1^{\otimes n}}{\sigma_1^{\otimes n}}-\epsilon_n^{(\sigma)}
        \leq \beta_{x-\epsilon_n^{(\rho)}}\rel{\rho_2^{\otimes R_n n}}{\sigma_2^{\otimes R_n n}}.
    \end{align}
    Now, consider the two terms on the left hand side. By \cref{eq:twoside_smallmod}, we know that the first $\beta_x$ term is exponentially decaying with $n$, with an exponent of $D\reli{1}$, and $\epsilon_n^{(\sigma)}$ is decaying with an exponent of $\lambda_\sigma$. As $\lambda_\sigma>D\reli{1}$, we have that this error term is asymptotically dominated, specifically
    \begin{align}
        \lim_{n\to\infty}
        \frac{\beta_{x}\rel{\rho_1^{\otimes n}}{\sigma_1^{\otimes n}}-\epsilon_n^{(\sigma)}}{\beta_{x}\rel{\rho_1^{\otimes n}}{\sigma_1^{\otimes n}}}=1.
    \end{align}
    As such, the $\epsilon_n^{(\sigma)}$ term in \cref{eq:twoside_smallmod2} is asymptotically irrelevant, reducing this optimality bound to that given in \cref{thm:rate_smalldev,thm:rate_moddev}.
    
    Now we turn to the high-error large deviation regime. From \cref{lem:ht_largedev} we have
    \begin{align}
        \lim_{n\to\infty} -\frac 1n \log \beta_{\epsilon_n^{(\rho)}}\rel{\rho_1^{\otimes n}}{\sigma_1^{\otimes n}} = \Gamma_{\lambda_\rho}\reli{1}.
    \end{align}
    So, if $\lambda_\rho < -\Gamma_{\lambda_\rho}\reli{1}$, then
    \begin{align}
        \beta_{\epsilon_n^{(\rho)}}\rel{\rho_1^{\otimes n}}{\sigma_1^{\otimes n}}\ev < \epsilon_n^{(\sigma)},
    \end{align}
    and so by \cref{lem:initial} we can conclude $R_n^*(\epsilon_n^{(\rho)},\epsilon_n^{(\sigma)})\ev=+\infty$. If however $\lambda_\rho> -\Gamma_{\lambda_\rho}\reli{1}$, then the error term will be exponentially dominated by all of the relevant hypothesis testing quantities in the optimality proof, and therefore \cref{thm:rate_largedev_hi} will remain unchanged.
    
    Lastly \cref{thm:rate_extreme} trivially remains unchanged, as the rate in that regime is unbounded, and introducing error on the second state can only increase the rate further.
\end{proof}

\subsection{Low errors}

Finally, we are left with large deviation, low-error. This is the one regime where a non-trivial change in the asymptotic rate occurs. For $\epsilon_n^{(\sigma)}=0$ we got that the rate was given by $\overline r$/$\widecheck r$ optimised over a range of type-I log odds determined by the error on the first state. Similarly, we will see that the optimal rate is once again an optimisation of $\overline r$/$\widecheck r$, this time optimised over a range of type-I log odds determined by the first state error \emph{and} type-II log odds determined by the second state error. Before we can give the modified result, we first need to define $\widecheck\Gamma_\lambda\reli{}:=\min\left\lbrace \GammaL_\lambda\reli{},\GammaR_\lambda\reli{} \right\rbrace$, which we can evaluate using \cref{lem:ht_largedev,lem:maxpinch} to be given by
\begin{align}
    \widecheck\Gamma_\lambda(\rho\|\sigma)=&
    \begin{dcases}
    \sup_{t<0}\Dm_t\reli{}+\frac{t}{1-t}\lambda &
    \lambda<-D\rel{\sigma}{\rho},\\
            \inf_{0<t<1}-\Dm_t\reli{}-\frac{t}{1-t}\lambda & 
            -D(\sigma\|\rho)< \lambda< 0,\\
            \sup_{t>1} -\Dm_t\reli{}+\frac{t}{1-t}\lambda & 
            \lambda> 0,
    \end{dcases}.
\end{align}
Using this we can now give the full two-sided low-error large deviation result.

\begin{lemma}[Two-sided large deviation, low-error]
    \label{lem:twoside_largelo}
    For any error of the form $\epsilon_n^{(\rho)}=\exp(-\lambda_\rho n)$ with constant $\lambda_\rho>0$, if $[\rho_2,\sigma_2]=0$, then the optimal rate is lower bounded
    \begin{align}
        \liminf_{n\to\infty} R^*_n(\epsilon_n^{(\rho)},\epsilon_n^{(\sigma)}) \geq 
        \inf_{\substack{-\lambda_\rho<\mu<\lambda_\rho \\ -\lambda_\sigma<\widecheck \Gamma_{\mu}\reli{1}<\lambda_\sigma }} \widecheck r(\mu)
        .
    \end{align}
    Furthermore, if we consider general output dichotomies, $[\rho_2,\sigma_2]\neq 0$, then the optimal rate is upper bounded by
    \begin{align}
        \limsup_{n\to\infty} R^*_n(\epsilon_n^{(\rho)},\epsilon_n^{(\sigma)}) \leq \inf_{\substack{-\lambda_\rho<\mu<\lambda_\rho \\ -\lambda_\sigma<\Gamma_{\mu}\reli{1}<\lambda_\sigma }} \overline r(\mu).
    \end{align}
    In the above $\overline r$ and $\widecheck r$ are defined in \cref{subsubsec:rate_large}. Moreover, these expressions hold even if these domains are empty, i.e., if $\Gamma_{-\lambda_\rho}\reli{1}<-\lambda_\sigma$, then $R^*_n(\epsilon_n^{(\rho)},\epsilon_n^{(\sigma)})\ev=+\infty$.
\end{lemma}
\begin{proof}[Proof sketch]
    Here we are just going to provide a sketch of the proof, see the proof of \cref{thm:rate_largedev_lo} for a more rigorous treatment of this argument. We start with optimality. \cref{lem:ht} gives that, for any achievable rate $R$,
    \begin{align}
        \forall x\in(\epsilon_n^{(\rho)},1):~~\beta_{x}\rel{\rho_1^{\otimes n}}{\sigma_1^{\otimes n}}-\epsilon_n^{(\sigma)}\leq \beta_{x-\epsilon_n^{(\rho)}}\rel{\rho_2^{\otimes Rn}}{\sigma_2^{\otimes Rn}}.
    \end{align}
    Firstly, we reparamaterise $x\to x+\epsilon_n^{(\rho)}$, which gives
    \begin{align}
        \forall x\in(\epsilon_n^{(\rho)}/2,1-\epsilon_n^{(\rho)}/2):~~\beta_{x+\epsilon_n^{(\rho)}/2}\rel{\rho_1^{\otimes n}}{\sigma_1^{\otimes n}}-\epsilon_n^{(\sigma)}\leq \beta_{x-\epsilon_n^{(\rho)}/2}\rel{\rho_2^{\otimes Rn}}{\sigma_2^{\otimes Rn}}.
    \end{align}
    Next, we want to reparameterise again by the log odds per copy instead of a probability. Specifically, we will switch from $x$ to $\mu$, where $x=L^{-1}[\mu n]$. Doing so gives 
    \begin{align}
        \forall \mu\in(-\lambda_\rho,+\lambda_\rho):~~\beta_{L^{-1}[\mu n]+\epsilon_n^{(\rho)}/2}\rel{\rho_1^{\otimes n}}{\sigma_1^{\otimes n}}-\epsilon_n^{(\sigma)}\leq \beta_{L^{-1}[\mu n]-\epsilon_n^{(\rho)}/2}\rel{\rho_2^{\otimes Rn}}{\sigma_2^{\otimes Rn}}.
    \end{align}
    As $\abs \mu< \lambda_\rho$, we have that the $L^{-1}[\mu n]$ terms must dominate over the $\epsilon_n^{(\rho)}$ terms, specifically
    \begin{align}
        \lim_{n\to\infty} \frac 1n L\left[L^{-1}[\mu n]\pm\epsilon_n^{(\rho)}/2\right]=\mu,
    \end{align}
    and so this essentially reduces to 
    \begin{align}
        \forall \mu\in(-\lambda_\rho,+\lambda_\rho):~~\beta_{L^{-1}[\mu n]}\rel{\rho_1^{\otimes n}}{\sigma_1^{\otimes n}}-\epsilon_n^{(\sigma)}\leq \beta_{L^{-1}[\mu n]}\rel{\rho_2^{\otimes Rn}}{\sigma_2^{\otimes Rn}}.
    \end{align}
    Put in terms of log odds per copy, this is
    \begin{align}
        \label{eq:twoside_largelo1}
        \forall \mu\in(-\lambda_\rho,+\lambda_\rho):~~
        L^{-1}\left[\gamma_{\mu n}\rel{\rho_1^{\otimes n}}{\sigma_1^{\otimes n}}\right]-\epsilon_n^{(\sigma)}\leq L^{-1}\left[\gamma_{\mu n}\rel{\rho_2^{\otimes Rn}}{\sigma_2^{\otimes Rn}}\right].
    \end{align}
    Now, we can use \cref{lem:ht_largedev}, which gives
    \begin{aligns}
        \lim_{n\to\infty} \frac 1n
        \gamma_{\mu n}\rel{\rho_1^{\otimes n}}{\sigma_1^{\otimes n}}&=\Gamma_{\mu}\reli{1},\\
        \lim_{n\to\infty} \frac 1n
        \gamma_{\mu n}\rel{\rho_2^{\otimes Rn}}{\sigma_2^{\otimes Rn}}&=R\Gamma_{\mu/R}\reli{2}.
    \end{aligns}
    So in \cref{eq:twoside_largelo1} we have that log odds terms on both sides are exponentially scaling. In the absence of $\epsilon_n^{(\sigma)}$, we can directly compare this, giving the optimality presented in \cref{thm:rate_largedev_lo}.

    We can break the analysis of \cref{eq:twoside_largelo1} into three cases based on how $\Gamma_\mu\reli{1}$ compares to $\pm \lambda_\sigma$. If $\Gamma_\mu\reli{1}<-\lambda_\sigma$ then the LHS of \cref{eq:twoside_largelo1} is eventually negative, and thus trivially satisfied. If $-\lambda_\sigma\leq \Gamma_\mu\reli{1}\leq +\lambda_\sigma$ then \cref{eq:twoside_largelo1} reduces to
    \begin{align}
        \Gamma_\mu\reli{}\leq R\Gamma_{\mu/R}\reli{2},
    \end{align}
    as we saw in the absence of $\epsilon^{\sigma}_n$, which in turn gives the bound $R\leq \overline r(\mu)$. Lastly we have $\Gamma_\mu\reli{1}>+\lambda_\sigma$, in which case the LHS of \cref{eq:twoside_largelo1} scales as $1-\epsilon_n^{\sigma}$, so this reduces to
    \begin{align}
        \lambda_\sigma \leq \Gamma_{\mu/R}\reli{2},
    \end{align}
    which is strictly weaker than the constrain $R\leq \overline r(\mu)$ for the $\mu$ for which $\Gamma_\mu\reli{1}=+\lambda_\sigma$. The upshot is that we're left with an expression similar to \cref{thm:rate_largedev_lo}, with the rate being an optimisation of $\overline r$, this time with a constraint both on $\mu$ (coming from $\epsilon_n^{(\rho)}$) \emph{and} on $\Gamma_\mu\reli{1}$ (coming from $\epsilon_n^{(\sigma)}$). Specifically,
    \begin{align}
        R \leq \inf_{\substack{-\lambda_\rho<\mu<\lambda_\rho \\ -\lambda_\sigma<\Gamma_{\mu}\reli{1}<\lambda_\sigma }} \overline r(\mu).
    \end{align}
    The same sort of argumentation works for the achievability, where we find that any rate $r$ such that
    \begin{align}
        r \leq \inf_{\substack{-\lambda_\rho<\mu<\lambda_\rho \\ -\lambda_\sigma<\widecheck\Gamma_{\mu}\reli{1}<\lambda_\sigma }} \widecheck r(\mu)
    \end{align}
    is achievable.
    
    Finally, we note that the above arguments also hold if the domains of the infima are empty. Specifically, if
    \begin{align}
        \Gamma_{-\lambda_\rho}\reli{1}<-\lambda_\sigma,
    \end{align}
    then \cref{lem:initial} gives us that the rate breaks down, and so for sufficiently large $n$ the optimal rate becomes infinite.
\end{proof}


\section[Thermal operations]{\texorpdfstring{Proof of \cref{thm:thermo}}{Thermal operations}}
\label{app:thermal}

In order to prove \cref{thm:thermo}, we start with the following lemma.

\begin{lemma}
    \label{lem:thermal_op}
    Consider the initial and target states, $\rho_1$ and $\rho_2$, together with the corresponding thermal states, $\gamma_1$ and $\gamma_2$, such that $[\rho_2,\gamma_2]=0$. Then, the condition
    \begin{align}               
        \label{eq:thermal_cond}
        \forall x\in(\epsilon,1):\quad \betaL_x (\rho_1\|\gamma_1)\leq \beta_{x-\epsilon}(\rho_2\|\gamma_2)
    \end{align}
    implies that there exists a thermal operation mapping $\rho_1$ into a state $\epsilon$-close to $\rho_2$ in trace distance,
    \begin{equation}
        \rho_1 \xrightarrow[\mathrm{TO}]{\epsilon} \rho_2.
    \end{equation}
    Note that in general however the right-pinched variant
    \begin{align}               
        \forall x\in(\epsilon,1):\quad \betaR_x (\rho_1\|\gamma_1)\leq \beta_{x-\epsilon}(\rho_2\|\gamma_2)
    \end{align}
    does not necessarily similarly yield a TO-achievable Blackwell Order.
\end{lemma}
\begin{proof}
    First, note that a pinching map with respect to the eigenspaces of the thermal state $\gamma_1$ is a thermal operation, and so $\pinch{\rho_1}{\gamma_1}$ can be obtained from $\rho_1$. Since $[\pinch{\rho_1}{\gamma_1},\gamma_1]=0$ and $[\rho_2,\gamma_2]=0$ by assumption, we are dealing with initial and target states commuting with the respective thermal states. For such states, however, it is known from Ref.~\cite{renes2016relative} that the condition
     \begin{align}               
        \forall x\in(\epsilon,1):\quad \beta_x (\pinch{\rho_1}{\gamma_1}\|\gamma_1)\leq \beta_{x-\epsilon}(\rho_2\|\gamma_2)
\end{align}
is equivalent to the existence of a thermal operation $\E$ mapping $\pinch{\rho_1}{\gamma_1}$ into a state $\epsilon$-close to $\rho_2$. However, given the definition of $\betaL_x$ from \cref{eq:pinched_beta}, the above is equivalent to \cref{eq:thermal_cond}. Thus, assuming \cref{eq:thermal_cond} holds, such $\E$ exists and a composition of thermal operations $\E\circ\mathcal{P}_{\gamma_1}$, which is itself a thermal operation, maps $\rho_1$ into a state $\epsilon$-close to $\rho_2$. While the right-pinched condition similarly yields a Blackwell order, the right-pinching operation $\mathcal P_{\rho_1}(\cdot)$ is not a thermal operation (unless $\rho_1$ and $\gamma_1$ commute).
\end{proof}

Now we need to recall the general strategy used to prove \cref{\ratetheoremsthermo} in \cref{subsec:rates}. In these cases the achievability exclusively used the left-pinched sufficient condition of \cref{lem:ht}, showing that this condition gives a rate with the same asymptotic expansion as the optimality bound given by the necessary condition of \cref{lem:ht}. Thus, using \cref{lem:thermal_op}, we conclude that the optimal rates from \cref{\ratetheoremsthermo} can be achieved by thermal operations.

In the achievability proofs of \cref{thm:rate_largedev_lo,thm:rate_zero} we needed to leverage both left- and right-pinching, and thus these results are not necessarily TO-achievable. In both proofs, however, we started by proving separate achievability results using left- and right-pinching separately, and constructed the final bound by combining the two. If we eschew the right-pinch-based bound and stick to left-pinch-based bound then these proofs do yield weaker, but TO-achievable, rates. For the low-error large deviation case of \cref{thm:rate_largedev_lo} the TO-achievable rate is
\begin{align}
\liminf_{n\to\infty}R_n^*\bigl(\exp(-\lambda n)\bigr)\geq \min_{-\lambda \leq \mu\leq \lambda}\lefthat r(\mu),
\end{align}
where $\lefthat r(\mu)$ is defined in \cref{subsubsec:rate_large}. Similarly for the zero-error case of \cref{thm:rate_zero} the TO-achievable rate is
\begin{align}
\liminf_{n\to\infty}R_n^*(0)\geq \inf_{\alpha\in\mathbb R}\frac{\DL_\alpha\reli 1}{D_\alpha\reli{2}}.
\end{align}

Furthermore, when dealing with energy-incoherent states, there is no need for pinching, and so one can stick only to thermal operations. Moreover, for commuting input \emph{and} output states, the lower and upper bound for the optimal rate in \cref{thm:rate_largedev_lo,thm:rate_zero} coincide (as they only differ by choice of R\'enyi divergence), and thus these theorems yield optimal transformation rates in their respective error regimes.


\section[Work-assisted rate]{\texorpdfstring{Proof of \cref{thm:small_work}}{Work-assisted rate}}

\label{app:battery}

In this appendix we present how one can modify the reasoning used to prove \cref{thm:rate_smalldev} to prove \cref{thm:small_work}. Our aim is thus to find $R_n^*$, which is the largest rate $R_n$ for which the following transformation can be performed by thermal operations:
\begin{align}
    \label{eq:battery_assisted}
    \rho^{\ot n}_1 \ot \proj{0}_W \xrightarrow[\mathrm{TO}]{\epsilon} \rho_2^{\ot nR_n} \ot \proj{1}_W.
\end{align}
Let us recall that here $W$ denotes the ancillary battery system with energy levels $\ket{0}_W$ and $\ket{1}_W$ separated by energy gap $w$, so that the thermal state of the battery is given by
\begin{equation}
    \gamma_W = \lambda \ketbra{0}{0}_W + (1-\lambda) \ketbra{1}{1}_W \quad \mathrm{with}\quad \lambda=\frac{1}{1+e^{-\beta w}}.
\end{equation}

From \cref{lem:ht} we know that the necessary condition for that is given by
\begin{align}
\label{eq:battery_nec}
   \forall x\in(\epsilon,1):\quad {\beta}_x(\rho_1^{\ot n} \ot \proj{0}_W\|\gamma_1^{\ot n} \ot \gamma_W) \leq  \beta_{x-\epsilon}(\rho_2^{\ot n R_n} \ot \proj{1}_W\|\gamma_2^{\ot n R_n} \ot \gamma_W),
\end{align}
whereas from \cref{lem:ht,lem:thermal_op} we know that the sufficient condition is given by
\begin{align}
\label{eq:battery_suff}
   \forall x\in(\epsilon,1):\quad\betaL_x(\rho_1^{\ot n} \ot \proj{0}_W\|\gamma_1^{\ot n} \ot \gamma_W) \leq  \beta_{x-\epsilon}(\rho_2^{\ot n R_n} \ot \proj{1}_W\|\gamma_2^{\ot n R_n} \ot \gamma_W).
\end{align}
We will simplify these conditions by using the fact that
\begin{align}
\betaL_x(\rho_1^{\ot n} \ot \proj{0}_W\|\gamma_1^{\ot n} \ot \gamma_W) = {\beta}_x(\mathcal{P}_{\gamma_1^{\ot n}}(\rho_1^{\ot n}) \ot \proj{0}_W\|\gamma_1^{\ot n} \ot \gamma_W)    
\end{align}
and employing the following lemma.
\begin{lemma}\label{lem:battery}
    Consider three quantum states $\rho$, $\sigma$ and $\gamma$, where $\gamma = \sum_i \gamma_i \ketbra{i}{i}$. Then for all $x \in [0, 1]$ we have
    \begin{align}
        \label{eq:lemma_battery}
        \beta_x\rel{\rho\otimes \ketbra{i}{i} \,}{\sigma\otimes \gamma} &= \gamma_i  \beta_x\reli{}.
    \end{align}
\end{lemma}
 \begin{proof}
     Expanding out the left hand side of \cref{eq:lemma_battery} using the definition from \cref{eq:beta1,eq:beta2,eq:beta3} and decomposing the test as $Q := \sum_{i,j} Q_{ij} \ot \ketbra{i}{j}$ yields
     \begin{aligns}
          \beta_x\rel{\rho\otimes \ketbra{i}{i} \,}{\sigma\otimes \gamma} &= \min_{Q} \left\{ \Tr[(\sigma \ot \gamma) Q] \,\middle|\, \Tr[(\rho \ot \ketbra{i}{i})Q] \geq 1-x \quad \text{and} \quad 0 \leq Q \leq 1 \right\}\\
          &= \min_{Q_{kk}} \left\{ \sum_{k} \gamma_k \Tr[\sigma Q_{kk}] \,\middle|\, \Tr[\rho Q_{ii}]\geq 1-x \quad \text{and} \quad 0 \leq Q_{kk} \leq 1 \quad \text{for all} \,\, k \right\}\\
          &=  \gamma_i \min_{Q_{ii}} \left\{ \Tr[\sigma Q_{ii}] \,\middle|\, \Tr[\rho Q_{ii}]\geq 1-x \quad \text{and} \quad 0 \leq Q_{ii} \leq 1 \right\} \\
          &= \gamma_i \beta_x\rel{\rho}{\sigma},
     \end{aligns}
     which proves the claim.
 \end{proof}

We can then rewrite \cref{eq:battery_nec,eq:battery_suff} as
\begin{aligns}
     &\forall x\in(\epsilon,1):\quad \lambda {\beta}_x(\rho_1^{\ot n} \|\gamma_1^{\ot n}) \leq (1-\lambda) \beta_{x-\epsilon}(\rho_2^{\ot nR_n^*}\|\gamma_2^{\ot nR_n^*}),\\
     &\forall x\in(\epsilon,1):\quad \lambda \betaL_x(\rho_1^{\ot n} \|\gamma_1^{\ot n}) \leq (1-\lambda) \beta_{x-\epsilon}(\rho_2^{\ot nR_n^*}\|\gamma_2^{\ot nR_n^*}).
\end{aligns}
 Taking the minus $\log$ of both sides and dividing by $n$, we thus get that the necessary condition and the sufficient condition for the transformation in \cref{eq:battery_assisted} are given by
\begin{aligns}
     &\forall x\in(\epsilon,1):\quad -\frac{1}{n} \log\left( {\beta}_x(\rho_1^{\ot n} \|\gamma_1^{\ot n})\right)-\frac{\beta w}{n} \geq R_n\left( -\frac{1}{nR_n} \log\left( \beta_{x-\epsilon}(\rho_2^{\ot nR_n}\|\gamma_2^{\ot nR_n})\right)\right),\\
    &\forall x\in(\epsilon,1):\quad -\frac{1}{n} \log\left( \betaL_x(\rho_1^{\ot n} \|\gamma_1^{\ot n})\right)-\frac{\beta w}{n} \geq R_n\left( -\frac{1}{nR_n} \log\left( \beta_{x-\epsilon}(\rho_2^{\ot nR_n}\|\gamma_2^{\ot nR_n})\right)\right).
\end{aligns}

Crucially now, as \cref{lem:ht_smalldev} tells us that the second-order asymptotic expansions of $-\frac{1}{n}\log\beta_x$ and $-\frac{1}{n}\log\betaL_x$ are the same, in the small deviation regime the above necessary and sufficient conditions coincide and are given by
\begin{equation}
\label{eq:battery_nec_suff}
    \forall x\in(\epsilon,1):\quad D(\rho_1\|\gamma_1)+\sqrt{\frac{V(\rho_1\|\gamma_1)}{n}}\Phi^{-1}(x)-\frac{\beta w}{n} \gtrsim R_n D(\rho_2\|\gamma_2)+\sqrt{\frac{R_nV(\rho_2\|\gamma_2)}{n}}\Phi^{-1}(x-\epsilon),
\end{equation}
where $\gtrsim$ denotes inequality up to terms $o(1/\sqrt{n})$. Introducing
\begin{equation}
    \xi':=\frac{V(\rho_1\|\gamma_1)}{R_nV(\rho_2\|\gamma_2)},
\end{equation}
using the definition and properties of the sesquinormal distribution, one can rearrange \cref{eq:battery_nec_suff} to arrive at the following equivalent condition: 
\begin{equation}
    \frac{\beta w}{n} \lesssim D(\rho_1\|\gamma_1)-R_n D(\rho_2\|\gamma_2)+\sqrt{\frac{V(\rho_1\|\gamma_1)}{n}} S^{-1}_{1/\xi'}(\epsilon).
\end{equation}

Clearly, if $\rho_2=\gamma_2$, then the above is satisfied for any rate $R_n$ as long as
\begin{equation}
        \frac{\beta w}{n} \lesssim D(\rho_1\|\gamma_1)+\sqrt{\frac{V(\rho_1\|\gamma_1)}{n}}\Phi^{-1}(\epsilon),
\end{equation}
which proves the second part of \cref{thm:small_work}. If $\rho_2\neq \gamma_2$, then we can expand $w$ and rearrange \cref{eq:battery_nec_suff} to obtain
\begin{equation}
    R_n \lesssim \frac{D(\rho_1\|\gamma_1)-\beta w_1}{D(\rho_2\|\gamma_2)} +\frac{\sqrt{V(\rho_1\|\gamma_1)}S_{1/\xi'}^{-1}(\epsilon)-\beta w_2}{\sqrt{n}D(\rho_2\|\gamma_2)}.
\end{equation}
Now, we note in the expression for $\xi'$ we only need to account for the constant term of $R_n$, as any higher order terms will result in corrections of the order $o(1/\sqrt{n})$. Thus, a transformation from \cref{eq:battery_assisted} exists for every rate $R_n$ satisfying the above inequality with
\begin{align}
    \xi'=\frac{V(\rho_1\|\gamma_1)}{D(\rho_1\|\gamma_1)-\beta w_1} \bigg/ \frac{V(\rho_2\|\gamma_2)}{D(\rho_2\|\gamma_2)},
\end{align}
which proves the first part of \cref{thm:small_work}.


\section[Entanglement transformation rates]{\texorpdfstring{Proof sketch of \cref{thm:entanglement}}{Entanglement transformation rates}}
\label{app:entanglement}

The proof of \cref{thm:entanglement} largely follows the proofs covered in \cref{\ratetheorems}, and so instead of reproducing all the gory details we will instead point out some key differences, and then give the resulting rate expressions. Consider a transformation transformation \mbox{$\ket{\psi_1}^{\ot n}\xrightarrow[\mathrm{LOCC}]{\epsilon} \ket{\psi_2}^{\ot Rn}$} for bipartite states $\ket{\psi_1}$ and $\ket{\psi_2}$ with local dimensions $d_1$ and $d_2$, and with Schmidt spectra $\v{p}_1$ and $\v{p}_2$.  Recalling \cref{eq:entanglement_condition}, such a transformation is possible if and only if
\begin{align}
    d^{Rn}_2{\beta_x\rel{\v{p_2}^{\otimes Rn}}{\v{f_2}^{\otimes Rn}}}  \leq   
      d^n_1{\beta_{x-\epsilon}\rel{\v{p_1}^{\otimes n}}{\v{f_1}^{\otimes n}}} \quad \forall x\in(\epsilon,1),
\end{align}
where $\v{f}_i$ denotes a uniform distributions of dimension $d_i$ and $\beta_x(\v{p}\|\v{q})$ should be understood as $\beta_x(\rho\|\sigma)$ with $\rho$ and $\sigma$ being diagonal states with the diagonals given by $\v{p}$ and $\v{q}$, respectively. Applying the techniques of \cref{subsec:rates} to convert hypothesis testing asymptotics into transformation rate asymptotics, we can extract from this second-order expressions for transformation rates in the entanglement setting.

Importantly, this condition has three major differences that will influence the resulting rate expressions. Firstly, the order of the expression is backwards to that seen in the thermodynamic setting, so the resulting rates will be reciprocated. Secondly, all hypothesis testing is relative to uniform states, meaning that all of our rates will involve information theoretic quantities relative to the uniform states. All of these can be expressed in terms of their non-relative analogues, e.g.,
\begin{align}
    D_\alpha\rel{\v{p}_i}{\v{f}_i}=\frac{\alpha}{\alpha-1}\log d_i-H_\alpha(\v{p}_i).
\end{align}
And then, thirdly, we have the lingering dimensional factors, which happen to all cancel out in such a way to yield rate expressions broadly similar to those seen in \cref{\ratetheorems}.

Taking these modifications into account, if one was to follow our techniques from \cref{subsec:rates} \textit{mutatis mutandis}, the entanglement transformation rates, for $\lambda>0$ and $a,\epsilon\in(0,1)$, scale as:
\begin{aligns}
    &\text{Zero-error}: & R_n^*(0)&=\min_{0\leq \alpha\leq \infty }\frac{H_\alpha(p)}{H_\alpha(q)}+o(1),\\
    &\text{Large deviation (lo)}: & R_n^*(\exp(-\lambda n))&=\min_{-\lambda\leq \mu\leq \lambda} r(\mu),\\
    &\text{Moderate deviation (lo)}: & R_n^*(\exp(-\lambda n^a))&=\frac{H(p)-\abs{1-\xi^{-1/2}}\sqrt{2V(p)n^{a-1}}\cdot S^{-1}_{1/\xi}(\epsilon)}{H(q)}+o\left(\sqrt{n^{a-1}}\right),\\
    &\text{Small deviation}: & 
    R_n^*(\epsilon)&=\frac{H(p)+\sqrt{V(p)/n}\cdot S^{-1}_{1/\xi}(\epsilon)}{H(q)}+o(1/\sqrt n),\\
    &\text{Moderate deviation (hi)}: & R_n^*(1-\exp(-\lambda n^a))&=\frac{H(p)+\left[1+\xi^{-1/2}\right]\sqrt{2V(p)n^{a-1}}\cdot S^{-1}_{1/\xi}(\epsilon)}{H(q)}+o\left(\sqrt{n^{a-1}}\right),\\
    &\text{Large deviation (hi)}: & R_n^*(1-\exp(-\lambda n))&=\inf_{\substack{t_1>1\\0<t_2<1}}\frac{H_{t_1}(p)-\left(\frac{t_1}{1-t_1}+\frac{t_2}{1-t_2}\right)\lambda}{H_{t_2}(q)}+o(1),
\end{aligns}
where
\begin{align}
    \xi=\frac{V(p)}{H(p)} \bigg/ \frac{V(q)}{H(q)}
    \qquad\text{and}\qquad
    r(\mu)
    =\begin{dcases}
        1 & \mu\leq -D\rel{f}{p},\\
        \sup_{0<t_2<1}\inf_{0<t_1<1}\frac{H_{t_1}(p)+\left(\frac{t_2}{1-t_2}-\frac{t_1}{1-t_1}\right)\mu}{H_{t_2}(q)} & -D\rel{f}{p}\leq \mu\leq 0,\\
        \inf_{t_2>1}\sup_{t_1>1}\frac{H_{t_1}(p)+\left(\frac{t_1}{1-t_1}-\frac{t_2}{1-t_2}\right)\mu}{H_{t_2}(q)} & \mu\geq 0.
    \end{dcases}
\end{align}

We note that the small deviation and moderate deviation rates are consistent in form with the existing infidelity results of Ref.~\cite{kumagai2016second} and Ref.~\cite{chubb2017moderate}, respectively, and the zero-error rate is a restatement of Ref.~\cite{Jensen_2019}.

\section{Numerical examples of strong and weak resonance}
\label{app:res}

In this appendix we will give a numerical example of a dichotomy transformation that exhibits both weak and strong resonance in the sense discussed in \cref{subsubsec:strong_res}. Following Ref.~\cite{korzekwa2019avoiding}, we can construct examples with resonance by considering two different input states and varying the relative ratio of their numbers. That is, instead of just considering the rates $R_n$ and errors $\epsilon_n$ such that 
\begin{align}
    (\rho_1^{\otimes n},\sigma^{\otimes n})\succeq_{(\epsilon_n,0)}(\rho_2^{\otimes R_nn},\sigma^{\otimes R_nn}),
\end{align}
we can instead consider
\begin{align}
    (\rho_1^{\otimes \lambda n}\otimes \rho_1'^{\otimes (1-\lambda) n},\sigma^{\otimes n})\succeq_{(\epsilon_n,0)}(\rho_2^{\otimes R_nn},\sigma^{\otimes R_nn}),
\end{align}
for some $\lambda\in(0,1)$. Consider the states
\begin{aligns}
\rho_1 &= \mathrm{Diag}(0.4309,0.4300,0.1391),\\
\rho_1' &= \mathrm{Diag}(0.5499,0.2300,0.2201),\\
\rho_2 &= \mathrm{Diag}(0.5121,0.3300,0.1579),\\
\sigma &= \mathrm{Diag}(0.3333,0.3333,0.3333).
\end{aligns}
These states exhibit \emph{weak} resonance, as shown in the left panel of \cref{fig:res_ex}. Alternatively, if we consider the reverse process of attempting to make a mixture of two possible output states,
\begin{align}
    (\rho_2^{\otimes n},\sigma^{\otimes n})
    \succeq_{(\epsilon_n,0)}
    (\rho_1^{\otimes \lambda R_nn}\otimes \rho_1'^{\otimes (1-\lambda) R_nn},\sigma^{\otimes R_nn})
    ,
\end{align}
then this in fact exhibits \emph{strong} resonance, as shown in the right panel of \cref{fig:res_ex}. As we can see, the weak resonance condition determines the behaviour of rates for high errors. But when the strong resonance is present, it dominates over this and in fact determines the behaviour of rates at all error levels.

\begin{figure}
\centering
\includegraphics[width=\linewidth]{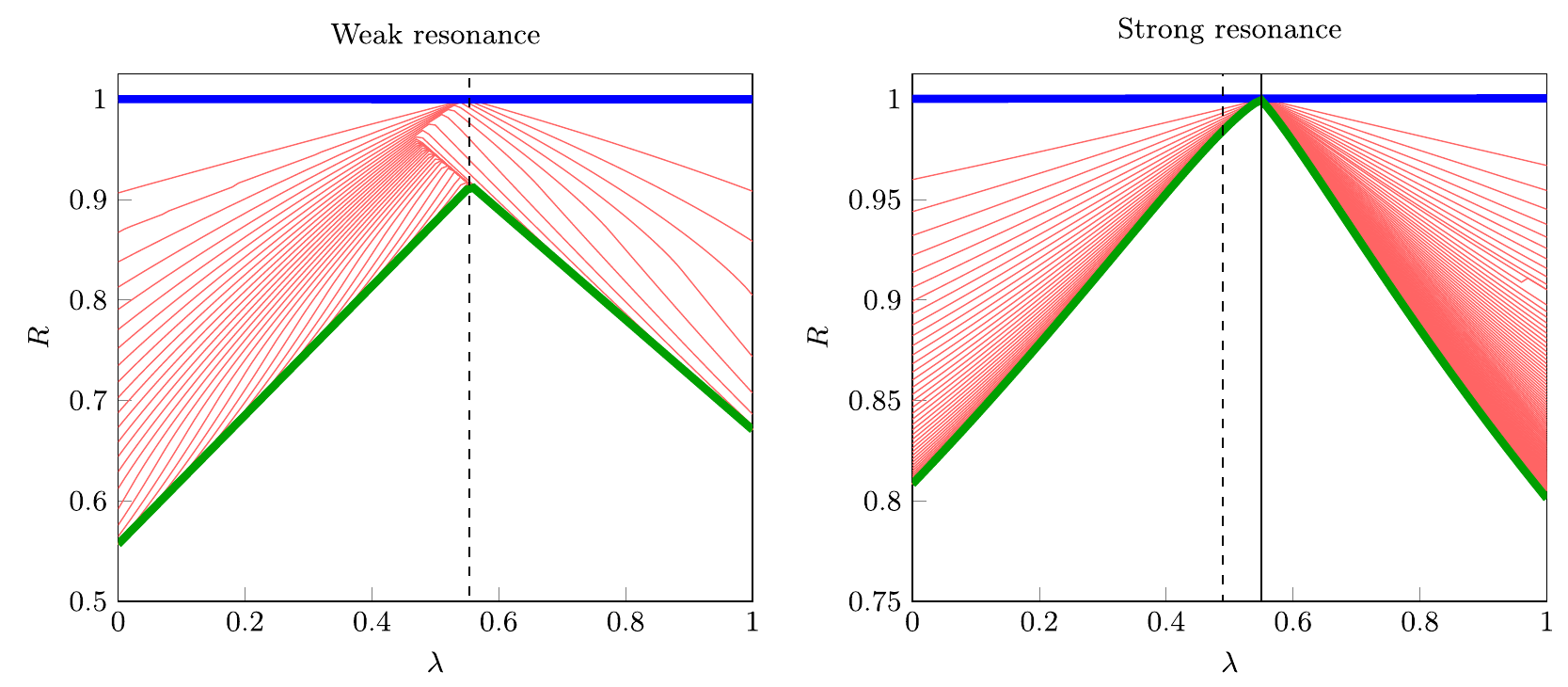}
\caption{\textbf{Examples of strong and weak resonance.} The upper blue lines correspond to the first-order rate, \emph{a la} \cref{thm:rate_first}. The lower green lines correspond to the zero-error transformation rates, \emph{a la} \cref{thm:rate_zero}. The internal red lines correspond to the optimal rates at an error level of $\exp(-\mu n)$, with each line correspondind to a different value of $\mu$. In the weak case $\mu\in \lbrace 0,0.05,\ldots,2\rbrace$ and in the strong case $\mu\in \lbrace 0,0.01,\ldots,1\rbrace$. Lastly, the vertical dashed black lines correspond to the mixture at which the weak resonance condition is satisfied, and the vertical solid black line to where the strong resonance condition is met.}
\label{fig:res_ex}
\end{figure}


\section{Asymptotic consistency}
\label{app:consistency}

In this appendix we will show rather satisfying `asymptotic consistencies' among our results, namely the hypothesis testing results of \cref{lem:ht_smalldev,lem:ht_moddev,lem:ht_largedev,lem:ht_extremedev}, the transformation rate results of \cref{\ratetheorems}, and the resonance phenomena considered in \cref{subsubsec:strong_res}. 

In \cref{sec:derive} we considered deriving the asymptotic behaviour of both hypothesis testing and transformation rates for dichotomies in several different error regimes (\cref{fig:sigmoid_ht,fig:sigmoid_rate}). Formally, these results must be separately proven in each of these distinct regimes. Eschewing rigour for the moment\footnote{May the {\scshape Lord} forgive us\ldots}, we might ask what happens if we take results from each error regime and na\"ively limit it into a neighbouring regime. By \emph{asymptotic consistency} we mean that this blasphemous and heretical procedure manages to reproduce the results of the rigorous treatments given in \cref{sec:derive}.

\subsection{Small and moderate deviation}

The small deviation error regime refers to errors $\epsilon\in(0,1)$ which are constant in $n$, and the moderate deviation regime concerns errors $\epsilon_n$ which are sub-exponentially approaching either 0 (low-error) or 1 (high-error). Here, we will consider starting with the small deviation results (\cref{lem:ht_smalldev} and \cref{thm:rate_smalldev}), and then applying expansions of this result around $\epsilon=0,1$, showing that this gives an entirely non-rigorous reproduction of the moderate deviation results (\cref{lem:ht_moddev} and \cref{thm:rate_moddev}).

As noted in Ref.~\cite{chubb2017moderate}, the inverse cdf of the standard Gaussian can be expanded for small positive $\epsilon$ as
\begin{align}
    \Phi^{-1}(\epsilon)\approx -\sqrt{\ln 1/\epsilon^2}
    \qquad\text{and}\qquad
    \Phi^{-1}(1-\epsilon)\approx +\sqrt{\ln1/\epsilon^2}.
\end{align}
Next, consider the small deviation expansion of the type-II hypothesis testing error given in \cref{lem:ht_smalldev}:
\begin{align}
    -\frac 1n\beta_\epsilon\rel{\rho^{\otimes }}{\sigma^{\otimes n}}
    \approx
    D\reli{} + \sqrt{\frac{V\reli{}}{n}}\cdot\Phi^{-1}(\epsilon).
\end{align}
If we simply substitute into these moderate error sequences $\epsilon_n:=\exp(-\lambda n^a)$ or $1-\epsilon_n$ for $\lambda>0,a\in(0,1)$, and use the above expansions, then we recover the moderate deviation expansion given in \cref{lem:ht_moddev}:
\begin{aligns}
    -\frac 1n\beta_{\epsilon_n}\rel{\rho^{\otimes }}{\sigma^{\otimes n}}
    &\approx
    D\reli{} - \sqrt{2V\reli{} \,\lambda n^{a-1}},\\
    -\frac 1n\beta_{1-\epsilon_n}\rel{\rho^{\otimes }}{\sigma^{\otimes n}}
    &\approx
    D\reli{} + \sqrt{2V\reli{} \,\lambda n^{a-1}}.
\end{aligns}

As for the dichotomy transformation rates, we need to consider expansions not just of the standard Gaussian, but also of the sesquinormal distribution considered in \cref{subsec:sesq}. In \cref{lem:sesqui} we saw that the sesquinormal distribution can be expressed in terms of the standard Gaussian distribution. Using this, we can expand the sesquinormal inverse cdf for small positive $\epsilon$ as
\begin{align}
    S^{-1}_{1/\xi}(\epsilon)\approx -\abs{1-\xi^{-1/2}}\sqrt{\ln1/\epsilon^2}
    \qquad\text{and}\qquad
    S^{-1}_{1/\xi}(1-\epsilon)\approx +\left[1+\xi^{-1/2}\right]\sqrt{\ln1/\epsilon^2}.
\end{align}
Similar to the case of hypothesis testing, if we take the small deviation dichotomy transformation rate (\cref{thm:rate_smalldev})
\begin{align}
    R_n^*(\epsilon)\approx \frac{D\reli{1} +\sqrt{V\reli{1}/n}\cdot S_{1/\xi}^{-1}(\epsilon) }{D\reli{2}},
\end{align}
and substitute moderate error rates, we reproduce the moderate deviation results (\cref{thm:rate_moddev}):
\begin{aligns}
    R_n^*\bigl(\exp(-\lambda n^a)\bigr)&\approx \frac{D\reli{1} -\abs{1-\xi^{-1/2}}\cdot \sqrt{2\lambda V\reli{1} n^{a-1}} }{D\reli{2}},\\
    R_n^*\bigl(1-\exp(-\lambda n^a)\bigr)&\approx \frac{D\reli{1} +\left[1+\xi^{-1/2}\right]\cdot \sqrt{2\lambda V\reli{1} n^{a-1}} }{D\reli{2}}.
\end{aligns}

\subsection{Large and moderate deviation}

The moderate deviation regime serves as a barrier between the small and large deviation regimes. As such, an alternate way of recovering the moderate deviation results is to consider the limit of large deviations. Specifically, errors which are exponentially approaching $0$ or $1$, but then consider the limit where we treat that exponent as arbitrarily small.

In the case of hypothesis testing, the large deviation results (\cref{lem:ht_largedev}) are
\begin{align}
    \frac 1n \gamma_{\lambda n} \rel{\rho^{\otimes n}}{\sigma^{\otimes n}} \to 
    \begin{dcases}
        \sup_{t<0}\Dm_t\reli{}+\frac{t}{1-t}\lambda & 
        \lambda\leq-D\rel{\sigma}{\rho},\\
        \inf_{0<t<1}\!\!\!-\Dp_t\reli{}-\frac{t}{1-t}\lambda & 
        -D(\sigma\|\rho)\leq \lambda\leq 0,\\
        \sup_{t>1} -\Dm_t\reli{}+\frac{t}{1-t}\lambda & 
        \lambda\geq 0.
    \end{dcases}
\end{align}
Substituting moderate errors into the large deviation result gives the expressions
\begin{aligns}
    \frac 1n \gamma_{-\lambda n^a} \rel{\rho^{\otimes n}}{\sigma^{\otimes n}}
    &\approx
    \inf_{0<t<1}-\Dp_t\reli{}+\frac{t}{1-t}\lambda n^{a-1},\\
    \frac 1n \gamma_{+\lambda n^a} \rel{\rho^{\otimes n}}{\sigma^{\otimes n}}
    &\approx
    \sup_{t>1}-\Dm_t\reli{}+\frac{t}{1-t}\lambda n^{a-1},
\end{aligns}
where $\lambda>0$ and $a\in(0,1)$. In both cases the optimisations approach $t\approx 1$ in this moderate regime, so we can expand the R\'enyi entropies using
\begin{align}
    \Dm_t\reli{}\approx \Dp_t\reli{}\approx D\reli{} + \frac{t-1}{2}V(\rho\|\sigma),
\end{align}
which gives
\begin{aligns}
    \frac 1n \gamma_{-\lambda n^a} \rel{\rho^{\otimes n}}{\sigma^{\otimes n}}
    &\approx
    \inf_{t<1}-D\reli{}+\frac{1-t}{2}V\reli{}+\frac{t}{1-t}\lambda n^{a-1},\\
    \frac 1n \gamma_{+\lambda n^a} \rel{\rho^{\otimes n}}{\sigma^{\otimes n}}
    &\approx
    \sup_{t>1}-D\reli{}+\frac{1-t}{2}V\reli{}+\frac{t}{1-t}\lambda n^{a-1}.
\end{aligns}
These optimisations can now be explicitly evaluated. To leading order, they give
\begin{align}
    \frac 1n \gamma_{\pm\lambda n^a} \rel{\rho^{\otimes n}}{\sigma^{\otimes n}}
    &\approx
    -D\reli{} \mp \sqrt{2V\reli{} \lambda n^{a-1}},
\end{align}
which is \cref{lem:ht_moddev}.

Next, we turn to the dichotomy transformation rates. We start with the high-error large deviation result \cref{thm:rate_largedev_hi},
\begin{align}
    R_n^*\bigl(1-\exp(-\lambda n)\bigr)
    &\approx \inf_{0<t_2<1} \inf_{t_1>1} \frac{\Dp_{t_1}\reli{1}+\left(\frac{t_1}{t_1-1}+\frac{t_2}{1-t_2} \right)\lambda}{D_{t_2}\reli{2}}
\end{align}
for $\lambda>0$. Substituting moderate errors, this becomes
\begin{align}
    R_n^*\bigl(1-\exp(-\lambda n^a)\bigr)
    &\approx \inf_{0<t_2<1} \inf_{t_1>1} \frac{\Dp_{t_1}\reli{1}+\left(\frac{t_1}{t_1-1}+\frac{t_2}{1-t_2} \right)\lambda n^{a-1}}{D_{t_2}\reli{2}}.
\end{align}
As with hypothesis testing, the optimisations will both approach $t_1,t_2\approx 1$, so we can expand the R\'enyi entropies around $t_1,t_2=1$,
\begin{align}
    R_n^*\bigl(1-\exp(-\lambda n^a)\bigr)
    &\approx \inf_{t_2<1} \inf_{t_1>1} \frac{D\reli{1}+\frac{t_1-1}{2}V\reli{1}+\left(\frac{t_1}{t_1-1}+\frac{t_2}{1-t_2} \right)\lambda n^{a-1}}{D\reli{2}+\frac{t_2-1}{2}V\reli{2}}\\
    &\approx \inf_{t_2<1} \inf_{t_1>1} \frac{
        D\reli{1}
        +\left[\frac{t_1-1}{2}V\reli{1}+\frac{\lambda n^{a-1}}{t_1-1}\right]
        +\left[-\frac{t_2-1}{2}\frac{D\reli{1}}{D\reli{2}}V\reli{2}+\frac{\lambda n^{a-1}}{1-t_2}\right]
    }{D\reli{2}}\\
    &\approx   \frac{
        D\reli{1}
        +\inf_{t_1>1}\left[\frac{t_1-1}{2}V\reli{1}+\frac{\lambda n^{a-1}}{t_1-1}\right]
        +\inf_{t_2<1}\left[\frac{1-t_2}{2}\frac{D\reli{1}}{D\reli{2}}V\reli{2}+\frac{\lambda n^{a-1}}{1-t_2}\right]
    }{D\reli{2}}\\
    &\approx   \frac{
        D\reli{1}
        +\sqrt{ 2V\reli{1}\cdot \lambda n^{a-1} }
        +\sqrt{ 2V\reli{2}\frac{D\reli{1}}{D\reli{2}}\cdot \lambda n^{a-1}}
    }{D\reli{2}}\\
    &\approx   \frac{
        D\reli{1}
        +[1+\xi^{-1/2}]\sqrt{ 2V\reli{1}\cdot \lambda n^{a-1} }
    }{D\reli{2}},
\end{align}
which is \cref{thm:rate_moddev}. For the low-error case, we can use the same arguments for $\overline r_2$/$\widecheck r_2$ and $r_3$. Specifically, for small negative $\mu$ we have
\begin{aligns}
    \overline r_2(-\mu)\approx \widecheck r_2(-\mu)&\approx
    \frac{
        D\reli{1}
        -[1-\xi^{-1/2}]\sqrt{ -2V\reli{1}\cdot \mu }
    }{D\reli{2}},\\
    r_3(\mu)&\approx\frac{
        D\reli{1}
        +[1-\xi^{-1/2}]\sqrt{ 2V\reli{1}\cdot \mu }
    }{D\reli{2}},
\end{aligns}
and thus
\begin{aligns}
    R_n^*\bigl(\exp(-\lambda n^a)\bigr)
    &\approx \min_{-\lambda n^{a-1}\leq \mu\leq \lambda n^{a-1}} 
    \begin{dcases}
        r_2(\mu) & \mu<0,\\
        r_3(\mu) & \mu>0,
    \end{dcases}\\
    &\approx \frac{
        D\reli{1}
        +\min\left\lbrace \xi^{-1/2}-1,1-\xi^{-1/2} \right\rbrace\sqrt{ 2V\reli{1}\cdot \lambda n^{a-1} }
    }{D\reli{2}}\\
    &\approx \frac{
        D\reli{1}
        -\abs{1-\xi^{-1/2}}\sqrt{ 2V\reli{1}\cdot \lambda n^{a-1} }
    }{D\reli{2}},
\end{aligns}
once again rederiving \cref{thm:rate_moddev}.

\subsection{Large and extreme deviation}

The other regime neighbouring large deviations is extreme deviations. Here, instead of taking the limit of an arbitrarily small error exponent, we will instead take the limit of an arbitrarily large error exponent, as a crude model of superexponential error. 

We start with hypothesis testing. For $\lambda >0$, \cref{lem:ht_largedev} gives that
\begin{align}
    \frac 1n \gamma_{\lambda n}\rel{\rho^{\otimes n}}{\sigma^{\otimes n}} \approx \sup_{t>1} -\Dm_t\reli{} + \frac{t}{1-t}\lambda.
\end{align}
As $\Dm_t\reli{}$ is monotonically increasing in $t$ and bounded, as we take $\lambda\to +\infty$ the optimising $t$ must also keep increasing. If we take $t\to \infty$, then this gives 
\begin{align}
    \frac 1n \gamma_{\lambda n}\rel{\rho^{\otimes n}}{\sigma^{\otimes n}} \approx -\Dm_{+\infty}\reli{} -\lambda.
\end{align}
Applying the same argument for $-\lambda$ gives
\begin{align}
    \frac 1n \gamma_{-\lambda n}\rel{\rho^{\otimes n}}{\sigma^{\otimes n}} \approx -\Dm_{-\infty}\reli{} +\lambda.
\end{align}
Both of these are precisely the extreme deviation results given in \cref{lem:ht_extremedev}.

Next we turn to the zero-error transformation rate of dichotomies. Recall that \cref{thm:rate_largedev_lo} gives that 
\begin{align}
    R_n^*\bigl(\exp(-\lambda n)\bigr) 
    \geq \min_{-\lambda\leq \mu\leq \lambda} \begin{dcases}
        r_1(\mu) & \mu<-D\rel{\sigma_1}{\rho_1},\\
        \widecheck r_2(\mu) & -D\rel{\sigma_1}{\rho_1}<\mu<0,\\
        r_3(\mu) & \mu>0,
    \end{dcases}
\end{align}
where
\begin{aligns}
    r_1(\mu):=&\sup_{t_2<0}\inf_{t_1<0}\frac{-\DL_{t_1}\reli{1}+\left(\frac{t_1}{t_1-1}-\frac{t_2}{t_2-1}\right)\mu}{-D_{t_2}\reli{2}},\\
    \widecheck r_2(\mu):=&\inf_{0<t_2<1}\sup_{0<t_1<1}\frac{\DL_{t_1}\reli{1}+\left(\frac{t_1}{1-t_1}-\frac{t_2}{1-t_2}\right)\mu}{D_{t_2}\reli{2}},\\
    r_3(\mu):=&\sup_{t_2>1}\inf_{t_1>1}\frac{\DL_{t_1}\reli{1}+\left(\frac{t_1}{t_1-1}-\frac{t_2}{t_2-1}\right)\mu}{D_{t_2}\reli{2}}.
\end{aligns}
If we take the limit of $\lambda\to\infty$, then the optimisation of $\mu$ becomes unconstrained, and $\mu$ can be seen as a Lagrange multiplier in the above optimisations. Ignoring issues of order-of-limits, this means that optimisations of $\mu$ can be converted into constrained optimisations, with the constraint being that $t_1=t_2$, specifically,
\begin{aligns}
    \inf_{\mu<-D\rel{\sigma_1}{\rho_1}} r_1(\mu)=&\inf_{t<0}\frac{\DL_{t}\reli{1}}{D_{t}\reli{2}},\\
    \inf_{-D\rel{\sigma_1}{\rho_1}<\mu<0}\widecheck r_2(\mu)=&\inf_{0<t<1}\frac{\DL_{t}\reli{1}}{D_{t}\reli{2}},\\
    \inf_{\mu>0}r_3(\mu)=&\inf_{t>1}\frac{\DL_{t}\reli{1}}{D_{t}\reli{2}}.
\end{aligns}
This means that 
\begin{align}
R_n^*(0)\gtrsim \inf_{t\in\mathbb R} \frac{\DL\reli{1}}{D\reli{2}}.
\end{align}
Following the discussion in \cref{thm:rate_zero} about pinching, this would extend to 
\begin{align}
    R_n^*(0)\gtrsim \ctc \max\left\lbrace
    \inf_{t\in\mathbb R} \frac{\DL\reli{1}}{D\reli{2}},
    \inf_{t\in\mathbb R} \frac{\DR\reli{1}}{D\reli{2}}
    \right\rbrace.
\end{align}

\subsection{Strong and weak resonance}

In \cref{subsubsec:strong_res} we discussed a strong resonance phenomenon which arises in the large/extreme deviation regimes, and complements the (weak) resonance discussed in Ref.~\cite{korzekwa2019avoiding}. We will now explain how the weak resonance condition can be seen as an edge case of the strong condition. The strong resonance condition is
\begin{align}
\operatorname*{arg\,min}_{\alpha\in\overline{\mathbb R}}\frac{\Dm_\alpha \reli{1}}{D_\alpha\reli{2}}&=1,
\end{align}
or, in other words,
\begin{align}
\min_{\alpha\in\overline{\mathbb R}}\frac{\Dm_\alpha \reli{1}}{D_\alpha\reli{2}}&=\frac{D\reli{1}}{D\reli 2}.
\end{align}
Weak resonance is a phenomenon that appears in the small and moderate deviation regimes. As we have shown before, these regimes can be seen as corresponding to values of $\alpha$ close to $1$. So, if we consider only such $\alpha$ values, and expand around $\alpha=1$, then this condition becomes
\begin{aligns}
\frac{D\reli{1}}{D\reli 2}
&=\min_{\alpha\in\overline{\mathbb R}}\frac{\Dm_\alpha \reli{1}}{D_\alpha\reli{2}}\\
&\approx \min_{\alpha}\frac{D\reli{1}+\frac{\alpha-1}{2}V\reli 1 + O((\alpha-1)^2)}{D\reli 2 +\frac{\alpha-1}{2}V\reli 2 + O((\alpha-1)^2)}\\
&\approx \frac{D\reli 1}{D\reli 2}\left[1+\frac{\alpha-1}{2}\left( \frac{V\reli{1}}{D\reli{1}}-\frac{V\reli{2}}{D\reli{2}} \right)+ O((\alpha-1)^2)\right],
\end{aligns}
which clearly then reduces to the weak resonance condition
\begin{align}
\frac{V\reli{1}}{D\reli{1}}&=\frac{V\reli{2}}{D\reli{2}}.
\end{align}


\section{Uniform hypothesis testing convergence}
\label{app:uni}

An important feature of \cref{lem:ht} is that it requires an ordering of the type-II errors simultaneously \emph{for all} values of $x$. If one were to na\"ively apply the hypothesis testing results in \cref{subsec:ht}, however, these would only provide pointwise convergence, instead of the uniform convergence such a statement would require. In this section we show that the hypothesis testing results of \cref{subsec:ht} can all be extended to uniform results as required essentially for free. This comes from the fact that the quantities being considered are monotonic, in such a way that rules out the pathologies necessary for non-uniform convergence. Specifically, we will use the following lemma:
\begin{lemma}[Prop~2.1 of \cite{resnick2007heavy}]
    \label{lem:monotone}
    Convergence of monotone functions on a compact set to a continuous function is uniform. In other words, if a sequence of functions $\lbrace f_n\rbrace_{n}$ from $[a,b]$ to $\mathbb R$ are all monotone, and pointwise converge to a continuous function $f$, then that convergence is in fact uniform.
\end{lemma}

\begin{lemma}[Uniform small deviation analysis of hypothesis testing]
    \label{lem:ht_smalldev_uni}
    For any $\delta>0$ there exists a finite $N(\rho,\sigma,\delta)$ such that both inequalities
    \begin{aligns}
        \abs{
        -\log \beta_{\epsilon}\rel{\rho^{\otimes n}}{\sigma^{\otimes n}}-nD\rel{\rho}{\sigma}-\sqrt{nV\rel{\rho}{\sigma}}\Phi^{-1}(\epsilon)
        } & \leq \delta \sqrt n,\\
        \abs{
        -\log\betaL_{\epsilon}\rel{\rho^{\otimes n}}{\sigma^{\otimes n}}-nD\rel{\rho}{\sigma}-\sqrt{nV\rel{\rho}{\sigma}}\Phi^{-1}(\epsilon)
        } & \leq \delta \sqrt n,
    \end{aligns}
    hold for all $n\geq N$ and $\epsilon\in[\delta,1-\delta]$.
\end{lemma}
\begin{proof}
    Start by defining
    \begin{aligns}
        f_n(x):=\frac{-\log\beta_{x}\rel{\rho^{\otimes n}}{\sigma^{\otimes n}}-nD\rel{\rho}{\sigma}}{\sqrt n}
        \qquad\text{and}\qquad
        \lefthat f_n(x):=\frac{-\log\betaL_{x}\rel{\rho^{\otimes n}}{\sigma^{\otimes n}}-nD\rel{\rho}{\sigma}}{\sqrt n},
    \end{aligns}
    and $f(x):=\sqrt{V\reli{}}\Phi^{-1}(x)$. \cref{lem:ht_smalldev} is equivalent to the statement that $f_n\to f$ and $\lefthat f_n\to f$ pointwise on $(0,1)$. However, because $\beta_x(\cdot\|\cdot)$ and $\betaL_x(\cdot\|\cdot)$ are monotone decreasing functions of $x$, we have that each $f_n$ and $\lefthat f_n$ is monotone increasing. Thus, if we constrain $x$ to some compact subset of $(0,1)$, say $x\in [\delta,1-\delta]$, then the uniformity of $f_n\to f$ and $\lefthat f_n\to f$ follows from \cref{lem:monotone}. This in turn implies that there exists an $\epsilon$-independent constant $N(\rho,\sigma,\delta)$ for which $\abs{f_n(\epsilon)-f(\epsilon)}\leq \delta$ and $|\lefthat f_n(\epsilon)-f(\epsilon)|\leq \delta$ hold for all $\epsilon\in[\delta,1-\delta]$ and $n\geq N$. Expanding this out gives the required inequalities.
\end{proof}

\begin{lemma}[Uniform large deviation analysis of hypothesis testing]
    \label{lem:ht_largedev_uni}
    For any constant $\delta>0$ there exists an $N(\rho,\sigma,\delta)$ such that the non-pinched/pinched log odds error per copy are bounded as
    \begin{aligns}
        \label{eq:ht_largedev_uni}
        \abs{\frac 1n {\gamma_{\lambda n}\rel{\rho^{\otimes n}}{\sigma^{\otimes n}}}-\Gamma_\lambda(\rho\|\sigma)}&\leq \delta,\\
        \abs{\frac 1n {\gammaL_{\lambda n}\rel{\rho^{\otimes n}}{\sigma^{\otimes n}}}-\GammaL_\lambda(\rho\|\sigma)}&\leq \delta,\\
        \abs{\frac 1n {\gammaR_{\lambda n}\rel{\rho^{\otimes n}}{\sigma^{\otimes n}}}-\GammaR_\lambda(\rho\|\sigma)}&\leq \delta,
    \end{aligns}
    for all $-1/\delta\leq \lambda\leq 1/\delta$ and $n\geq N$.
\end{lemma}
\begin{proof}
    This proof follows similarly to \cref{lem:ht_smalldev_uni}. Here we define
    \begin{aligns}
        f_n(x):=&\frac{\gamma_{nx}\rel{\rho^{\otimes n}}{\sigma^{\otimes n}}}n,\\
        \lefthat f_n(x):=&\frac{\gammaL_{nx}\rel{\rho^{\otimes n}}{\sigma^{\otimes n}}}n,\\
        \righthat f_n(x):=&\frac{\gammaR_{nx}\rel{\rho^{\otimes n}}{\sigma^{\otimes n}}}n,
    \end{aligns}
    as well as $f(x):=\Gamma_x\reli{}$, $\lefthat f(x):=\GammaL_x\reli{}$, and $\righthat f(x):=\GammaR_x\reli{}$. \cref{lem:ht_largedev} gives that $f_n\to f$, $\lefthat f_n\to \lefthat f$, and $\righthat f_n\to \righthat f$ pointwise on $\mathbb R$, and \cref{lem:monotone} allows us to make this uniform on $[-1/\delta,1/\delta]$. This uniform convergence in turn implies the existence of a finite $N(\rho,\sigma,\delta)$ such that $\abs{f_n(x)-f(x)}\leq \delta$, $\abs{\lefthat f_n(x)-\lefthat f(x)}\leq \delta$, and $\abs{\righthat f_n(x)-\righthat f(x)}\leq \delta$ for any $n\geq N$ and $x\in[-1/\delta,1/\delta]$. Expanding this gives the required inequalities.
\end{proof}

\begin{lemma}[Uniform moderate deviation analysis of hypothesis testing]
    \label{lem:ht_moddev_uni}
    For any constant $\delta>0$ and $a\in(0,1)$ there exists an $N(\rho,\sigma,\delta,a)$ such that the non-pinched/pinched log odds error per copy are bounded as
    \begin{aligns}
        \label{eq:ht_moddev_uni}
        \abs{\frac1n {\gamma_{\lambda n^a}\rel{\rho^{\otimes n}}{\sigma^{\otimes n}}}+D\reli{}+\mathrm{sgn}(\lambda)\cdot\sqrt{2V\reli{}\abs\lambda n^{a-1}}
        }\leq \delta\sqrt{n^{a-1}},\\
        \abs{\frac 1n{\gammaL_{\lambda n^a}\rel{\rho^{\otimes n}}{\sigma^{\otimes n}}}+D\reli{}+\mathrm{sgn}(\lambda)\cdot \sqrt{2V\reli{}\abs\lambda n^{a-1}}
        }\leq \delta\sqrt{n^{a-1}},
    \end{aligns}
    for all $-1/\delta\leq \lambda\leq 1/\delta$ and $n\geq N$.
\end{lemma}
\begin{proof}
    For this we define
    \begin{align}
        f_n(x):=\frac{\frac1n {\gamma_{x n^a}\rel{\rho^{\otimes n}}{\sigma^{\otimes n}}}+D\reli{}}{\sqrt{n^{a-1}}}
        \qquad\text{and}\qquad
        \lefthat f_n(x):=\frac{\frac1n {\gammaL_{x n^a}\rel{\rho^{\otimes n}}{\sigma^{\otimes n}}}+D\reli{}}{\sqrt{n^{a-1}}},
    \end{align}
    and $f(x):=-\mathrm{sgn}(x)\cdot\sqrt{2V\reli{} \abs{x}}$. \cref{lem:ht_moddev} is equivalent to the statement that $f_n,\lefthat f_n\to f$ pointwise on~$\mathbb R$. As $\gamma_x$ and $\gammaL_x$ are monotone increasing functions, we can apply \cref{lem:monotone} to upgrade this convergence to uniform, which gives that there exists a $N(\rho,\sigma,\delta,a)$ such that $\abs{f_n(x)-f(x)}\leq \delta$ and $\abs{\lefthat f_n(x)-f(x)}\leq \delta$ for any $n\geq N$ and $x\in[-1/\delta,1/\delta]$. Expanding these out gives the desired bounds.
\end{proof}

\end{document}